\documentclass[twoside,11pt]{article}

\usepackage{jmlr2e}
\usepackage[final]{microtype}

\newcommand{\T}{\scriptscriptstyle\text{T}}
\newtheorem{assumption}{Assumption}

\def\T{{\mathrm{\scriptscriptstyle T}}}
\def\v{{\varepsilon}}
\def\R{\mathbb{R}}
\def\L{\mathcal{L}}
\def\pr{\textnormal{pr}}

\def\E{\mathbb{E}}

\def\Y{Y}

\def\X{X}
\def\Z{Z}

\def\f{f}
\def\yita{\eta}
\def\E{E}
\def\U{U}
\def\D{D}
\def\d{d}
\def\tU{\widetilde{\U}}
\def\eps{\varepsilon}
\def\bata{\beta}
\def\Gama{\Gamma}
\def\gama{\gamma}
\def\L{\mathcal{L}}
\def\pr{\mathbb{P}}

\ShortHeadings{High-dimensional nonparametric additive IV regression}{Niu, Gu and Li}
\firstpageno{1}

\begin{document}

\title{Estimation and inference for high-dimensional nonparametric additive
  instrumental-variables regression}

\author{\name Ziang Niu \email ziangniu@sas.upenn.edu \\
       \addr Applied Mathematics and Computational Science Graduate Group\\
       University
  of Pennsylvania\\
       209 S 33rd St, Philadelphia\\
       Pennsylvania 19104, USA
       \AND
       \name Yuwen Gu \email yuwen.gu@uconn.edu \\
       \addr Department of Statistics\\ University of Connecticut\\
       215 Glenbrook Road, Storrs\\
       Connecticut 06269, USA
       \AND
       \name Wei Li\thanks{Corresponding author} \email weilistat@ruc.edu.cn \\
       \addr Center for Applied Statistics and School of Statistics\\
       Renmin University of China\\
       59 Zhongguancun Street\\
       Beijing 100872, China}

\editor{}


\maketitle

\begin{abstract}
The method of instrumental variables provides a fundamental and practical tool
  for causal inference in many empirical studies where unmeasured confounding
  between the treatments and the outcome is present. Modern data such as the
  genetical genomics data from these studies are often high-dimensional. The
  high-dimensional linear instrumental-variables regression has been considered
  in the literature due to its simplicity albeit a true nonlinear relationship
  may exist. We propose a more data-driven approach by considering the
  nonparametric additive models between the instruments and the treatments while
  keeping a linear model between the treatments and the outcome so that the
  coefficients therein can directly bear causal interpretation. We provide a
  two-stage framework for estimation and inference under this more general
  setup. The group lasso regularization is first employed to select optimal
  instruments from the high-dimensional additive models, and the outcome
  variable is then regressed on the fitted values from the additive models to
  identify and estimate important treatment effects. We provide non-asymptotic
  analysis of the estimation error of the proposed estimator. A debiasing
  procedure is further employed to yield valid inference. Extensive numerical
  experiments show that our method can rival or outperform existing approaches
  in the literature. We finally analyze the mouse obesity data and discuss new
  findings from our method.
\end{abstract}


\section{Introduction}

The method of instrumental variables has been extensively used in observational
studies to control for unmeasured confounding. If measurements of the treatments
and the outcome are confounded by unobserved variables, the casual effects of
the endogenous treatments can be identified if instrumental variables are
available. The instrumental variables need to be independent of the unmeasured
confounders and can only affect the outcome indirectly through the treatment
variables.
The method originates from early research on structural equations in
econometrics \citep{wright1928tariff,anderson1949estimation},
and has become increasingly popular in biostatistics and epidemiology with
Mendelian randomization as one of the most exciting recent applications
\citep{davey2003mendelian,didelez2007mendelian,
  lin2015regularization,burgess2017review}.
The general setting of this method involves modeling the instrument-treatment
and treatment-outcome relationships. The classical two-stage least squares approach
assumes linearity of both relationships and is widely used in the
low-dimensional setting. However, in many concurrent studies, data are often
high-dimensional. For example, gene expression data collected to identify
genetic variants associated with complex traits in genome-wide association
studies are usually high-dimensional. Many factors such as unmeasured
environmental conditions may induce spurious associations and distort the true
relationships between the gene expressions and the outcome of interest.
Nevertheless, the random assortment of the genes transferred from parents to
offspring resembles the use of randomization in experiments, and single
nucleotide polymorphisms (SNPs) therefore serve as natural instrumental
variables. The SNPs are also high-dimensional.

Recent developments of the instrumental-variables regression have introduced
regularization as a means to address the high dimensionality issue
\citep{gautier2011high,belloni2012sparse,fan2014endogeneity,
  cheng2015select,belloni2022high}.
For example, \cite{belloni2012sparse} use the lasso to select optimal
instruments from a large pool when the number of treatments remains fixed or
low-dimensional. Various procedures using different types of regularization have
been proposed thereafter. See \cite{hansen2014instrumental} and
\cite{fan2018nonparametric}, among others. Linear methods in which the
instruments and treatments are both high-dimensional have also been considered
\citep{lin2015regularization,zhu2018sparse,gold2020inference}.
\cite{lin2015regularization} demonstrate an application of the high-dimensional
linear instrumental-variables regression to genetic genomics. However, nonlinear
effects of the SNPs on the gene expressions are likely to exist as can be seen
from some recent articles that employ different kernel-based procedures to
capture possible nonlinear relationships
\citep{wang2015kernel,zhang2017use,zhan2017fast,zhao2019composite}. While these
methods keep fully nonparametric forms in linking the gene expressions and SNPs,
they are not very effective when applied to the high-dimensional regime.
\cite{zhu2018sparse} also considers the high-dimensional linear
instrumental-variables regression for peer effect estimation in econometrics.
Specifically, to analyze the effects of peers' output on a firm's production
output using panel data, the Research and Development expenditures of peer firms from a
previous period are treated as potential instrumental variables for the
endogenous treatments. Nevertheless, when the linear relationships are in
question, which likely are, the approach by \cite{zhu2018sparse} may lead to
unignorable bias.

Specification of the outcome equation, either in a parametric or nonparametric
form, is often based on expert knowledge or domain theory. The treatment model,
however, can be more data-driven and should involve nonlinear relationships when
possible to reduce bias \citep{newey1990efficient,fan2018nonparametric}. To
better approximate the treatments using optimal instruments, a general
nonparametric model can be beneficial. A substantial body of the recent
literature on high-dimensional nonparametric estimation focuses on the additive
models \citep[see, e.g.,][and references therein]{huang2010}. To this end,
we consider the high-dimensional additive models to capture the nonlinear
effects of the large number of instrumental variables on the treatments. We keep
the linearity assumption for the outcome model so that its coefficients directly
bear causal interpretations. We allow the dimensions of both the instrumental
variables and the treatments to be larger than the sample size. Similar to the
regularized two-stage framework for the high-dimensional linear
instrumental-variables regression, our proposed procedure consists of a first
stage in which we use the group lasso to select important instruments to best
predict the treatments, and a second stage in which we employ lasso to regress
the outcome on the first-stage predictions to perform variable selection and
estimation. We provide rigorous non-asymptotic analysis of the estimator and
further employ a debiasing procedure to establish valid inference.

In contrast to existing methods in the literature, the present work has the
following favorable features and makes several contributions to the
high-dimensional instrumental-variables regression. Firstly, the proposed
procedure is more data-adaptive which allows possible nonlinear
instrument-treatment relationships under high dimensions. A few recent articles
from the machine learning literature adopt deep learning to better estimate the
instrumental-treatment relationships
\citep{hartford2017deep,xu2020learning-arxiv}. However, these methods typically
require the dimensions of the instruments and the treatments both be smaller
than the sample size, and are not directly applicable to the setting considered
in the present article. Secondly, for the high-dimensional additive models in
the first stage, we develop a probabilistic bound for the estimation error of
the group lasso estimator. Compared with existing work in this area
\citep[e.g.,][]{huang2010}, we explicitly derive the non-asymptotic
probabilistic bounds of the estimation errors, which may be of independent
interest. Similar probabilistic bounds for the estimation error of the second
stage are also provided. Lastly, we provide statistical inference for the causal
parameters of interest by leveraging the debiasing procedures under high
dimensionality. It is recognized that inference for high-dimensional models is
typically difficult even when endogeneity is not present
\citep{javanmard2014confidence,zhang2014confidence,van2014asymptotically}.
\cite{gold2020inference} consider inference in the high-dimensional linear
instrumental-variables regression to deal with endogeneity. The present work
goes beyond that by establishing valid inference in the more flexible additive
models. Hence, our work enriches the literature on high-dimensional inference
that explicitly handles endogeneity.

\section{The sparse additive instrumental-variables model}\label{sec:proof}

Suppose we have $n$ independent and identically distributed observations from a
population of interest. Let $y_i$, $x_i$, and $z_i$ denote the $i$th
observations of the outcome, the $p\times 1$ vector of treatment variables, and
the $q\times 1$ vector of instrumental variables, respectively, where
$i=1,\ldots,n$.
Without loss of generality, assume the $x_i$'s and $y_i$'s are centered.
Consider the following joint modeling framework:
\begin{equation}\label{eq:joint-modeling}
  y_i=x_i^{\T}\beta+\eta_i,\quad
  x_{i\ell}=\sum_{j=1}^{q}f_{j\ell}(z_{ij})+\varepsilon_{i\ell}\quad
  (i=1,\ldots,n;\ \ell=1,\ldots,p),
\end{equation}
where $\eta_i\sim N(0,\sigma^{2}_{0})$ and $\varepsilon_{i\ell}\sim
N(0,\sigma^{2}_{\ell})$. Assume the treatment variables are endogeneous in the
sense that $E(\eta_i\mid x_{i\ell})\neq 0$, and the instrumental variables
satisfy $E(\eta_i\mid z_{ij})=E(\varepsilon_{i\ell}\mid z_{ij})=0$. The
$f_{j\ell}(\cdot)$'s are unknown smooth functions with compact support $[a,b]$,
where $a<b$. To ensure identifiability, assume $E\{f_{j\ell}(z_{ij})\}=0$ for
each $i$, $j$ and $\ell$. This is commonly assumed in the literature on additive
models. We also impose some smoothness conditions on the $f_{j\ell}(\cdot)$'s
and set the function class of consideration to a H\"older space $\mathcal{F}$.
\begin{assumption}\label{assu:smoothness}
  For $j=1,\ldots,q$ and $\ell=1,\ldots,p$, the function $f_{j\ell}$ belongs to
  $\mathcal{F}$, where
\begin{align*}
  \mathcal{F}=\biggl\{f\colon|f^{(k_0)}(z')-f^{(k_0)}(z)|
  \leq C|z'-z|^{\alpha_0},\ z,z'\in[a,b];\ \sup_{z\in[a,b]}|f(z)|\leq C_0\biggr\}
\end{align*}
with $d=k_0+\alpha_0>1.5$ and a universal constant $C_0>0$.
\end{assumption}

This assumption is common in nonparametric regression \citep[see,
e.g.,][]{fan2015functional,stone1985additive,huang2010}. Other similar
assumptions such as existence of high-order continuous derivatives are also
widely adopted \citep[Assumption A3,][]{horowitz2004nonparametric}.


In model~\eqref{eq:joint-modeling}, we are mainly interested in estimating the
average treatment effects, $\beta$. The linear setting has been investigated by
\cite{lin2015regularization} and \cite{zhu2018sparse}, where
$\sum_{j=1}^{q}f_{j\ell}(z_{ij})=z_{i}^{\T}\gamma_{\ell}$. Here, we relax the
linearity assumption and embrace the more general nonparametric additive form.
Define $J_{\ell}=\{j\colon{}f_{j\ell}\neq0\}$ for $\ell=1,\ldots,p$, and
$\L=\{\ell\colon\beta_{\ell}\neq0\}$. The sparsity assumption for the
high-dimensional additive model entails that $|J_{\ell}|\leq r$ for all $\ell$
and some positive integer $r$, where $|J_{\ell}|$ denotes the cardinality of the
set $J_{\ell}$. Similarly, we assume $s$-sparsity in the second stage with
$|\L|\leq s$, where
$s$ is a positive integer. To rewrite model~\eqref{eq:joint-modeling} in matrix
form, let $Y=(y_1,\ldots,y_n)^\T\in \R^n$,
$X=(x_1,\ldots,x_n)^\T\in\R^{n\times{}p}$,
$\eta=(\eta_1,\ldots,\eta_n)^\T\in\R^n$,
$F_j=(F_{j1},\ldots,F_{jp})\in\R^{n\times p}$ with
$F_{j\ell}=\{f_{j\ell}(z_{1j}),\ldots,f_{j\ell}(z_{nj})\}^{\T}\in\R^{n}$, and
$\varepsilon=(\varepsilon_1,\ldots,\varepsilon_p)\in\R^{n\times p}$ with
$\varepsilon_{\ell}=(\varepsilon_{1\ell},\ldots,\varepsilon_{n\ell})^\T\in
\R^{n}$ for $j=1,\ldots,q$ and $\ell=1,\ldots,p$. Then
model~\eqref{eq:joint-modeling} can be rewritten as
\begin{align*}
  Y=X\beta+\eta,\quad X=F+\varepsilon,
\end{align*}
where $F=\sum_{j=1}^q F_j$. To handle high dimensionality and endogeneity, we
allow $p,q\gg n$ and develop a two-stage penalized estimation framework.
In the first stage, we estimate each univariate function $f_{j\ell}$ via the
B-spline approximation. Let $a=\xi_{0}<\xi_{1}<\cdots<\xi_{K}<\xi_{K+1}=b$ be an
equal-distanced partition of $[a,b]$, where $K=[n^{\nu}]$ is a positive integer
for some $0<\nu<0.5$. Let $I_{Kt}=[\xi_{t},\xi_{t+1})$ for $t=0,\ldots,K-1$ and
$I_{KK}=[\xi_{K},\xi_{K+1}]$. Let $\{\phi_k(\cdot)\}_{k=1}^{m}$ be the
normalized B-splines such that each of them is (i) a polynomial function of
degree $L$ on $I_{Kt}$ for $t=0,\ldots,K$, and (ii) up to $(L-1)$ times
continuously differentiable on $[a,b]$, where $L>1$ is an integer and $m=K+L$.
A well-known property of such normalized basis functions is that
$0\leq\phi_k(z)\leq1$ and $\sum_{k=1}^{m}\phi_k(z)=1$ for any $z\in[a,b]$
\citep[][Theorem 4.20]{schumaker_2007}.
Given the $z_{ij}$'s, let
$\psi_{kj}(\cdot)=\phi_{k}(\cdot)-n^{-1}\sum_{i=1}^{n}\phi_{k}(z_{ij})$ for
$i=1,\ldots,n$ and $j=1,\ldots,q$. We will denote
$\psi_{k}(\cdot)=\psi_{kj}(\cdot)$ when no confusion arises. Now approximate the
additive functions using
$\{\psi_{k}(\cdot)\}_{k=1}^{m}$:
\begin{align}\label{eqn:approximation}
  x_{i\ell}\approx\sum_{j=1}^{q}\sum_{k=1}^{m}
  \gamma_{kj\ell}\psi_k(z_{ij})+\varepsilon_{i\ell}
  \quad (i=1,\ldots,n;\ \ell=1,\ldots,p).
\end{align}
Let $U=( U_{1},\ldots, U_{q})\in\R^{n\times qm}$, where for each $i$ and $j$,
$U_{j}=( U_{1j},\ldots, U_{nj})^{\T}\in \R^{n\times m}$ and
$U_{ij}=\{\psi_{1}(z_{ij}),\ldots,\psi_{m}(z_{ij})\}^{\T}$. Further define the
parameter matrix $\Gamma=(\gamma_{1},\ldots, \gamma_{p})\in\R^{qm\times p}$,
where for each $j$ and $\ell$,
$\gamma_{\ell}=(\gamma_{1\ell}^\T,\ldots,\gamma_{q\ell}^\T)^{\T}\in \R^{qm}$ and
$\gamma_{j\ell}=(\gamma_{1j\ell},\ldots,\gamma_{mj\ell})^\T\in\R^{m}$. The
approximation in~\eqref{eqn:approximation} becomes
$X\approx U\Gamma+\varepsilon$.

\begin{lemma}\label{lem:center_approximation}
  For each $f_{j\ell}\in\mathcal{F}$, there exists
  $\bar{\gamma}_{j\ell}=(\bar{\gamma}_{1j\ell},\ldots,\bar{\gamma}_{mj\ell})^{\T}$
  such that with probability at least $1-2(pqm)^{-2}$, the following holds
\begin{align*}
  \sup_{z\in[a,b]}\left|f_{j\ell}(z)-\tilde{f}_{nj\ell}(z)\right|
  \leq 2C_Lm^{-d}+2C_0\{\log(pqm)/n\}^{1/2},
\end{align*}
where $\tilde{f}_{nj\ell}(z)=\sum_{k=1}^{m}\bar{\gamma}_{kj\ell}\psi_k(z)$ and
$C_L$ is a universal constant depending only on $L$.
\end{lemma}
Lemma~\ref{lem:center_approximation} characterizes the approximation error of
the centered B-splines $\{\psi_k(\cdot)\}_{k=1}^{m}$ to each $f_{j\ell}(\cdot)$
with corresponding coefficients $\bar\gamma_{j\ell}$. Define
$\bar\gamma_{\ell}=(\bar\gamma_{1\ell}^\T,\ldots,\bar\gamma_{q\ell}^\T)^{\T}\in
\R^{qm}$.
Lemma~\ref{lem:center_approximation} implies that an intermediate step of
recovering $f_{j\ell}$ is to estimate $\bar{\gamma}_{\ell}$ by considering the
following penalized problem:
\begin{align}\label{eq:solution_first}
  \widehat{\gamma}_{\ell}=\mathop{\arg\min}\limits_{\gamma_{\ell}\in\mathbb{R}^{qm}}
  \bigg\{\frac{1}{2n}\|X_{\ell}- U\gamma_{\ell}\|^{2}_{2}
  +\lambda_{\ell}\sum_{j=1}^{q}\|\gamma_{j\ell}\|_{2}\bigg\}
  \quad (\ell=1,\ldots,p),
\end{align}
where $X_{\ell}$ is the $\ell$th column of $X$ and $\lambda_{\ell}\ge0$ is a
tuning parameter. This is a group lasso problem \citep{yuan2006model}
and is motivated by the fact that when $f_{j\ell}=0$, the vector
$\bar\gamma_{j\ell}=0$. After obtaining the predicted treatments
$\widehat{X}=(\widehat{X}_{1},\ldots,\widehat{X}_{p})$ with $\widehat{X}_{\ell}=
U\widehat{\gamma}_{\ell}$, we plug $\widehat{X}$ into the following lasso
problem to estimate $\beta$ in the second stage:
\begin{align}\label{eq:solution_second}
  \widehat{\beta}=\mathop{\arg\min}\limits_{\beta\in\mathbb{R}^{p}}
  \bigg\{\frac{1}{2n}\|Y-\widehat{X}\beta\|^{2}_{2}+\mu\|\beta\|_{1}\bigg\}
\end{align}
for some tuning parameter $\mu\ge0$. Estimation with high-dimensional predictors
has been a popular research topic in the past two decades. We note that the
above formulation is slightly different from the original lasso problem due to
the observed data being replaced by their estimations from the first stage. This
turns out to be more involved when showing the estimation consistency.

\section{Non-asymptotic analysis}\label{sec:nasymp-anlys}
We provide an estimation error bound for the first-stage group lasso problem.
Compared with existing results in this area
\citep{huang2010,ravikumar2009sparse}, we make contributions by explicitly
deriving the non-asymptotic probability bound. Based on this bound, we establish
a similar error bound for the parameter of interest in the second-stage lasso
problem.
Define $\Sigma_U=E(U^{\T}U/n)$. We make the following assumptions.

\begin{assumption}\label{assu:density_restric_iv}
  Each instrumental variable $z_{ij}$ has a continuous density on $[a,b]$ and
  the density is bounded away from zero and infinity.
\end{assumption}

\begin{assumption}\label{assu:re}
  There exists a universal constant $\rho\in(0,1)$ such that
  \begin{align*}
    \min\biggl\{\frac{\gamma^{\T}\Sigma_U\gamma}{\|\gamma_{J}\|_2^{2}}\colon
    |J|\leq r,\gamma\in\mathbb{R}^{qm}\backslash\{0\},\sum_{j\in J^{c}}
    \|\gamma_{j}\|_2\leq 3\sum_{j\in J}
    \|\gamma_{j}\|_2\biggr\}\geq\frac{\rho}{m},
  \end{align*}
  where $J\subset\{1,\ldots,q\}$ is an index set, $J^c$ denotes its complement,
  and $\gamma_J=(\gamma_{j\ell}\colon{}j\in J)^{\T}$.
\end{assumption}

Assumption \ref{assu:density_restric_iv} is rather standard in the
high-dimensional additive models \citep{huang2010,fan2018nonparametric}.
Assumption~\ref{assu:re} is often called the group restricted eigenvalue
condition \citep{lounici2011oracle,lv2018estimating,lu2020kernel}. This is a
natural extension of the restricted eigenvalue condition for the standard lasso
and Dantzig selector problems \citep{bickel2009simultaneous}. When the
instrumental variables are independent, $\Sigma_U$ is a block diagonal matrix
with diagonals $\Sigma_j=E( U_j^{\T} U_j/n)$. It is well known that
$\lambda_{\min}(\Sigma_j)\geq c_{*}/m$ \citep{lian2012semiparametric,huang2010}
when each instrumental variable is uniformly distributed, where $c_{*}>0$ is a
constant depending on the smoothness of the B-splines. Denote $a_n=o(b_n)$ if
$\lim_{n\rightarrow\infty}a_n/b_n=0$, $a_n=O(b_n)$ if there exists a positive
constant $C_1$ such that $\limsup_{n\rightarrow\infty}a_n/b_n\leq C_1$, and
$a_n=\Theta(b_n)$ if there exists positive constants $C_2$ and $C_3$ such that
$C_2\leq \liminf_{n\rightarrow\infty}a_n/b_n\leq
\limsup_{n\rightarrow\infty}a_n/b_n\leq C_3$.

\begin{theorem}\label{thm:estimation_first_stage}
  Suppose Assumptions \ref{assu:smoothness}--\ref{assu:re} hold. There exist
  positive constants $c_1$, $c_2$, and $c_3$ such that if
  \begin{align*}
    \lambda_{\max}=\max_{\ell}\lambda_{\ell}=\max
    \biggl[c_1\sigma_{\max}\Bigl\{\frac{\log(pqm)}{n}\Bigr\}^{1/2},
    c_2rm^{-(2d+1)/2}+c_3r\Bigl\{\frac{\log(pqm)}{mn}\Bigr\}^{1/2}\biggr],
  \end{align*}
  then for sufficiently large $n$, with probability at least $1-20(pqm)^{-1}$,
  the regularized estimator $\widehat\gamma_{\ell}$ in~\eqref{eq:solution_first}
  satisfies
  \begin{align*}
    \max_{\ell}\bigg\Vert\sum_{j=1}^{q}F_{j\ell}-U\widehat{\gamma}_{\ell}
    \bigg\Vert_2^2\leq \frac{50rmn\lambda_{\max}^2}{\rho},
    \quad \max_{\ell}\sum_{j=1}^{q}\|\widehat{\gamma}_{j\ell}
    -\bar{\gamma}_{j\ell}\|_2\leq\frac{32rm\lambda_{\max}}{\rho},
  \end{align*}
  where $\sigma_{\max}=\max_{\ell}\sigma_{\ell}$,
  $m=\Theta\{n^{1/(2d+1)}\}$, and $r^{2}=o[n /\{m^4\log(pqm)\}]$.
\end{theorem}

The performance of the group lasso estimator depends crucially on the eigen
behavior of the empirical covariance matrix $U^{\T}U/n$.
While it can be shown that the group restricted eigenvalue condition for the
empirical covariance matrix is satisfied under Assumption \ref{assu:re}, this
does come with a price on the rate of the sparsity level, that is,
$r^{2}=o[n/\{m^4\log(pqm)\}]$. Similar requirements can be found in Corollary 1
of \cite{raskutti2010restricted}. In view of the conditions of
Theorem~\ref{thm:estimation_first_stage}, it is easy to verify that
$\lambda^{2}_{\max}=O\{r^{2}\log(pqm)/n\}$.
Thus, to ensure the consistency of the average in-sample prediction, it is
required that $r^{3}=o\{n^{2d/(2d+1)}/\log(pqm)\}$, while for the estimation
consistency of the coefficients, $r^{4}=o\{n^{(2d-1)/(2d+1)}/\log(pqm)\}$ is
required. This is a more restrictive requirement than that in the standard
lasso, but it is expected due to the unspecified additive functional forms.
In contrast to Theorem 1 of \cite{huang2010} that only gives the convergence
rates, our result is completely non-asymptotic. Moreover, our
result allows the sparsity $r$ to grow with the sample size and dimension of the
data while this is not allowed in \cite{huang2010}. Guaranteed by Theorem 1 part
(i) of \cite{huang2010}, we can directly compare the estimation consistency
result obtained here with part (ii) of their theorem, and when $r$ is a fixed
number, it is easy to show they are the same. Other aligned results include
\cite{ravikumar2009sparse} and \cite{lu2020kernel}. \cite{ravikumar2009sparse}
obtain the out-of-sample risk consistency while both the explicit rate and the
in-sample error bound remain unclear. \cite{lu2020kernel} consider a
kernel-sieve hybrid estimator and obtain a similar non-asymptotic bound.

To provide an estimation error bound for $\widehat{\beta}$ defined
by~\eqref{eq:solution_second}, we make an extra assumption on the population
covariance matrix $\Sigma_{F}=E(F^\T F/n)$.

\begin{assumption}\label{assu:mini_Eigen}
There exists a  constant $\kappa>0$ such that
\begin{align*}
  \min\bigg\{\frac{\beta^{\T}\Sigma_F\beta}{\|\beta_{\mathcal{L}}\|_2^{2}}
  \colon|\mathcal{L}|\leq s,\beta\in\mathbb{R}^{p}\backslash\{0\},
  \sum_{\ell\in \mathcal{L}^{c}}|\beta_{\ell}|\leq
  3\sum_{\ell\in \mathcal{L}}|\beta_{\ell}|\bigg\}\geq\kappa,
\end{align*}
where $\mathcal{L}\subset\{1,\ldots,p\}$ is an index set, $\mathcal{L}^{c}$
denotes its complement, and
$\beta_{\mathcal{L}}=(\beta_{\ell}\colon\ell\in\mathcal{L})^{\T}$.
\end{assumption}

Assumption~\ref{assu:mini_Eigen} is the restricted eigenvalue condition on
$\Sigma_F$ and is useful for deriving the error bounds in the second-stage lasso
problem. This assumption imposes some requirements on the covariance structures
of the treatment matrix $X$ and the noise variables $\varepsilon$. For example,
when $\mathrm{cov}(\varepsilon_{\ell},\varepsilon_{\ell'})=0$ for
$\ell\neq\ell'$ and the minimum eigenvalue of $\Sigma_{X}=E(X^{\T}X/n)$ is
larger than $\max_{\ell}\sigma_\ell$, the above condition immediately holds.
To provide an estimation error bound for $\widehat{\beta}$, we restrict the
parameter space of consideration to an $L_{1}$-ball $\|\beta\|_{1}\leq B$ for
some $B>0$. Similar technique has been frequently used in the literature
\citep[see, e.g.,][]{lin2015regularization}. This restriction can be further
relaxed to the $L_{\infty}$-ball $\|\beta\|_{\infty}\leq{}B$, but it may lead
to a sacrifice of the convergence rate.

\begin{theorem}\label{thm:second_stage_consistency}
  Suppose Assumptions \ref{assu:smoothness}--\ref{assu:mini_Eigen} hold. Let the
  regularization parameter $\lambda_{\max}$ be chosen as in
  Theorem~\ref{thm:estimation_first_stage}. Further assume $\lambda_{\max}$
  satisfies $560C_0\lambda_{\max}(2rm/\rho)^{1/2}\leq\kappa^{2}/(4rs)$. If we
  choose the second-stage regularization parameter as
  \begin{align*}
    \mu=2r\lambda_{\max}(7\sigma_0+8\sqrt{5}B\sigma_{\max}+30B)(2m/\rho)^{1/2},
  \end{align*}
  then with probability at least $1-234(pqm)^{-1}$, the estimator
  $\widehat\beta$ in~\eqref{eq:solution_second} satisfies
  \begin{align*}
    \|\widehat{\beta}-\beta\|_{1}\leq\frac{64}{\kappa^{2}}s\mu,
    \quad \|\widehat{X}(\widehat{\beta}-\beta)\|^{2}_{2}
    \leq\frac{64}{\kappa^{2}}ns\mu^{2}.
  \end{align*}
\end{theorem}

As far as we know, Theorem \ref{thm:second_stage_consistency} is the first to
present a non-asymptotic estimation error bound for the two-stage additive
model. Straightforward analysis shows that consistency is guaranteed if we take
$\mu^{2}=O\{r^{4}\log(pqm)/n^{2d/(2d+1)}\}$ and
$s^{2}r^{5}=o\{n^{2d/(2d+1)}/\log(pqm)\}$. When $r$ is fixed, we have
$s^{2}=o[n/\{m\log(pqm)\}]$. This almost recovers the sparsity in the
classical lasso setting when $d$ is large enough.
Since the two-stage linear model considered by \cite{lin2015regularization} and
\cite{zhu2018sparse} is a special case of our setting, the empirical results in
Section~\ref{sec:simulation} demonstrate similar performance between the two
models when the true relationship in both stages is linear.

\section{Inference}\label{sec:inference}

We develop a method to draw inference on each component of the outcome
regression parameters $\beta$. When endogeneity is absent, various methods via
debiasing the penalized estimator have been proposed to conduct valid inference
for high-dimensional models \citep{zhang2014confidence,van2014asymptotically,
  javanmard2014confidence}. A recent article by \cite{gold2020inference}
considers the endogeneity issue and adapts the parametric one-step update
procedure to construct confidence intervals for parameters in the
high-dimensional two-stage linear model. We extend the approach therein to draw
inference under the more general setting entailed by
model~\eqref{eq:joint-modeling}.

A key step in deriving the debiased estimator is to utilize the conditional
moment restriction $E(\eta\mid Z)=0$. This equation entails the orthogonality
condition $E(\Gamma^{\T} U^\T\eta)=0$. Here $\Gamma$ is any fixed coefficient
matrix and will later be set to
$\overline\Gamma=(\bar{\gamma}_{1},\ldots,\bar{\gamma}_{p})^{\T}$. Let
$D=U\overline\Gamma$ and $d_i$ be the $i$th row of $ D$. By
Lemma~\ref{lem:center_approximation} and
Theorem~\ref{thm:estimation_first_stage}, the estimate
$\widehat{D}=U\widehat\Gamma$ is a good estimate of $D$, where
$\widehat\Gamma=(\widehat\gamma_1,\ldots,\widehat\gamma_p)$. The orthogonality
condition then implies that the empirical counterpart based on the estimate
$\widehat{D}$ is approximately equal to zero, that is,
\begin{align*}
  E_n\{h(y_i,x_i,\widehat{d}_i;\beta)\}\coloneqq
  -\widehat{D}^\T(Y-X\beta)/n\approx0,
\end{align*}
where $E_n(w_i)=n^{-1}\sum_{i=1}^{n}w_i$ is the expectation with respect to the
empirical measure. The one-step update to the second-stage estimator
$\widehat{\beta}$ can thus be written as
\begin{align*}
  \widetilde{\beta}=\widehat{\beta}-\widehat{\Omega}
  E_n\{h(y_i,x_i,\widehat{ d}_i;\widehat\beta)\}
  =\widehat{\beta}+\widehat{\Omega}\widehat{D}^\T
  (Y-X\widehat{\beta})/n,
\end{align*}
where $\widehat{\Omega}$ is some estimate of $\Omega=\Sigma_F^{-1}$. Similar to
\cite{gold2020inference}, we construct the estimator $\widehat{\Omega}$ from a
modification of the constrained $L_1$-minimization approach to sparse precision
matrix estimation proposed by \cite{cai2011constrained}. The rows
$\widehat{\theta}_\ell$ of the estimator $\widehat{\Omega}$ are obtained as
solutions to the following program:
\begin{align}\label{eq:opt_precision}
  \min_{\theta_\ell\in\mathbb{R}^p}\|\theta_\ell\|_1,\enskip
  \text{subject to}\enskip\|\widehat{\Sigma}_F\theta_\ell-e_\ell\|_{\infty}
  \leq\upsilon\quad (\ell=1,\ldots,p),
\end{align}
where $e_\ell$ is the $\ell$th canonical basis vector in $\R^{p}$ and
$\upsilon>0$ is the tolerance parameter. The following lemma characterizes a
decomposition of the one-step estimator $\widetilde{\beta}$.

\begin{lemma}\label{lem:decom}
  The one-step estimator $\widetilde{\beta}=\widehat{\beta}+\widehat{\Omega}
  \widehat{D}^{\T}(Y-X\widehat{\beta})/n$ satisfies
  $\sqrt{n}(\widetilde{\beta}-\beta)=\Omega D^\T\eta/\sqrt{n}+\sum_{k=1}^4R_k$,
  where
  \begin{align*}
    R_1 & =(\widehat{\Omega}-\Omega) D^\T\eta/\sqrt{n},
    & R_2 & =\widehat{\Omega}(\widehat{ D}- D)^\T\eta/\sqrt{n},\\
    R_3 & =\widehat{\Omega}\widehat{D}^\T(X-\widehat{D})
          (\beta-\widehat{\beta})/\sqrt{n},
    & R_4 & =\sqrt{n}(\widehat{\Omega}\widehat{\Sigma}_{F}-I)
            (\beta-\widehat{\beta}).
  \end{align*}
\end{lemma}

Lemma~\ref{lem:decom} implies that to establish the asymptotic normality of each
component $\widetilde\beta_{\ell}$, it suffices to make sure each remainder term
$\|R_k\|_{\infty}=o_p(1)$, $k=1,2,3,4$. The $L_1$-bound on
$\widehat{\theta}_\ell-\theta_\ell$, which is needed for controlling the
remainder terms, becomes manageable when the following restriction on the
population precision matrix is imposed.

\begin{assumption}\label{assu:precision_mat}
  There exist some positive number $m_{\Omega}$, tolerance $b\in[0,1)$, and
  generalized sparsity level $s_{\Omega}$ such that the population precision
  matrix $\Omega\in\mathcal{U}(m_{\Omega},b,s_{\Omega})$, where
  \begin{align*}
    \mathcal{U}(m_{\Omega},b,s_{\Omega})
    =\bigg\{\Omega=(\theta_{\ell\ell'})_{\ell,\ell'=1}^p\succ0
    \colon\|\Omega\|_{1}\leq m_{\Omega};\max_{\ell\in\{1,\ldots,p\}}
    \sum_{\ell'=1}^{p}|\theta_{\ell\ell'}|^{b}\leq s_{\Omega}\bigg\}
  \end{align*}
  and $\|\Omega\|_1=\sup_{\ell}\|\theta_\ell\|_1$.
\end{assumption}

We assume the event that the rows $\theta_\ell$ of $\Omega$ are feasible for the
minimization program~\eqref{eq:opt_precision} has probability approaching one,
that is, $\pr(\|\Omega\widehat{\Sigma}_F-I\|_{\infty}\leq\upsilon)\rightarrow1$
as $n\rightarrow\infty$. The validity of such a requirement mainly depends on
the choice of the tolerance $\upsilon$. We give a theoretical choice of
$\upsilon$ in Theorem~\ref{thm:asym_normality} and provide some rate conditions
under which $\sqrt{n}(\widetilde{\beta}_\ell-\beta_\ell)/\omega_{\ell}$ is
asymptotically normal, where $\omega_\ell^{2}=\sigma_0^{2}\theta_{\ell\ell}$ and
$\theta_{\ell\ell}$ is the $\ell$th diagonal entry of $\Omega$.

\begin{theorem}\label{thm:asym_normality}
  Suppose Assumptions \ref{assu:smoothness}--\ref{assu:precision_mat} and the
  conditions of Theorems
  \ref{thm:estimation_first_stage}--\ref{thm:second_stage_consistency} hold.
  Assume each element $\theta_{\ell\ell}>\vartheta$ for some universal constant
  $\vartheta>0$ and let
  $\upsilon=36C_0m_{\Omega}\lambda_{\max}r(2rm/\rho)^{1/2}$. If the following
  rate conditions hold
\begin{align*}
    r^{(7-5b)/2}\bigg\{\frac{\log(pqm)}{n}\bigg\}^{(1-b)/2}
    & \Big[m^{1/2}+\big\{\log(pqm)\big\}^{1/2}\Big]=o(1),\\
    r^{2}\big(m^{3}/n\big)^{1/2}\log(pqm)=o(1),
    & \quad r^{7/2}s\big(m^{2}/n\big)^{1/2}\log(pqm)=o(1),
\end{align*}
then $\|R_k\|_{\infty}=o_{p}(1)$ as $n\rightarrow\infty$, where $k=1,2,3,4$.
Moreover, $\sqrt{n}(\widetilde{\beta}_\ell-\beta_\ell)/\omega_\ell$ converges in
distribution to the standard normal distribution.
\end{theorem}

The rates in Theorem~\ref{thm:asym_normality} are required to make the remainder
terms negligible. When the sparsity parameter $r$ in the first stage is fixed,
the requirement for the sparsity level of the second stage is
$s=o(n^{1/2}m/\log(pqm))$. This is almost the same as the requirement for the
debiased lasso: $s=o(n^{1/2}/\log(p))$ \citep[see, e.g.,][]
{javanmard2014confidence, van2014asymptotically}. By
Theorem~\ref{thm:asym_normality}, we can construct a confidence interval for
$\beta_{\ell}$ if a consistent estimator $\widehat{\omega}_{\ell}$ is available.
Given $\ell\in\{1,\ldots,p\}$ and $\alpha\in(0,1)$, an asymptotic
$100(1-\alpha)\%$ confidence interval for $\beta_{\ell}$ is
$[\widetilde{\beta}_\ell-z_{\alpha}\widehat{\omega}_\ell/\sqrt{n},
\widetilde{\beta}_\ell+z_\alpha\widehat{\omega}_\ell/\sqrt{n}]$, where
$z_\alpha=\Phi^{-1}(1-\alpha/2)$ and $\Phi$ is the cumulative distribution
function of the standard normal distribution. The following theorem provides a
way to construct a consistent estimator $\widehat{\omega}_\ell$.
\begin{theorem}\label{lem:consis_omega}
  Suppose the conditions of Theorems
  \ref{thm:estimation_first_stage}--\ref{thm:asym_normality} hold. Define
\begin{align*}
  \widehat{\omega}_\ell=\widehat{\sigma}_0\big(\widehat{\theta}_\ell^{\T}
  \widehat{\Gamma}^{\T}U^{\T}U\widehat{\Gamma}
  \widehat{\theta}_\ell/n\big)^{1/2},\enskip
  \widehat{\sigma}_0=n^{-1/2}\|Y-X\widehat{\beta}\|_2.
\end{align*}
Then $\widehat\omega_{\ell}$ is a consistent estimator of $\omega_\ell$ for each
$\ell\in\{1,\ldots,p\}$.
\end{theorem}

\section{Simulation}\label{sec:simulation}


We conduct simulation studies to evaluate the finite-sample performance of the
proposed methods. Our objective is to test both the estimation and inferential
procedures under a variety of experiments. For estimation purpose, we compare
our procedure with the classical one-stage  penalized least squares (PLS)
and the two-stage least-square method with lasso penalty (2SLS-L). For
inferential purpose, we compare our method with the updated two-stage lasso estimator (Up-2SLS-L)
proposed by
\cite{gold2020inference}.

We first investigate the estimation performance under two design settings where
endogeneous treatments are generated from linear and nonlinear models,
respectively. In both settings, we take $p=q=600$ and vary $n$ from $100$ to
$2100$. Experiments for other values of $(p,q,n)$ are provided in the
Appendix. We generate the instrumental variables $z_i$ of the
$i$th observation from the multivariate normal distribution with zero mean and
covariance matrix $\Sigma_{Z}=\{(\Sigma_Z)_{jj'}\}$, where
$(\Sigma_Z)_{jj'}=0.2^{|j-j'|}$ for $j,j'=1,\ldots,q$. To generate the noise
vector $(\eta_i, \v_{i1},\ldots, \v_{ip})^{\T}$, we sample from another normal
distribution with zero mean and covariance matrix $\Sigma=(\Sigma_{\ell\ell'})$,
where $\Sigma_{\ell\ell'} =0.2^{|\ell-\ell'|}$ for $\ell,\ell'=2,\ldots,p+1$ and
$\Sigma_{11}=1$. We also set $\Sigma_{1\ell}$ for $\ell=2,\ldots,6$ and five
other random selected entries from the first column of $\Sigma$ to $0.3$. All
other entries are set to zero. We finally set $\Sigma_{1\ell}=\Sigma_{\ell1}$
for $\ell=2,\ldots,p+1$ to make $\Sigma$ symmetric. Note that the nonzero
$\Sigma_{1\ell}$'s induce endogeneity in the data. In the linear setting, we
generate the treatment variables $x_i$ according to $x_i=\Gamma^\T z_i+\v_i$,
where $\Gamma=(\gamma_{j\ell})\in\R^{q\times p}$ is a sparse coefficient matrix
obtained by sampling $r=5$ nonzero entries of each column from the uniform
distribution $U(0.75,1)$. In the nonlinear setting, we generate $x_i$ from the
following equation:
\begin{align*}
  x_{i\ell} = \gamma_{1\ell}z_{i1}^{2} + \gamma_{2\ell}z_{i2}
  + \gamma_{3\ell}z_{i3}^{2}+\gamma_{4\ell}\sin(\pi z_{i4})
  +\gamma_{5\ell}z_{i5}^{2}+\epsilon_{i\ell},\quad (\ell=1,\ldots,p).
\end{align*}
The sampling strategy for the coefficients $\gamma_{j\ell}$'s are the same as
that in the linear setting. We finally generate the outcome response according
to $y_i=x_i^\T\beta+\eta_i$, where the coefficient vector $\beta$ is generated
by sampling $s=5$ nonzero components from the uniform distribution over two
disjoint intervals $U\{(-1,-0.75)\cup(0.75,1)\}$.

In all simulations, we use the Bayesian information criterion to select the
first-stage tuning parameters $m$ and $\lambda_\ell$, and the five-fold cross
validation to select the second-stage regularization parameter $\mu$. We report
the $L_1$ error $\Vert\widehat\beta-\beta\Vert_1$ of each method based on one
hundred replications. The results are summarized in
Table~\ref{tab:estimationloss}. It can be seen that when the
instrument-treatment relationship is linear, our method is as good as 2SLS-L. As
the sample size increases, the $L_1$ errors of both 2SLS-L and our method decrease,
whereas that of PLS increases due to ignorance of endogeneity. When the
instrument-treatment relationship is nonlinear, the performance of our method is
similar to that in the linear setting. The $L_1$ error of our method is the
smallest in almost all settings and exhibits a decreasing trend as the sample
size increases. In contrast, neither PLS nor 2SLS-L shares such trend in their
performance under these settings. We can clearly see that the 2SLS-L method induces a large bias since it suffers losses from model misspecification. It is also implied that model misspecification
may have a heavier impact on 2SLS-L than on PLS. This may be related to the specific setting of 
our simulations.
Overall, the results from
Table~\ref{tab:estimationloss} demonstrate the estimation consistency of our
estimator.

\begin{table}[h]
\centering
\caption{$L_1$ estimation loss of each method averaged over 100 replications with standard deviation shown in parentheses for $p=600$.}
\label{tab:estimationloss}
\begin{tabular}{ccccccc}
\toprule  
\multirow{2}{*}{Sample} & \multicolumn{3}{c}{Linear} & \multicolumn{3}{c}{Nonlinear}\cr\cmidrule(lr){2-4}\cmidrule(lr){5-7}
& Our method & 2SLS-L & PLS & Our method & 2SLS-L & PLS\\
\midrule  
 100 & 1.26 (0.53)& 2.52 (1.19) &0.86 (0.22) & 2.96 (1.41)& 1.19 (0.38) &1.25 (0.38)\\
 \addlinespace[0.7mm]
300 &0.50 (0.23)& 0.59 (0.30) &0.51 (0.16) & 0.74 (0.28)& 1.27 (0.79) & 0.79 (0.27)\\
\addlinespace[0.7mm]
600 &0.34 (0.15) & 0.27 (0.11)& 0.46 (0.17)& 0.43 (0.17) & 1.73 (1.26)& 0.89 (0.27)\\
\addlinespace[0.7mm]
900& 0.25 (0.10)& 0.23 (0.08)& 0.52 (0.14)& 0.32 (0.13)& 2.90 (1.98)& 1.10 (0.23)\\
\addlinespace[0.7mm]
1200&0.19 (0.09) & 0.19 (0.08)& 0.61 (0.18)& 0.27 (0.13)& 3.35 (2.36)& 1.18 (0.18)\\
\addlinespace[0.7mm]
1500& 0.17 (0.08)& 0.18 (0.08)& 0.70 (0.25)& 0.24 (0.10)& 4.17 (3.29)& 1.27 (0.16)\\
\addlinespace[0.7mm]
1800& 0.16 (0.07)& 0.17 (0.08)& 0.81 (0.29)& 0.21 (0.09)& 4.74 (4.03)& 1.34 (0.17)\\
\addlinespace[0.7mm]
2100& 0.14 (0.05) & 0.16 (0.07)& 1.09 (0.42)& 0.21 (0.11)& 5.61 (4.38)& 1.43 (0.16)\\

\bottomrule 
\end{tabular}
\end{table}

Next, we evaluate the performance of the inferential procedure proposed in
Section~\ref{sec:inference} and consider a more challenging nonlinear setting:
\begin{align*}
  x_{i\ell}=-8\gamma_{1\ell}z_{i1}^{2}+\gamma_{2\ell}\sin(\pi z_{i2})
  +2\gamma_{3\ell}\log(z_{i3}^{2})+\gamma_{4\ell}(10z_{i4})^{3}
  +\gamma_{5\ell}z_{i5}^{2}+\epsilon_{i\ell},~~ (\ell=1,\ldots,p).
\end{align*}
We calculate the 95\% confidence interval of each element of $\beta$ based on
our method and the Up-2SLS-L based inferential procedure. The coverage probabilities
and interval lengths averaged over all elements are shown in
Table~\ref{tab:asym_normality}. Our confidence intervals have coverage
probabilities close to the nominal level of 0.95. In contrast, the
Up-2SLS-L based confidence intervals have coverage probabilities well
below the nominal level. Moreover, their intervals are much wider than ours
under all settings. These are expected as the Up-2SLS-L based inferential procedure
is proposed under the linear setting and may not
perform well under the nonlinear setting. The results from
Table~\ref{tab:asym_normality} validate our theory in
Section~\ref{sec:inference}.

\begin{table}[htbp!]
  \centering
  \caption{Coverage probabilities and lengths of the 95\% confidence intervals
    by our method and the Up-2SLS-L based inferential procedure. Numbers shown are
    multiplied by one hundred.}
  \begin{tabular}{cccccc}
    \toprule
    \multirow{2}{*}{Dimension} & Sample
    & \multicolumn{2}{c}{Our method}
    & \multicolumn{2}{c}{Up-2SLS-L}
      \cr\cmidrule(lr){3-4}\cmidrule(lr){5-6}
    & size & Coverage & Length & Coverage  & Length\\
    \midrule
    250&200 & 92.0& 0.396 & 87.3& 2.663 \\
    400& 300 & 93.5& 0.264& 89.0& 1.585\\
    500&400 &94.0& 0.245& 87.1& 2.102\\
    600&500& 93.7& 0.140& 87.8& 2.614\\
    \bottomrule
  \end{tabular}
  \label{tab:asym_normality}
\end{table}

\section{Real data analysis}

We further illustrate our proposed method by analyzing the mouse obesity data
described by \cite{wang2006genetic}. This data set consists of genotype, gene
expression and clinical information about the F2 intercross mice. The genotypes
are characterized by the SNPs at an average density of 1.5 cM across the whole
genome and the gene expressions in the liver tissues of the mice are profiled by
microarrays. We are interested in the causal effect of the gene expressions on
the body weights of the mice. We
consider the data collected from $n=287$ (144 female and 143 male) mice with
$q=1250$ SNPs and $p=2816$ genes.
Since there are only three genotypes, a sparse high-dimensional linear model
between the SNPs and the gene expressions is often postulated. However, some
recent articles suspect that nonlinear effects may exist in the data
\citep{li2014model,zhang2017use,guha2020integrated}. This has motivated us to
consider the nonlinear setting. Before applying our method to the data, we
adjust for confounding induced by the sex of the mice. We first regress the body
weight on the sex and subtract the estimated effect from the body weight. We
then apply our proposed estimation method to the data with the adjusted body
weight as the outcome. We use the five-fold cross validation to select the
tuning parameters. The resultant model includes 28 genes. To increase the
stability and interpretability of our analysis, we apply the stability selection
approach proposed by \cite{meinshausen2010stability} to compute the selection
probability of each gene over one hundred subsamples of size $\lfloor
n/2\rfloor$ for a sequence of values of the tuning parameter $\mu$. We set the threshold probability to 0.5 so that the number of genes selected is reasonable. The
results are shown in Table~\ref{tab:selection_prob}.

\begin{table}[htbp!]
  \centering
  \caption{Stability selection: genes with selection probability exceeding 0.5.}
  \begin{threeparttable}[b]
    \begin{tabular}{lrlr}
      \toprule
      Gene Name& Selection Probability & Gene Name & Selection Probability \\
      \midrule
      Vwf & 0.77 & Krtap19-2 & 0.59\\
      Akap12 & 0.63 & Tmem184c & 0.74 \\
      2010002N04Rik & 0.84 & Igfbp2 & 0.51\\
      Slc43a1 &0.76 & Gstm2 & 0.91\\
      Ccnl2 & 0.54 & D14Abb1e &0.52 \\
      B4galnt4 &0.71 & & \\
      \bottomrule
    \end{tabular}
  \end{threeparttable}
  \label{tab:selection_prob}
\end{table}

We observe that
there are five genes: Igfbp2, Gstm2, Vwf, 2010002N04Rik, and Ccnl2 that are also selected in \cite{lin2015regularization}.
These overlapping genes are highly likely to be connected to obesity from a
biological point of view. This  indicates the effectiveness and stability of
our method in finding the risky genes.
Table \ref{tab:selection_prob} also shows
our method identifies some other genes that were not previously found, possibly due to the nonparametric form we consider
in the first stage. In fact, some genes selected with high probability by our
method have been verified by many biological studies. In particular, Solute
Carrier Family 43 Member 1 (Slc43a1) is a Protein Coding gene which can encode
the amino acid transporters that are known to regulate the transmembrane
transport of phenylalanine. \cite{gill2010statistical} found that the expression
of Slc43a1 in the fat mice group is quite different from that in the lean mice
group so Slc43a1 is a potential factor leading to obesity.

We also apply our inferential procedure to the data to quantify the uncertainty
associated with our estimation. We use the R package Flare to obtain the optimal
precision matrix from solving \eqref{eq:opt_precision}. Then we construct the
confidence intervals based on the debiased estimator. Table \ref{tab:debias}
presents the causal effects of the genes on the body weight whose corresponding
confidence intervals do not contain zero. Note that the confidence intervals are
wide due to possibly low signal-to-noise ratio and small sample size of the
data. The result is generally consistent with that obtained by the stable
selection. The confidence intervals that are far away from zero include Vwf,
2010002N04Rik, Gstm2, Gp1d1, Slc43a1, Igfbp2, which are also shown to have high
selection probability in both Table~\ref{tab:selection_prob}. Moreover, several other genes from
Table~\ref{tab:debias} that are shown to have significant causal effects on the
body weight are newly found and have been confirmed to have close biological
relation with obesity. For example, a recent study in \cite{wang2019anxa2} shows
that by silencing Anxa2, the obesity-induced insulin resistance is attenuated
and our result confirms such a positive relation. Cyp4f15 genes are known to
control the omega-hydroxylated fatty acids in the liver tissue and such acids can be
used for energy production \citep{hardwick2009ppar}. Another
independent study shows that the downregulation of Cyp4f15 happens in the liver
tissue among the group of mice fed with high-fat diet \citep{gai2020farnesoid}.
The negative causal effect obtained in our result coincides with these findings.

\begin{table}[htbp!]
  \centering
  \caption{$95\%$ confidence intervals for the causal effects of the genes on
    the body weights of the mice. Shown are only the genes whose corresponding
    intervals do not contain zero.}
  \begin{threeparttable}[b]
\begin{tabular}{lrlr}
  \toprule
  Gene Name & Confidence Interval  & Gene Name & Confidence Interval\\
  \midrule
  Anxa5         & (0.010, 7.269) & Kif22   & (0.615, 7.930)\\
  Vwf           & (0.500, 7.841) & Gstm2   & (0.537, 8.231) \\
  Aqp8          & (0.066, 6.855) & Gpld1   & ($-$7.448, $-$0.447)\\
  Lamc1         & (0.094, 5.877) & Slc43a1 & ($-$6.641, $-$1.412)\\
  Acot9         & (0.056, 8.298) & Abca8a  & ($-$7.152, $-$0.072)\\
  Anxa2         & (1.086, 9.331) & Cyp4f15 & ($-$7.468, $-$0.250)\\
  2010002N04Rik & (1.343, 8.240) & Igfbp2  & ($-$6.451, $-$0.666)\\
  Msr1          & (0.004, 6.783)\\
  \bottomrule
\end{tabular}
\end{threeparttable}
\label{tab:debias}
\end{table}

\section{Discussion}
Motivated by the data-driven modeling spirit, we develop a high-dimensional additive instrumental-variables regression method with a sound non-asymptotic analysis and valid inference procedure when both instruments and treatments are allowed to be high-dimensional.
There are a lot of directions that are worth further exploration. Firstly, while we estimate the nonparametric functions separately in the first-stage, it would be helpful to take into account the correlations among the
treatments and borrow information across those regressions.
Secondly, although it is convenient for causal interpretation when considering the linear outcome model, many applications in econometrics  and biostatistics have considered nonparametric method to model the relationships between treatments and outcome.
One may further relax our second-stage model setting by considering a
 high-dimensional single-index or nonparametric additive outcome model \citep{radchenko2015high,huang2010}.
Thirdly, it is of great interest to extend our results to handling other types of outcome data, for instance, binary or survival outcome. We plan to pursue
these and other related issues in future research.


\acks{We are grateful to Prof. Wei Lin for sharing the data. This research is
supported by the National Natural Science Foundation of China (Grant Nos.
12101607, 12071015), the Fundamental Research Funds for the Central Universities
and the Research Funds of Renmin University of China.
}



\appendix
\section*{Appendix}

We first define the notation that will be used throughout the Appendix.
Suppose $A\in\mathbb{R}^{m\times n}$. Let $\|A\|_2$ be the largest singular
value of $A$ and $\|A\|_{\infty}$ the largest absolute value of the entries of
$A$. Also, let $\|A\|_{1}=\max_{j\in [n]}\|A_{j}\|_1$, where $A_j$ is the $j$th
column of $A$ and $[n]=\{1,\ldots,n\}$. Denote by $\lambda_{\min}(A)$ and
$\lambda_{\max}(A)$ the minimum and maximum eigenvalues of $A$, respectively,
when $A$ is symmetric.

\section*{Appendix A: some useful Lemmas}
\label{app:theorem}



In this appendix present some useful lemmas that will be used for proving the main results.

\noindent

\begin{lemma}[Hoeffding Bound, Proposition 2.5 of \citeauthor{wainwright_2019},
  \citeyear{wainwright_2019}]\label{lem:hoeffding}
  Suppose random variables $X_i,i=1,\ldots,n$ are independent, and $X_i$ has
  mean $\mu_i$ and sub-Gaussian parameter $\sigma_i$. Then for all $t\geq0$, we
  have
\begin{align*}
  \pr\bigg\{\bigg|\frac{1}{n}\sum_{i=1}^{n}(X_i-\mu_i)\bigg|\geq t\bigg\}
  \leq2\exp\bigg(-\frac{n^{2}t^{2}}{2\sum_{i=1}^{n}\sigma_i^{2}}\bigg).
\end{align*}
Specifically, if $X_i$ is bounded in $[a,b]$, we have
\begin{align*}
  \pr\bigg\{\bigg|\frac{1}{n}\sum_{i=1}^{n}(X_i-\mu_i)\bigg|\geq
  t\bigg\}\leq2\exp\left\{-\frac{2nt^{2}}{(b-a)^{2}}\right\}.
\end{align*}
\end{lemma}

Lemma~\ref{lem:B_spline_approx} follows from Lemma 5 of
\cite{stone1985additive}, which is a well-known result on B-spline approximation
and has been frequently used in the additive models.
\begin{lemma}\label{lem:B_spline_approx}
  For each $f_{j\ell}\in\mathcal{F}$, there exists
  $\Bar{f}_{nj\ell}=\sum_{k=1}^{m}\Bar{\gamma}_{kj\ell}\phi_k$ such that
  $$\sup_{z\in[a,b]}|f_{j\ell}(z)-\Bar{f}_{nj\ell}(z)|\leq C_{L}m^{-d},$$
  where $C_L>0$ is a constant depending only on the degree $L$ of the B-splines.
\end{lemma}

The next lemma shows the minimum and maximum eigenvalues of the spline matrix
have the same order. Denote the non-centered B-spline matrix by
$\tU=(\tU_1,\ldots,\tU_q)\in\mathbb{R}^{n\times{}qm},$ where
$\tU_j=(\tU_{1j},\ldots,\tU_{nj})^{\T}\in\mathbb{R}^{n\times m}$ and
$\tU_{ij}=\{\phi_{1}(Z_{ij}),\ldots,\phi_{m}(Z_{ij})\}^{\T} \in\mathbb{R}^{m}$.
\begin{lemma}[Lemma 6.2 in \citeauthor{shen1998local}, \citeyear{shen1998local}]
  \label{lem:eigenvalue_order_SplineMat}
  For $j=1,\ldots,q$, we have
  \begin{align*}
    \frac{3c_*}{m}-2\|\mathbb{P}_j-\mathbb{P}^n_j\|_{\infty}
    \leq\lambda_{\min}\Bigl(\frac{\tU_j^\T\tU_j}{n}\Bigr)
    \leq\lambda_{\max}\Bigl(\frac{\tU_j^\T\tU_j}{n}\Bigr)
    \leq \frac{c^*}{2m}+2\|\mathbb{P}_j-\mathbb{P}^n_j\|_{\infty},
  \end{align*}
  where $0<c_{*}<1<c^{*}$ and
  $\|\mathbb{P}_j-\mathbb{P}^n_j\|_{\infty}=\sup_{z}\bigl|\pr(Z_j\leq
  z)-n^{-1}\sum_{i=1}^{n}\bm{1}\{Z_{ij}\leq z\}\bigr|$ with $\mathbb{P}_j$ and
  $\mathbb{P}^{n}_j$ being the population and empirical distributions of
  $Z_{j}$, respectively.
\end{lemma}

The following result gives a bound on $\|\mathbb{P}_j-\mathbb{P}^n_j\|_{\infty}$.
\begin{lemma}[Glivenko--Cantelli Theorem, Corollary 4.15 of
  \citeauthor{wainwright_2019}, \citeyear{wainwright_2019}]
  \label{lem:G-C_Thm}
  Let $\mathbb{P}$ be the distribution of a random variable $X$ and
  $\mathbb{P}_n$ the empirical distribution based on $n$ i.i.d. copies,
  $X_{1},\ldots,X_{n}$, of $X$. Then, we have
  \begin{align*}
    \left\|\mathbb{P}_n-\mathbb{P}\right\|_\infty\leq 8\{\log(n+1)/n\}^{1/2}+\delta
  \end{align*}
  with probability at least $1-\exp(-n\delta^{2}/2)$ for any $\delta\geq0$.
\end{lemma}

Combining Lemmas~\ref{lem:eigenvalue_order_SplineMat} and~\ref{lem:G-C_Thm}, we
obtain the following result on the B-spline matrix.
\begin{lemma}\label{lem:eigenvalue-tilde-Uj}
  If $8\{\log(n+1)/n\}^{1/2}+\{2/m^{3}\}^{1/2}\leq\min\{c_*/m,c^*/4m\}$, then
  for each $j=1,\ldots,q$, with probability at least $1-\exp(-nm^{-3})$, we
  have
\begin{align*}
  \frac{c_*}{m}\leq\lambda_{\min}\Bigl(\frac{\tU_j^\T\tU_j}{n}\Bigr)
  \leq\lambda_{\max}\Bigl(\frac{\tU_j^\T\tU_j}{n}\Bigr)\leq\frac{c^*}{m}.
\end{align*}
\end{lemma}
\begin{proof}
  Apply the Glivenko--Cantelli Theorem (Lemma~\ref{lem:G-C_Thm}) with
  $\delta=\sqrt{2/m^{3}}$.
\end{proof}

\begin{lemma}[Bernstein's Inequality, see Proposition 2.14 of
  \citeauthor{wainwright_2019}, \citeyear{wainwright_2019}]
  \label{lem:heavy_tail_bound}
  Let $X_1,\ldots, X_n$ be i.i.d. random variables such that
  $\left|X_i\right|\leq b$ almost surely for some constant $b>0$. Then for any
  $\delta>0$,
  \begin{align*}
    \pr\biggl(\bigg|\frac{1}{n}\sum_{i=1}^{n}\big\{X_i-\mathbb{E}(X_i)\big\}
    \bigg|\leq\delta\biggr)\geq1-2\exp[-n\delta^{2}/\{2(b^{2}+\delta/3)\}].
  \end{align*}
\end{lemma}

\begin{lemma}[Concentration Inequality for $\chi^{2}$ Variable]
  \label{lem:concen_chi}
  Suppose $X_1,\ldots, X_n$ are i.i.d. standard normal random variables. Then
  for $t\geq0$,
  \begin{align*}
    \pr\bigg(\bigg|\frac{1}{n}\sum_{i=1}^{n}X_i^{2}-1\bigg|
    \geq t\bigg)\leq 2\exp\big(-nt^{2}/8\big).
  \end{align*}
\end{lemma}
\begin{proof}
  See Example 2.11 of \cite{wainwright_2019}.
\end{proof}

We present a key concentration result which will be used in the following
proofs.
\begin{lemma}[Lemma B.1 in \citeauthor{lounici2011oracle},
  \citeyear{lounici2011oracle}]\label{lem:concentration}
  Let $\xi_{1},\ldots,\xi_{N}$ be i.i.d. standard normal random variables.
  Moreover, let $v=(v_{1},\ldots,v_{N})\neq0,$
  $\eta_{v}=\sum_{i=1}^{N}(\xi_{i}^{2}-1)v_{i}/(\sqrt{2}\|v\|_{2})$ and
  $m(v)=\|v\|_{\infty}/\|v\|_{2}$. Then for all $x>0$, we have
  \begin{align*}
    \pr(|\eta_{v}|>x)\leq2\exp\bigg(-\frac{x^{2}}
    {2\big\{1+\sqrt{2}xm(v)\big\}}\bigg).
  \end{align*}
\end{lemma}

The following lemma validates the restricted eigenvalue condition given
Assumption \ref{assu:re}.
\begin{lemma}\label{lem:pop_emp_re}
  Suppose Assumption \ref{assu:re} holds. If
  $160r\{m^{4}\log(pqm)/n\}^{1/2}\leq\rho/2$, with probability at least
  $1-6(qm)(pqm)^{-2}$,
  \begin{align*}
    \min\bigg\{\frac{\gama^{\T}\U^{\T}\U\gama}{n\|\gama_{J}\|_2^{2}}\colon
    |J|\leq r,\gama\in\mathbb{R}^{qm}\backslash\{0\},
    \sum_{j\in J^{c}}\|\gama_{j}\|_2\leq 3\sum_{j\in J}\|\gama_{j}\|_2
    \bigg\}\geq\frac{\rho}{2m}.
  \end{align*}
\end{lemma}

\begin{proof}
  Note that
  \begin{align*}
    \Sigma-\frac{\U^{\T}\U}{n}
    &=\frac{n-1}{n}\mathbb{E}\Bigg[
      \frac{\big\{\tU-\mathbb{E}(\tU)\big\}^{\T}
      \big\{\tU-\mathbb{E}(\tU)\big\}}{n}\Bigg]
      -\frac{\big\{\tU-\mathbb{E}(\tU)\big\}^{\T}
      \big\{\tU-\mathbb{E}(\tU)\big\}}{n}\\
    &\qquad+\frac{\mathbb{E}(\tU)^{\T}\mathbb{E}(\tU)}{n}
      +\frac{\tU^{\T}11^{\T}\tU}{n^{2}}
      -\frac{\mathbb{E}(\tU)^{\T}\tU}{n}
      -\frac{\tU^{\T}\mathbb{E}(\tU)}{n}.
  \end{align*}
  Let $\tilde{\Sigma}=\mathbb{E}[\{\tU-\mathbb{E}(\tU)\}^{\T}
  \{\tU-\mathbb{E}(\tU)\}/n]$. Then, we have
  \begin{align*}
    \left\|\Sigma-\frac{\U^{\T}\U}{n}\right\|_{\infty}
    &\leq\left\|\frac{n-1}{n}\tilde{\Sigma}
      -\frac{(\tU-\mathbb{E}(\tU))^{\T}
      (\tU-\mathbb{E}(\tU))}{n}\right\|_{\infty}
      +\left\|\frac{\mathbb{E}(\tU)^{\T}
      \mathbb{E}(\tU)}{n}-\frac{\tU^{\T}
      \mathbb{E}(\tU)}{n}\right\|_{\infty}\\
    &\qquad+\left\|\frac{\tU^{\T}11^{\T}\tU}{n^{2}}
      -\frac{\mathbb{E}(\tU)^{\T}\tU}{n}\right\|_{\infty}
    \eqqcolon{}T_1+T_2+T_3.
  \end{align*}
  For $T_2$, we can write
  \begin{align*}
    &T_{2}=\bigg\|\frac{\mathbb{E}(\tU)^{\T}\mathbb{E}(\tU)}{n}
      -\frac{\tU^{\T}\mathbb{E}(\tU)}{n}\bigg\|_{\infty}
      =\bigg\|\mathbb{E}(\tU_{i})\mathbb{E}(\tU_i)^{\T}
      -\frac{\tU^{\T}\mathbb{E}(\tU)}{n}\bigg\|_{\infty}\\
    &=\max_{j,j',k,k'}\bigg|\mathbb{E}\{\phi_{k}(Z_{ij})\}
      \mathbb{E}\{\phi_{k'}(Z_{ij'})\}-\sum_{i=1}^{n}
      \frac{\phi_{k}(Z_{ij})}{n}\mathbb{E}\{\phi_{k'}(Z_{ij'})\}
      \bigg|,\\
    &\qquad\qquad \forall j,j'\in\{1,\ldots,q\},\ k,k'\in\{1,\ldots,m\},
  \end{align*}
  where $\tU_i^{\T}$ is the $i$th row of $\tU$. By the property of the B-spline
  matrix, we have $0\leq\phi_k(z)\leq 1$ and
  $0\leq\mathbb{E}\{\phi_{k}(Z_{ij})\}\leq1$. Therefore,
  \begin{align*}
    T_{2}=\bigg\|\frac{\mathbb{E}(\tU)^{\T}\mathbb{E}(\tU)}{n}
    -\frac{\tU^{\T}\mathbb{E}(\tU)}{n}\bigg\|_{\infty}
    \leq\max_{j,k}\bigg|\mathbb{E}\{\phi_{k}(Z_{ij})\}
    -\sum_{i=1}^{n}\frac{\phi_{k}(Z_{ij})}{n}\bigg|.
  \end{align*}
  Now apply Lemma \ref{lem:hoeffding} with $t=\{4\log(pqm_n)/n\}^{1/2}$ and take
  the union bound to obtain
  \begin{align*}
    T_2=\bigg\|\frac{\mathbb{E}(\tU)^{\T}\mathbb{E}(\tU)}{n}
    -\frac{\tU^{\T}\mathbb{E}(\tU)}{n}\bigg\|_{\infty}
    \leq\bigg\{\frac{4\log(pqm_n)}{n}\bigg\}^{1/2},
  \end{align*}
  which holds with probability at least $1-2(qm)(pqm)^{-2}$. For $T_3$, the same
  bound can be obtained with the same probability since
  \begin{align*}
    &T_3=\bigg\|\frac{\tU^{\T}11^{\T}\tU}{n^{2}}
      -\frac{\mathbb{E}(\tU)^{\T}\tU}{n}\bigg\|_{\infty}\\
    &=\max_{j,j',k,k'}\bigg|\frac{\left\{\sum_{i=1}^{n}
      \phi_{k}(Z_{ij})\right\}\left\{\sum_{i=1}^{n}\phi_{k'}
      (Z_{ij'})\right\}}{n^{2}}-\frac{\left\{\sum_{i=1}^{n}
      \phi_{k}(Z_{ij})\right\}\mathbb{E}\{\phi_{k'}(Z_{ij'})\}}{n}\bigg|\\
    &=\max_{j,j',k,k'}\biggl|\frac{\sum_{i=1}^{n}\phi_{k}(Z_{ij})}{n}
      \biggr|\cdot\biggl|\frac{\sum_{i=1}^{n}\phi_{k'}(Z_{ij'})}
      {n}-\mathbb{E}\{\phi_{k'}(Z_{ij'})\}\biggr|\\
    &\leq \max_{j',k'}\bigg|\frac{\sum_{i=1}^{n}\phi_{k'}
      (Z_{ij'})}{n}-\mathbb{E}\{\phi_{k'}(Z_{ij'})\}\bigg|.
  \end{align*}
  Now we bound $T_1$
  \begin{align*}
    T_1
    &=\Bigg\|\frac{n-1}{n}\tilde{\Sigma}-\frac{\big\{\tU
      -\mathbb{E}(\tU)\big\}^{\T}\big\{\tU
      -\mathbb{E}(\tU)\big\}}{n}\Bigg\|_{\infty}\\
    &\leq\Bigg\|\tilde{\Sigma}-\frac{\big\{\tU
      -\mathbb{E}(\tU)\big\}^{\T}\big\{\tU
      -\mathbb{E}(\tU)\big\}}{n}\Bigg\|_{\infty}
      +\frac{1}{n}\big\|\tilde{\Sigma}\big\|_{\infty}\\
    &\leq\max_{j,j',k,k'}\bigg|\mathbb{E}\big[\phi_{k}(Z_{ij})
      \phi_{k'}(Z_{ij'})-\mathbb{E}\big\{\phi_{k}(Z_{ij})\big\}
      \mathbb{E}\big\{\phi_{k'}(Z_{ij'})\big\}\big]\\
    &\qquad-\frac{1}{n}\sum_{i=1}^{n}\big[\phi_{k}(Z_{ij})
      -\mathbb{E}\big\{\phi_{k}(Z_{ij})\big\}\big]\big[\phi_{k'}
      (Z_{ij'})-\mathbb{E}\big\{\phi_{k'}(Z_{ij'})\big\}\big]\bigg|+\frac{2}{n},
  \end{align*}
  where the last inequality holds because each entry of $\tilde{\Sigma}$ can be
  bounded
  \begin{align*}
    &|\mathbb{E}\{\phi_{k}(Z_{ij})\phi_{k'}(Z_{ij'})\}
      -\mathbb{E}\{\phi_{k}(Z_{ij})\}\mathbb{E}\{\phi_{k'}(Z_{ij'})\}|\\
    &\leq|\mathbb{E}\{\phi_{k}(Z_{ij})\phi_{k'}(Z_{ij'})\}|
      +|\mathbb{E}\{\phi_{k}(Z_{ij})\}\mathbb{E}\{\phi_{k'}(Z_{ij'})\}|
      \leq 2.
  \end{align*}
  Note $|\phi_{k}(Z_{ij})-\mathbb{E}\{\phi_{k}(Z_{ij})\}|\leq1,\forall
  j\in\{1,\ldots,q\},k\in\{1,\ldots,m\}$, so we have
  \begin{align*}
    \big|\big[\phi_{k}(Z_{ij})-\mathbb{E}\{\phi_{k}(Z_{ij})\}\big]
    \big[\phi_{k'}(Z_{ij'})-\mathbb{E}\{\phi_{k'}(Z_{ij'})\}\big]\big|
    &\leq 1.
  \end{align*}
  Applying Lemma \ref{lem:hoeffding} again with $t=\{8\log(pqm)/n\}^{1/2}$, we
  obtain by the union bound argument over $j,j'\in\{1,\ldots,
  q\},k,k'\in\{1,\ldots,m\}$
  \begin{align}\label{eq:bound-t123}
    T_1\leq \{8\log(pqm)/n\}^{1/2}+2/n\leq 6\{\log(pqm)/n\}^{1/2},
  \end{align}
  which holds with probability at least $1-2(qm)^{2}(pqm)^{-4}$. Thus, we have
  \begin{align*}
    \bigg\|\Sigma-\frac{\U^{\T}\U}{n}\bigg\|_{\infty}
    \leq T_1+T_2+T_3\leq 10\bigg\{\frac{\log(pqm)}{n}\bigg\}^{1/2},
  \end{align*}
  which holds with probability at least $1-6(qm)(pqm)^{-2}$. Consider $\gama$
  such that $\sum_{j\in J^{c}}\|\gama_j\|_2\leq 3\sum_{j\in J}\|\gama_j\|_2$,
  then by Assumption \ref{assu:re}, we have
  \begin{align*}
    &\frac{\gama^{\T}\U^{\T}\U\gama}{n}
      =\frac{\gama^{\T}\U^{\T}\U\gama}{n}-\gama^{\T}\Sigma\gama
      +\gama^{\T}\Sigma\gama\geq\frac{\rho\|\gama_J\|_2^{2}}{m}
      -\bigg|\frac{\gama^{\T}\U^{\T}\U\gama}{n}
      -\gama^{\T}\Sigma\gama\bigg|\\
    & \geq\frac{\rho\|\gama_J\|_2^{2}}{m}-\|\gama\|_1^{2}
      \bigg\|\Sigma-\frac{\U^{\T}\U}{n}\bigg\|_{\infty}
      =\frac{\rho\|\gama_J\|_2^{2}}{m}-\Bigg(\sum_{j\in J}\|\gama_j\|_1
      +\sum_{j\in J^{c}}\|\gama_{j}\|_1\Bigg)^{2}\bigg\|
      \Sigma-\frac{\U^{\T}\U}{n}\bigg\|_{\infty}\\
    &
      \geq\frac{\rho\|\gama_J\|_2^{2}}{m}-m\Bigg(\sum_{j\in J}
      \|\gama_j\|_2+\sum_{j\in J^{c}}\|\gama_j\|_2\Bigg)^{2}
      \bigg\|\Sigma-\frac{\U^{\T}\U}{n}\bigg\|_{\infty}\\
    &
      \geq \frac{\rho\|\gama_J\|_2^{2}}{m}-m\Bigg(\sum_{j\in J}
      \|\gama_j\|_2+3\sum_{j\in J}\|\gama_j\|_2\Bigg)^{2}
      \bigg\|\Sigma-\frac{\U^{\T}\U}{n}\bigg\|_{\infty}\\
    &
      =\frac{\rho\|\gama_J\|_2^{2}}{m}-16m\Bigg(\sum_{j\in J}
      \|\gama_j\|_2\Bigg)^{2}\bigg\|\Sigma
      -\frac{\U^{\T}\U}{n}\bigg\|_{\infty}\\
    &
      \geq \frac{\rho\|\gama_J\|_2^{2}}{m}-16rm
      \sum_{j\in J}\|\gama_j\|_2^{2}\bigg\|\Sigma
      -\frac{\U^{\T}\U}{n}\bigg\|_{\infty}
      = \frac{\rho\|\gama_J\|_2^{2}}{m}-16rm
      \|\gama_{J}\|_2^{2}\bigg\|\Sigma
      -\frac{\U^{\T}\U}{n}\bigg\|_{\infty}.
  \end{align*}
  It follows from~\eqref{eq:bound-t123} that
  \begin{align*}
    \frac{\gama^{\T}\U^{\T}\U\gama}{n}\geq
    \frac{\rho\|\gama_J\|_2^{2}}{m}-160rm_n\|\gama_{J}\|_2^{2}
    \bigg\{\frac{\log(pqm)}{n}\bigg\}^{1/2}\geq
    \frac{\|\gama_J\|_2^{2}\rho}{2m}
  \end{align*}
  holds with probability at least $1-6qm(pqm)^{-2}$, as long as
  $160r\{m^{4}\log(pqm)/n\}^{1/2}\leq\rho/2$. This concludes the lemma.
\end{proof}

The following lemma is crucial for proving
Theorem~\ref{thm:estimation_first_stage}.
\begin{lemma}\label{lem:bound_of_residual}
  For every $\ell\in\{1,\ldots,p\}$, consider the random event
  $\mathcal{A}_\ell=\cap_{j=1}^{q}\mathcal{A}_{j\ell}$, where
  \begin{align*}
    \mathcal{A}_{j\ell}=\bigg\{\frac{1}{n}\|\U^{\T}_{j}\varepsilon_{\ell}\|_{2}
    \leq\frac{\lambda_{\ell}}{4}\bigg\}.
  \end{align*}
  If $\lambda_{\ell}\geq4c^{*}\sigma_{\ell}\{14\log(pqm)/n\}^{1/2}$, then
  $\pr(\mathcal{A}^{c}_{\ell})\leq3(pqm)^{-2}q$.
\end{lemma}

\begin{proof}
  Note that
  \begin{align*}
    \pr(\mathcal{A}_{j\ell})=\pr\bigg(\frac{1}{n^{2}}\varepsilon_{\ell}^{\T}
    \U_{j}\U_{j}^{\T}\varepsilon_{\ell}\leq\frac{\lambda^{2}_{\ell}}{16}\bigg)
    =\pr\bigg\{\frac{\sum_{i=1}^{n}\nu_{j,i}(\xi_{i}^{2}-1)}
    {\sqrt{2}\|\nu_{j}\|_{2}}\leq{}x_{j\ell}\bigg\},
  \end{align*}
  where $\xi_{1},\ldots,\xi_{n}$ are i.i.d standard normal random variables,
  $\nu_{j,1},\ldots,\nu_{j,n}$ denote the eigenvalues of the matrix
  $\U_{j}\U_{j}^{\T}/n$ among which the positive ones are the same as those of
  $\Psi_{j}=\U_{j}^{\T}\U_{j}/n$, and
  \begin{align*}
    x_{j\ell}=\frac{\lambda_{\ell}^{2}n/(16\sigma_\ell^{2})
    -\mathrm{tr}(\Psi_{j})}{\sqrt{2}\|\Psi_{j}\|_{F}}.
  \end{align*}
  We bound from the above the probability of $\mathcal{A}_{j}^{c}$ using Lemma
  \ref{lem:concentration}. Specifically, choose
  $\nu=\nu_{j}=(\nu_{j,1},\ldots,\nu_{j,n}),x=x_{j}$ and
  $m(\nu)=\|\Psi_{j}\|_{2}/\|\Psi_{j}\|_{F}=\|v_{j}\|_{\infty}/\|v_{j}\|_{2}$ to
  obtain
  \begin{align*}
    \pr(\mathcal{A}^{c}_{j\ell})\leq 2\exp\bigg\{-\frac{x^{2}_{j\ell}}
    {2\big(1+\sqrt{2}x_{j\ell}\|\Psi_{j}\|_{2}/\|\Psi_{j}\|_{F}\big)}\bigg\}.
  \end{align*}
  Now we find the appropriate $\lambda_{\ell}$ such that the above probability
  approaches one as $p$ and $q$ increase with $n$. Let
  \begin{align*}
    \exp\bigg\{-\frac{x^{2}_{j\ell}}{2\big(1+\sqrt{2}x_{j\ell}
    \|\Psi_{j}\|_2/\|\Psi_{j}\|_{\mathrm{F}}\big)}\bigg\}\leq (pqm)^{-2},
  \end{align*}
  which implies $x_{j\ell}^{2}\geq4\log(pqm)(1+\sqrt{2}x_{j\ell}
  \|\Psi_{j}\|_2/\|\Psi\|_{F})$.
  Equivalently, we need to have
  \begin{align*}
    &x_{j\ell}=\frac{\lambda^{2}_{\ell}n/(16\sigma_\ell^{2})
    -\mathrm{tr}(\Psi_{j})}{\sqrt{2}\|\Psi_{j}\|_{\mathrm{F}}}\\
    &\geq2\sqrt{2}\|\Psi_{j}\|_2/\|\Psi_{j}\|_{\mathrm{F}}\log(pqm)
    +\big[2\{2\|\Psi_{j}\|_2/\|\Psi_{j}\|_{\mathrm{F}}
    \log(pqm)\}^{2}+4\log(pqm)\big]^{1/2},
  \end{align*}
  or
  \begin{align*}
    \lambda^{2}_{\ell}n/16\sigma_\ell^{2}\geq\mathrm{tr}(\Psi_{j})
    +4\|\Psi_{j}\|_2\log(pqm)+\big[4\{2\|\Psi_{j}\|_2\log(pqm)\}^{2}
    +8\|\Psi_{j}\|^{2}_{F}\log(pqm)\big].
  \end{align*}
  Therefore, a sufficient condition for the above is (by noting $(a+b)^{1/2}\leq
  a^{1/2}+b^{1/2}$, $\forall a,b\geq0$)
  \begin{align}\label{eq:lambda-bound}
    \lambda^{2}_{\ell}n/16\sigma^{2}_\ell\geq \mathrm{tr}(\Psi_{j})
    +8\|\Psi_{j}\|_2\log(pqm)+\big\{8\|\Psi_{j}\|^{2}_{F}\log(pqm)\big\}^{1/2}.
  \end{align}
  Therefore, when~\eqref{eq:lambda-bound} holds,
  $\pr(\mathcal{A}_{j\ell}^{c})\leq2(pqm)^{-2}$. It suffices to find the
  probability of~\eqref{eq:lambda-bound} when
  $\lambda_{\ell}\geq4c^{*}\sigma_\ell \{14\log(pqm)/n\}^{1/2}$. By Lemma
  \ref{lem:eigenvalue-tilde-Uj}, when $\{2\log(pqm)m^{3}/n\}^{1/2}\leq 1$, with
  probability at least $1-\exp(-nm^{-3})\geq1-(pqm)^{-2}$, the following
  inequality
  \begin{align*}
    \mathrm{tr}(\Psi_{j})=\mathrm{tr}\bigg(\frac{\U_j^{\T}\U_j}{n}\bigg)
    =\mathrm{tr}\bigg(\frac{\tU_j^{\T}\tU_j}{n}-\frac{\tU_j^{\T}11^{\T}\tU_j}
    {n^{2}}\bigg)\leq\mathrm{tr}\bigg(\frac{\tU_j^{\T}\tU_j}{n}\bigg)
    \leq m\frac{c^{*}}{m}=c^{*}
  \end{align*}
  holds, where the first inequality follows from the positive semidefiniteness
  of $\tU_j^{\T}11^{\T}\tU_j/n^{2}$. Similarly, $\|\Psi_{j}\|_{2}\leq c^{*}/m$
  and $\|\Psi_{j}\|_{F}\leq\{\|\Psi_{j}\|_{2}\mathrm{tr}(\Psi_{j})\}^{1/2}\leq
  c^{*}/\sqrt{m}$. Therefore, when $\lambda_{\ell}\geq4c^{*}\sigma_\ell
  \{14\log(pqm)/n\}^{1/2}$,
  \begin{align*}
    &\lambda^{2}_{\ell}n/(16\sigma^{2}_\ell)
      \geq c^{*}+\frac{8c^{*}\log(pqm)}{m}
      +\bigg\{\frac{8(c^{*})^{2}\log(pqm)}{m}\bigg\}^{1/2}\\
    &\geq\mathrm{tr}(\Psi_{j})+8\|\Psi_{j}\|_2\log(pqm)
      +\big\{8\|\Psi_{j}\|^{2}_{F}\log(pqm)\big\}^{1/2}
  \end{align*}
  holds with probability at least $1-(pqm)^{-2}$. As a result,
  $\mathcal{A}_{j\ell}$ happens with probability at least $1-3(pqm)^{-2}$ when
  $\lambda_{\ell}\geq4\sigma_\ell c^{*}\{14\log(pqm)/n\}^{1/2}$. The proof is
  complete by taking the union bound argument.
\end{proof}

\begin{lemma}[Error Bound for Approximation]
  \label{lem:approximation_error_bound}
  Under the assumption of Lemma \ref{lem:B_spline_approx}, the event
  \begin{align*}
    \mathcal{D}_\ell\coloneqq\Bigg\{\frac{1}{\sqrt{n}}\bigg\|
    \U\bar{\gama}_{\ell}-\sum_{j=1}^{q}F_{j\ell}\bigg\|_{2}
    \leq 4rC_Lm^{-d}+4r\big\{C_0^{2}\log(pqm)/n\big\}^{1/2}\Bigg\}
  \end{align*}
  happens with probability at least $1-2r(pqm)^{-2}$.
\end{lemma}
\begin{proof}
  First, note that
  \begin{align*}
    &\frac{1}{n}\bigg\|\U\bar{\gama}_{\ell}
      -\sum_{j=1}^{q}F_{j\ell}\bigg\|^{2}_{2}
      =\frac{1}{n}\sum_{i=1}^{n}\bigg\{\sum_{j=1}^{q}
      f_{j\ell}(Z_{ij})-\sum_{j=1}^{q}\sum_{k=1}^{m}
      \bar{\gamma}_{jk\ell}\psi_{k}(Z_{ij})\bigg\}^{2}\\
    &=\frac{1}{n}\sum_{i=1}^{n}\bigg\{\sum_{j\in J_\ell}
      f_{j\ell}(Z_{ij})-\sum_{j\in J_\ell}\sum_{k=1}^{m}
      \bar{\gamma}_{jk\ell}\psi_{k}(Z_{ij})\bigg\}^{2}\\
    &\leq \frac{1}{n}r\sum_{i=1}^{n}\sum_{j\in J_\ell}
      \bigg\{f_{j\ell}(Z_{ij})-\sum_{k=1}^{m}
      \bar{\gamma}_{jk\ell}\psi_{k}(Z_{ij})\bigg\}^{2}
      \leq r\sum_{j\in J_\ell}\sup_{z\in[a,b]}
      \big|f_{j\ell}(z)-\tilde{f}_{nj\ell}(z)\big|^{2}.
  \end{align*}
  Applying Lemma \ref{lem:center_approximation} with a union bound argument,
  with probability at least $1-2r(pqm)^{-2}$, we have
  \begin{align*}
    \frac{1}{\sqrt{n}}\bigg\|\U\bar{\gama}_{\ell}
    -\sum_{j=1}^{q}F_{j\ell}\bigg\|_{2}\leq 4rC_Lm^{-d}
    +4r\big\{C_0^{2}\log(pqm)/n\big\}^{1/2}.
  \end{align*}
  This completes the proof.
\end{proof}

The next lemma resembles Lemma \ref{lem:pop_emp_re} but focuses on the classical
restricted eigenvalue condition.
\begin{lemma}\label{lem:re_sec_stage}
  Suppose Assumption \ref{assu:mini_Eigen} holds. Then, with probability at
  least $1-2(pqm)^{-2}$,
  \begin{align*}
    \min\bigg\{\frac{\bata^{\T}F^{\T}F\bata}{n\|\bata_{\mathcal{L}}\|_2^{2}}
    \colon|\mathcal{L}|\leq s,\bata\in\mathbb{R}^{p}\backslash\{0\},
    \sum_{\ell\in \mathcal{L}^{c}}|\beta_{\ell}|\leq
    3\sum_{\ell\in \mathcal{L}}|\beta_{\ell}|\bigg\}\geq\frac{\kappa}{2}
  \end{align*}
  holds when $64sr^{2}C_0\{\log(pqm)/n\}^{1/2}\leq\kappa/2$ and
  $\{\log(pqm)/n\}\leq1/2$.
\end{lemma}

\begin{proof}
  Consider $\bata$ satisfying $\sum_{\ell\in\mathcal{L}^{c}}|\beta_{\ell}|
  \leq3\sum_{\ell\in\mathcal{L}}|\beta_{\ell}|$. Then,
  \begin{align*}
    &\frac{\bata^{\T}F^{\T}F\bata}{n}
      =\bata^{\T}\Sigma_f\bata-\bata^{\T}\Sigma_f\bata
      +\frac{\bata^{\T}F^{\T}F\bata}{n}\\
    &\geq \|\bata_{\mathcal{L}}\|_2^{2}\kappa-\bata^{\T}
      \bigg(\Sigma_f-\frac{F^{\T}F}{n}\bigg)\bata
      \geq \|\bata_{\mathcal{L}}\|_2^{2}\kappa-\|\bata\|_1^{2}
      \bigg\|\Sigma_f-\frac{F^{\T}F}{n}\bigg\|_{\infty}\\
    &=\|\bata_{\mathcal{L}}\|_2^{2}\kappa-\bigg(\sum_{\ell\in\mathcal{L}^{c}}
      |\beta_\ell|+\sum_{\ell\in\mathcal{L}}|\beta_\ell|\bigg)^{2}\bigg\|
      \Sigma_f-\frac{F^{\T}F}{n}\bigg\|_{\infty}\\
    &\geq \|\bata_{\mathcal{L}}\|_2^{2}\kappa-\bigg(3\sum_{\ell\in\mathcal{L}}
      |\beta_\ell|+\sum_{\ell\in\mathcal{L}}|\beta_\ell|\bigg)^{2}
      \bigg\|\Sigma_f-\frac{F^{\T}F}{n}\bigg\|_{\infty}\\
    & =\|\bata_{\mathcal{L}}\|_2^{2}\kappa-16\|\bata_{\mathcal{L}}\|_1^{2}
      \bigg\|\Sigma_f-\frac{F^{\T}F}{n}\bigg\|_{\infty}
      \geq \|\bata_{\mathcal{L}}\|_2^{2}\kappa-16s\|\bata_{\mathcal{L}}\|_2^{2}
      \bigg\|\Sigma_f-\frac{F^{\T}F}{n}\bigg\|_{\infty}.
  \end{align*}
  We first bound $\|\Sigma_f-F^{\T}F/n\|_{\infty}$,
  \begin{align*}
    \bigg\|\Sigma_f-\frac{F^{\T}F}{n}\bigg\|_{\infty}
    =\max_{\ell,\ell'}\Bigg|\mathbb{E}\bigg[&\bigg\{\sum_{j=1}^{q}
      f_{j\ell}(Z_j)\bigg\}\bigg\{\sum_{j=1}^{q}f_{j\ell'}(Z_j)\bigg\}\bigg]\\
    &-\frac{\sum_{i=1}^{n}\big\{\sum_{j=1}^{q}f_{j\ell}(Z_{ij})\big\}
      \big\{\sum_{j=1}^{q}f_{j\ell'}(Z_{ij})\big\}}{n}\Bigg|.
  \end{align*}
  As there are at most $r$ nonzero summands in $\sum_{j=1}^{q}f_{j\ell}(z)$ for
  every $\ell\in\{1,\ldots,p\}$, we have
  \begin{align*}
    \bigg|\sum_{j=1}^{q}f_{j\ell}(z)\bigg|\leq rC_0,
    \ \forall\ell\in\{1,\ldots,p\},
  \end{align*}
  which implies
  \begin{align*}
    \bigg|\bigg\{\sum_{j=1}^{q}f_{j\ell}(Z_j)\bigg\}
    \bigg\{\sum_{j=1}^{q}f_{j\ell}(Z_j)\bigg\}\bigg|\leq r^{2}C_0^{2}.
  \end{align*}
  Now applying Lemma \ref{lem:heavy_tail_bound} with $b=r^{2}C_0^{2}$ and
  $\delta=4r^{2}C_0^{2}\{\log(pqm)/n\}^{1/2}$, with probability at least
  $1-2\exp\big[-8r^{4}C_0^{4}\log(pqm)/ \big\{r^{4}C_0^{4}+\frac{4}{3}
  r^{4}C_0^{4}(\log(pqm)/n)^{1/2}\big\}\big]$, we have
  \begin{align*}
    &\Bigg|\mathbb{E}\bigg[\bigg\{\sum_{j=1}^{q}f_{j\ell}(Z_j)
      \bigg\}\bigg\{\sum_{j=1}^{q}f_{j\ell'}(Z_j)\bigg\}\bigg]
      -\frac{\sum_{i=1}^{n}\big\{\sum_{j=1}^{q}f_{j\ell}(Z_{ij})
      \big\}\big\{\sum_{j=1}^{q}f_{j\ell'}(Z_{ij})\big\}}{n}\Bigg|\\
    &\leq4r^{2}C_0^{2}\bigg\{\frac{\log(pqm)}{n}\bigg\}^{1/2}.
  \end{align*}
  When $\{\log(pqm)/n\}^{1/2}\leq1/2$, we have
  \begin{align*}
    1-2\exp\bigg[-\frac{8r^{4}C_0^{4}\log(pqm)}
    {\big\{r^{4}C_0^{4}+\frac{4}{3} r^{4}C_0^{4}
    (\log(pqm)/n)^{1/2}\big\}}\bigg]\geq 1-2(pqm)^{-4}
  \end{align*}
  and by taking the union bound we get
  \begin{align*}
    \bigg\|\Sigma_f-\frac{F^{\T}F}{n}\bigg\|_{\infty}\leq
    4r^{2}C_0^{2}\bigg\{\frac{\log(pqm)}{n}\bigg\}^{1/2},
  \end{align*}
  which holds with probability at least $1-2p^{2}(pqm)^{-4}$. It follows that
  \begin{align*}
    \frac{\bata^{\T}F^{\T}F\bata}{n}\geq \|\bata_{\mathcal{L}}\|_2^{2}
    \kappa-64s\|\bata_{\mathcal{L}}\|_2^{2}r^{2}C_0^{2}
    \bigg\{\frac{\log(pqm)}{n}\bigg\}^{1/2}
  \end{align*}
  holds with probability at least $1-2p^{2}(pqm)^{-4}$. When
  $64sr^{2}C_0^{2}\{\log(pqm)/n\}^{1/2}\leq\kappa/2$, we have
  $\bata^{\T}F^{\T}F\bata/n\geq\kappa\|\bata_{\mathcal{L}}\|_2^{2}/2$ holds with
  probability at least $1-2(pqm)^{-2}$. This concludes the lemma.
\end{proof}

\section*{Section B: Proof of Results in Section~\ref{sec:proof}}
\subsection*{Proof of Lemma \ref{lem:center_approximation} in the main text}

\begin{proof}
  By the assumption $\mathbb{E}\{f_{j\ell}(z_{ij})\}=0$ and Lemma
  \ref{lem:B_spline_approx}, we know
  \begin{align*}
    &\sup_{z\in[0,1]}\big|f_{j\ell}(z)-\tilde{f}_{nj\ell}(z)\big|\\
    &\leq\sup_{z\in[0,1]}\big|f_{j\ell}(z)-\Bar{f}_{nj\ell}(z)\big|
      +\bigg|\frac{1}{n}\sum_{k=1}^{m_n}\sum_{i=1}^{n}
      \bar{\gamma}_{jk\ell}\phi_{k}(Z_{ij})-\frac{1}{n}\sum_{i=1}^{n}
      f_{j\ell}(Z_{ij})\bigg|\\
      &
      \qquad\qquad+\bigg|\frac{1}{n}\sum_{i=1}^{n}
      f_{j\ell}(Z_{ij})-\mathbb{E}(f_{j\ell}(Z_{ij}))\bigg|\\
    &\leq C_Lm_n^{-d}+\frac{1}{n}\sum_{i=1}^{n}
      \big|\Bar{f}_{nj\ell}(Z_{ij})-f_{j\ell}(Z_{ij})\big|
      +\bigg|\frac{1}{n}\sum_{i=1}^{n}
      f_{j\ell}(Z_{ij})-\mathbb{E}(f_{j\ell}(Z_{ij}))\bigg|\\
    &\leq 2C_Lm_n^{-d}+\bigg|\frac{1}{n}\sum_{i=1}^{n}
      f_{j\ell}(Z_{ij})-\mathbb{E}(f_{j\ell}(Z_{ij}))\bigg|,
  \end{align*}
  where the first inequality follows from the triangle inequality and the second
  to the last inequalities follow from Lemma~\ref{lem:B_spline_approx}. Finally,
  noting $|f(Z_{ij})|\leq C_0$ and applying Lemma \ref{lem:hoeffding} with
  $t=\{4C_0^{2}\log(pqm)/n\}^{1/2}$, we have
  \begin{align*}
    \bigg|\frac{1}{n}\sum_{i=1}^{n}f_{j\ell}(Z_{ij})
    -\mathbb{E}(f_{j\ell})\bigg|\leq \{4C_0^{2}\log(pqm)/n\}^{1/2},
  \end{align*}
  which holds with probability at least $1-2(pqm)^{-2}$. This completes the
  proof.
\end{proof}

\section*{Section C: Proof of Results in Section~\ref{sec:nasymp-anlys}}
\subsection*{Proof of Theorem \ref{thm:estimation_first_stage} in the main text}

Consider the $\ell$th optimization problem in \eqref{eq:solution_first}. By the
optimality of $\widehat{\gama}_\ell$ in \eqref{eq:solution_first}, we have
$\forall\gama_\ell\in\mathbb{R}^{qm}$,
\begin{align*}
  \frac{1}{n}\left\|\U\widehat{\gama}_\ell-\X_\ell\right\|_2^{2}
  +2\lambda_\ell\sum_{j=1}^{q}\|\widehat{\gama}_{j\ell}\|_2
  \leq\frac{1}{n}\left\|\U\gama_\ell-\X_\ell\right\|_2^{2}
  +2\lambda_\ell\sum_{j=1}^{q}\|\gama_{j\ell}\|_2.
\end{align*}
Recall that $\sum_{j=1}^{q}F_{j\ell}+\varepsilon_\ell=\X_\ell$. Therefore,
\begin{align*}
  &\frac{1}{n}\bigg\|\U\widehat{\gama}_\ell-\sum_{j=1}^{q}
    F_{j\ell}+\U\bar{\gama}_\ell-\U\bar{\gama}_\ell
    +\varepsilon_\ell\bigg\|_2^{2}+2\lambda_\ell\sum_{j=1}^{q}
    \|\widehat{\gama}_{j\ell}\|_2\\
  &\leq\frac{1}{n}\bigg\|\U\gama_\ell-\sum_{j=1}^{q}F_{j\ell}
    +\U\bar{\gama}_\ell-\U\bar{\gama}_\ell+\varepsilon_\ell\bigg\|_2^{2}
    +2\lambda_\ell\sum_{j=1}^{q}\|\gama_{j\ell}\|_2.
\end{align*}
Thus, we obtain
\begin{align*}
  &\frac{1}{n}\|\U\widehat{\gama}_\ell-\U\bar{\gama}_\ell\|_2^{2}
    +\frac{1}{n}\|\varepsilon_\ell\|_2^{2}
    +\frac{1}{n}\bigg\|\sum_{j=1}^{q}F_{j\ell}-\U\bar{\gama}_\ell\bigg\|_2^{2}\\
  &\quad+\frac{2}{n}\big(\U\widehat{\gama}-\U\bar{\gama}_{\ell}\big)^{\T}
    \bigg(\varepsilon_\ell+\U\bar{\gama}_{\ell}-\sum_{j=1}^{q}F_{j\ell}\bigg)
    +\frac{2}{n}\varepsilon_\ell^{\T}\bigg(\U\bar{\gama}_\ell
    -\sum_{j=1}^{q}F_{j\ell}\bigg)+2\lambda_\ell\sum_{j=1}^{q}
    \|\widehat{\gama}_{j\ell}\|_2\\
  &\leq\frac{1}{n}\|\U\gama_\ell-\U\bar{\gama}_\ell\|_2^{2}
    +\frac{1}{n}\|\varepsilon_\ell\|_2^{2}+\frac{1}{n}
    \bigg\|\sum_{j=1}^{q}F_{j\ell}-\U\bar{\gama}_\ell\bigg\|_2^{2}\\
  &\quad+\frac{2}{n}(\U\gama_\ell-\U\bar{\gama}_\ell)^{\T}
    \bigg(\varepsilon_\ell+\U\bar{\gama}_\ell-\sum_{j=1}^{q}F_{j\ell}\bigg)
    +\frac{2}{n}\varepsilon_\ell^{\T}\bigg(\U\bar{\gama}_\ell
    -\sum_{j=1}^{q}F_{j\ell}\bigg)+2\lambda_\ell\sum_{j=1}^{q}\|\gama_{j\ell}\|_2,
\end{align*}
which implies that
\begin{align*}
  &\frac{1}{n}\|\U\widehat{\gama}_\ell-\U\bar{\gama}_\ell\|_2^{2}
    +2\lambda_\ell\sum_{j=1}^{q}\|\widehat{\gama}_{j\ell}\|_2\\
  &\leq\frac{1}{n}\|\U\gama_\ell-\U\bar{\gama}_\ell\|_2^{2}
    +\frac{2}{n}\bigg\{\sum_{j=1}^{q}(\U_j\gama_{j\ell}
    -\U_j\widehat{\gama}_{j\ell})\bigg\}^{\T}
    \bigg(\varepsilon_\ell+\U\bar{\gama}_\ell
    -\sum_{j=1}^{q}F_{j\ell}\bigg)+2\lambda_\ell
    \sum_{j=1}^{q}\|\gama_{j\ell}\|_2\\
  &\leq\frac{1}{n}\|\U\gama_\ell-\U\bar{\gama}_\ell\|_2^{2}
    +\sum_{j=1}^{q}\frac{2}{n}(\gama_{j\ell}^{\T}
    -\widehat{\gama}_{j\ell}^{\T})\U_j^{\T}\varepsilon_\ell\\
  &\quad+\sum_{j=1}^{q}\frac{2}{n}(\gama_{j\ell}^{\T}
    -\widehat{\gama}_{j\ell}^{\T})\U_j^{\T}
    \bigg(\U\bar{\gama}_\ell-\sum_{j=1}^{q}F_{j\ell}\bigg)
    +2\lambda_\ell\sum_{j=1}^{q}\|\gama_{j\ell}\|_2\\
  &\leq\frac{1}{n}\|\U\gama_\ell-\U\bar{\gama}_\ell\|_2^{2}
    +\sum_{j=1}^{q}\frac{2}{n}\|\U_j^{\T}\varepsilon_\ell\|_2
    \left\|\widehat{\gama}_{j\ell}-\gama_{j\ell}\right\|_2\\
  &\quad+\sum_{j=1}^{q}\frac{2}{n}\bigg\|\U_j^{\T}
    \bigg(\U\bar{\gama}_{\ell}-\sum_{j=1}^{q}F_{j\ell}\bigg)\bigg\|_2
    \|\widehat{\gama}_{j\ell}-\gama_{j\ell}\|_2
    +2\lambda_\ell\sum_{j=1}^{q}\|\gama_{j\ell}\|_2,
\end{align*}
where the last inequality holds due to the Cauchy--Schwarz inequality. By Lemma
\ref{lem:bound_of_residual} and the choice of $\lambda_\ell$, we have
\begin{align*}
  &\frac{1}{n}\left\|\U\widehat{\gama}_\ell-\U\bar{\gama}_\ell\right\|_2^{2}
    +2\lambda_\ell\sum_{j=1}^{q}\|\widehat{\gama}_{j\ell}\|_2\\
  &\leq\frac{1}{n}\left\|\U\gama_\ell-\U\bar{\gama}_\ell\right\|_2^{2}
    +\frac{\lambda_{\ell}}{2}\sum_{j=1}^{q}
    \|\widehat{\gama}_{j\ell}-\gama_{j\ell}\|_2\\
  &\quad+\sum_{j=1}^{q}\frac{\|\U_j\|_2}{n^{1/2}}
    \frac{2\big\|\U\bar{\gama}_{\ell}-\sum_{j=1}^{q}F_{j\ell}\big\|_2}
    {n^{1/2}}\|\widehat{\gama}_{j\ell}-\gama_{j\ell}\|_2
    +2\lambda_\ell\sum_{j=1}^{q}\|\gama_{j\ell}\|_2
\end{align*}
with probability at least $1-3q(pqm)^{-2}$. Then by Lemma
\ref{lem:eigenvalue-tilde-Uj}, we have
\begin{align*}
  &\frac{1}{n}\left\|\U\widehat{\gama}_\ell-\U\bar{\gama}_\ell\right\|_2^{2}
    +2\lambda_\ell\sum_{j=1}^{q}\|\widehat{\gama}_{j\ell}\|_2\\
  &\leq\frac{1}{n}\left\|\U\gama_\ell-\U\bar{\gama}_\ell\right\|_2^{2}
    +\frac{\lambda_{\ell}}{2}\sum_{j=1}^{q}
    \|\widehat{\gama}_{j\ell}-\gama_{j\ell}\|_2\\
  &\quad+\sum_{j=1}^{q}\Bigl(\frac{c^{*}}{m}\Bigr)^{\frac12}
    \frac{2\|\U\bar{\gama}_{\ell}-\sum_{j=1}^{q}F_{j\ell}\|_2}{n^{\frac12}}
    \|\widehat{\gama}_{j\ell}-\gama_{j\ell}\|_2
    +2\lambda_\ell\sum_{j=1}^{q}\|\gama_{j\ell}\|_2,
\end{align*}
which holds with probability at least $1-4q(pqm)^{-2}$ (taking the union bound
here) when $1-\exp(-nm^{-3})\geq 1-(pqm)^{-2}$. Then by Lemma
\ref{lem:approximation_error_bound}, we have
\begin{align*}
  &\frac{1}{n}\|\U\widehat{\gama}_\ell-\U\bar{\gama}_\ell\|_2^{2}
    +2\lambda_\ell\sum_{j=1}^{q}\|\widehat{\gama}_{j\ell}\|_2\\
  &\leq\frac{1}{n}\|\U\gama_\ell-\U\bar{\gama}_\ell\|_2^{2}
    +\frac{\lambda_{\ell}}{2}\sum_{j=1}^{q}
    \|\widehat{\gama}_{j\ell}-\gama_{j\ell}\|_2
    +2\lambda_\ell\sum_{j=1}^{q}\|\gama_{j\ell}\|_2\\
  &\quad+\bigg[8rC_L(c^{*})^{1/2}m^{-1/2-d}+8r\biggl\{
    \frac{c^{*}C_0^{2}\log(pqm)}{nm}\biggr\}^{1/2}\bigg]
    \sum_{j=1}^{q}\|\widehat{\gama}_{j\ell}-\gama_{j\ell}\|_2\\
  &\leq\frac{1}{n}\|\U\gama_\ell-\U\bar{\gama}_\ell\|_2^{2}
    +\frac{\lambda_{\ell}}{2}\sum_{j=1}^{q}\|\widehat{\gama}_{j\ell}
    -\gama_{j\ell}\|_2+\frac{\lambda_\ell}{2}\sum_{j=1}^{q}
    \|\widehat{\gama}_{j\ell}-\gama_{j\ell}\|_2
    +2\lambda_\ell\sum_{j=1}^{q}\|\gama_{j\ell}\|_2\\
  &\leq\frac{1}{n}\|\U\gama_\ell-\U\bar{\gama}_\ell\|_2^{2}
    +\lambda_\ell\sum_{j=1}^{q}\|\widehat{\gama}_{j\ell}
    -\gama_{j\ell}\|_2+2\lambda_\ell\sum_{j=1}^{q}\|\gama_{j\ell}\|_2,
\end{align*}
which holds with probability at least $1-6q(pqm)^{-2}$ since $r\leq q$.
Setting $\gama_\ell=\bar{\gama}_\ell$, we have
\begin{align}
  &\nonumber
    \frac{1}{n}\|\U\widehat{\gama}_\ell-\U\bar{\gama}_\ell\|_2^{2}
    +\lambda_{\ell}\sum_{j=1}^{q}\|\widehat{\gama}_{j\ell}
    -\bar{\gama}_{j\ell}\|_2\\
  &\nonumber
    \leq 2\lambda_\ell\sum_{j=1}^{q}\|\widehat{\gama}_{j\ell}
    -\bar{\gama}_{j\ell}\|_2+2\lambda_\ell\sum_{j=1}^{q}
    \|\bar{\gama}_{j\ell}\|_2-2\lambda_\ell\sum_{j=1}^{q}
    \|\widehat{\gama}_{j\ell}\|_2\\
  &\label{eq:first_stage_oracle}
    \leq4\lambda_\ell\sum_{j\in J_\ell}
    \|\bar{\gama}_{j\ell}-\widehat{\gama}_{j\ell}\|_2.
\end{align}
Therefore, we know $\lambda_{\ell}\sum_{j=1}^{q}
\|\widehat{\gama}_{j\ell}-\bar{\gama}_{j\ell}\|_2\leq
4\lambda_\ell\sum_{j\in{}J_\ell}\|\bar{\gama}_{j\ell}
-\widehat{\gama}_{j\ell}\|_2$, which implies that
\begin{align*}
  \sum_{j\in J^{c}_{\ell}}\|\widehat{\gama}_{j\ell}-\bar{\gama}_{j\ell}\|_2
  \leq3\sum_{j\in J_\ell}\|\bar{\gama}_{j\ell}-\widehat{\gama}_{j\ell}\|_2
\end{align*}
holds with probability at least $1-6q(pqm)^{-2}$. Then, combine with Lemma
\ref{lem:pop_emp_re} to obtain
\begin{align*}
  \|\bar{\gama}_{J_{\ell}\ell}-\widehat{\gama}_{J_{\ell}\ell}\|_2
  \leq\frac{\|\U\bar{\gama}_{\ell}-\U\widehat{\gama}_{\ell}\|_2}
  {n^{1/2}}\bigg(\frac{2m}{\rho}\bigg)^{1/2},
\end{align*}
which holds with probability at least $1-12(qm)(pqm)^{-2}$. We also know from
inequality \eqref{eq:first_stage_oracle} that
\begin{align*}
  \frac{1}{n}\|\U\widehat{\gama}_\ell-\U\bar{\gama}_{\ell}\|_2^{2}
  \leq4\lambda_\ell\sum_{j\in J_\ell}
  \|\bar{\gama}_{j\ell}-\widehat{\gama}_{j\ell}\|_2
  \leq4\lambda_\ell\sqrt{r}\|\bar{\gama}_{J_\ell}
  -\widehat{\gama}_{J_\ell}\|_2,
\end{align*}
which holds with probability at least $1-6q(pqm)^{-2}$, where the second
inequality follows from the Cauchy--Schwarz inequality. Then, it follows that
\begin{align*}
  \frac{\|\U\widehat{\gama}_\ell-\U\bar{\gama}_{\ell}\|_2}{n^{1/2}}
  \leq4\lambda_{\ell}r^{1/2}\bigg(\frac{2m}{\rho}\bigg)^{1/2}
\end{align*}
holds with probability at least $1-18(qm)(pqm)^{-2}$. It holds with the same
probability that
\begin{align*}
  &\sum_{j=1}^{q}\|\widehat{\gama}_{j\ell}-\bar{\gama}_{j\ell}\|_2
    \leq4\sum_{j\in J_\ell}\|\bar{\gama}_{j\ell}
    -\widehat{\gama}_{j\ell}\|_2\leq4n^{1/2}
    \|\bar{\gama}_{J_\ell}-\widehat{\gama}_{J_\ell}\|_2\\
  &\leq 4r^{1/2}\frac{\|\U\bar{\gama}_{\ell}
    -\U\widehat{\gama}_{\ell}\|_2}{\sqrt{n}}
    \bigg(\frac{2m}{\rho}\bigg)^{1/2}
    \leq 32r\lambda_\ell\frac{m}{\rho}.
\end{align*}
Now we apply Lemma \ref{lem:approximation_error_bound} to obtain
\begin{align*}
  &\bigg\|\sum_{j=1}^{q}F_{j\ell}-\U\widehat{\gama}_{\ell}\bigg\|_2
    =\bigg\|\sum_{j=1}^{q}F_{j\ell}-\U\bar{\gama}_{\ell}
    +\U\bar{\gama}_{\ell}-\U\widehat{\gama}_{\ell}\bigg\|_2\\
  &\leq\bigg\|\sum_{j=1}^{q}F_{j\ell}-\U\bar{\gama}_{\ell}\bigg\|_2
  +\|\U\bar{\gama}_{\ell}-\U\widehat{\gama}_{\ell}\|_2\\
  &\leq4rn^{1/2}C_{L}m^{-d}+4rn^{1/2}C_0\{\log(pqm)/n\}^{1/2}
    +4\lambda_{\ell}(rn)^{1/2}\bigg(\frac{2m}{\rho}\bigg)^{1/2}\\
  &\leq5\lambda_{\ell}(rn)^{1/2}\bigg(\frac{2m}{\rho}\bigg)^{1/2},
\end{align*}
which holds with probability at least $1-20(qm)(pqm)^{-2}$, where the last
inequality follows from the definition of $\lambda_\ell$. The proof is complete
by taking the union bound over all $\ell$.

\subsection*{Proof of Theorem \ref{thm:second_stage_consistency}}

\begin{lemma}\label{lem:first_inter_res_first}
  Suppose Assumptions \ref{assu:smoothness}--\ref{assu:mini_Eigen} hold. If the
  regularization parameters satisfy
\begin{align}\label{eq:lambda_max}
  560C_0\lambda_{\max}\bigg(\frac{2rm}{\rho}\bigg)^{1/2}
  \leq\frac{\kappa^{2}}{4rs},
\end{align}
then with probability at least $1-62(pqm)^{-1}$, the matrix
$\widehat{\X}=\U\widehat{\Gama}$ satisfies
\begin{align*}
  \min\bigg\{\frac{\|\widehat{\X}\bata\|}{n^{1/2}\|\bata_{\L}\|}
  \colon|\L|\leq s,\bata\in\mathbb{R}^{p}\backslash \{0\},
  \|\bata_{\mathcal{L}^{c}}\|_{1}\leq 3\|\bata_{\mathcal{L}}\|_{1}\bigg\}
  \geq\frac{\kappa}{2}
\end{align*}
when $n$ is sufficiently large.
\end{lemma}

\begin{proof}
  For any subset $\L\subset\{1,\ldots,p\}$ with $|\L|\leq s$ and any
  $\delta\in\mathbb{R}^{p}$ such that $\delta\neq0$ and
  $\|\delta_{\L^{c}}\|_{1}\leq3\|\delta_{\L}\|_{1}$, we have
\begin{align*}
  \frac{\delta^{\T}(F^{\T}F-\widehat{\X}^{\T}\widehat{\X})\delta}
  {n\|\delta_{\L}\|^{2}_{2}}\leq\frac{\|\delta\|^{2}_{1}
  \max_{1\leq\ell,\ell'\leq p}\bigg|\Bigl(\sum_{j=1}^{q}F_{j\ell'}\Bigr)^{\T}
  \Bigl(\sum_{j=1}^{q}F_{j\ell}\Bigr)-\widehat{\X}^{\T}_{\ell'}
  \widehat{\X}_{\ell}\bigg|}{n\|\delta_{\L}\|^{2}_{2}}.
\end{align*}
Since $\|\delta_{\L^{c}}\|_{1}\leq 3\|\delta_{\L}\|_{1}$, we have
$\|\delta\|^{2}_{1}=(\|\delta_{\L}\|_{1}+\|\delta_{\L^{c}}\|_{1})^{2}
\leq16\|\delta_{\L}\|^{2}_{1}\leq16s\|\delta_{\L}\|^{2}_{2}$, which implies that
\begin{align}\label{eq:entrywise_bound}
  \frac{\delta^{\T}(F^{\T}F-\widehat{\X}^{\T}\widehat{\X})\delta}
  {n\|\delta_{\L}\|^{2}_{2}}\leq\frac{16s\max_{1\leq \ell,\ell'\leq p}
  \bigg|\Bigl(\sum_{j=1}^{q}F_{j\ell'}\Bigr)^{\T}
  \Bigl(\sum_{j=1}^{q}F_{f\ell}\Bigr)
  -\widehat{\X}^{\T}_{\ell'}\widehat{\X}_{\ell}\bigg|}{n}.
\end{align}
To bound the entrywise maximum, we write
\begin{align*}
  &\bigg|\bigg(\sum_{j=1}^{q}F_{j\ell}\bigg)^{\T}
    \bigg(\sum_{j=1}^{q}F_{j\ell'}\bigg)
    -\widehat{\X}^{\T}_{\ell}\widehat{\X}_{\ell'}\bigg|\\
  &=\bigg|\bigg(\widehat{\X}_{\ell}-\sum_{j=1}^{q}
    F_{j\ell}\bigg)^{\T}\bigg(\widehat{\X}_{\ell'}
    -\sum_{j=1}^{q}F_{j\ell'}\bigg)
    +\bigg(\widehat{\X}_{\ell}-\sum_{j=1}^{q}
    F_{j\ell}\bigg)^{\T}\bigg(\sum_{j=1}^{q}F_{j\ell'}\bigg)\\
  &\qquad+\bigg(\sum_{j=1}^{q}F_{j\ell}\bigg)^{\T}
    \bigg(\widehat{\X}_{\ell'}-\sum_{j=1}^{q}
    F_{j\ell'}\bigg)\bigg|\\
  &\leq\bigg|\bigg(\widehat{\X}_{\ell}-\sum_{j=1}^{q}
    F_{j\ell}\bigg)^{\T}\bigg(\widehat{\X}_{\ell'}
    -\sum_{j=1}^{q}F_{j\ell'}\bigg)\bigg|+\bigg|
    \bigg(\widehat{\X}_{\ell}-\sum_{j=1}^{q}
    F_{j\ell}\bigg)^{\T}\bigg(\sum_{j=1}^{q}
    F_{j\ell'}\bigg)\bigg|\\
  &\qquad+\bigg|\bigg(\sum_{j=1}^{q}F_{j\ell}\bigg)^{\T}
    \bigg(\widehat{\X}_{\ell'}
    -\sum_{j=1}^{q}F_{j\ell'}\bigg)\bigg|\eqqcolon{}
    T_{1}+T_{2}+T_{3}.
\end{align*}
For $T_1$, by the Cauchy--Schwarz inequality and Theorem
\ref{thm:estimation_first_stage}, we have
\begin{align*}
  T_{1}\leq\bigg\|\widehat{\X}_{\ell}-\sum_{j=1}^{q}
  F_{j\ell}\bigg\|_{2}\cdot\bigg\|\widehat{\X}_{\ell'}
  -\sum_{j=1}^{q}F_{j\ell'}\bigg\|_{2}
  &\leq25n\lambda_{\max}^{2}r\frac{2m}{\rho},
\end{align*}
which holds with probability at least $1-20(pqm)^{-1}$. For $T_2$, note that
$\big\|\sum_{j=1}^{q}F_{j\ell}\big\|_{2}\leq{}C_0n^{1/2}r$,
$\forall\ell\in\{1,\ldots,p\}$ due to the sparsity assumption. Then, we have
\begin{align*}
  T_{2}\leq\bigg\|\U\widehat{\gama}_{\ell}
  -\sum_{j=1}^{q}F_{j\ell}\bigg\|_{2}\cdot
  \bigg\|\sum_{j=1}^{q}F_{j\ell'}\bigg\|_{2}
  \leq5n\lambda_{\max}\bigg(\frac{2rm}{\rho}\bigg)^{1/2}rC_0
\end{align*}
which holds with probability at least $1-20(pqm)^{-1}$. Similarly, it holds with
the same probability that
\begin{align*}
  T_{3}\leq5n\lambda_{\max}\bigg(\frac{2rm}{\rho}\bigg)^{1/2}rC_0.
\end{align*}
Therefore, we have
\begin{align*}
  \max_{1\leq \ell,\ell'\leq p}\bigg|\bigg(\sum_{j=1}^{q}
  F_{j\ell'}\bigg)^{\T}\bigg(\sum_{j=1}^{q}F_{f\ell}\bigg)
  -\widehat{\X}^{\T}_{\ell'}\widehat{\X}_{\ell}\bigg|
  \leq35n\lambda_{\max}rC_f\bigg(\frac{2rm}{\rho}\bigg)^{1/2},
\end{align*}
which holds with probability at least $1-60(pqm)^{-1}$ when
$\lambda_{\max}\{2m/\rho\}^{1/2}\leq r^{1/2} C_0$. We know from Lemma
\ref{lem:re_sec_stage} that
\begin{align*}
  \min\bigg\{\frac{\|\f\bata\|}{\sqrt{n}\|\bata_{\L}\|}
  \colon|\L|\leq s,\bata\in\mathbb{R}^{p}\backslash \{0\},
  \|\bata_{\L^{c}}\|_{1}\leq 3\|\bata_{\L}\|_{1}\bigg\}
  \geq\frac{\sqrt{2}\kappa}{2}
\end{align*}
holds with probability at least $1-2(pqm)^{-2}$. It follows from
\eqref{eq:lambda_max} and \eqref{eq:entrywise_bound} that
\begin{align*}
  \frac{\delta^{\T}(F^{\T}F-\widehat{\X}^{\T}\widehat{\X})\delta}
  {n\|\delta_{\L}\|^{2}_{2}}\leq560\lambda_{\max}C_0rs
  \bigg(\frac{2rm}{\rho}\bigg)^{1/2}\leq\frac{\kappa^{2}}{4}
\end{align*}
holds with probability at least $1-60(pqm)^{-1}$. Therefore, with probability at
least $1-62(pqm)^{-1}$, we get
\begin{align*}
  \frac{\bata^{\T}(\widehat{\X}^{\T}\widehat{\X})\bata}
  {n\|\bata_{\L}\|^{2}_{2}}\geq\frac{\kappa^{2}}{4}
\end{align*}
This completes the proof.
\end{proof}

\begin{lemma}\label{lem:second_inter_res_first}
  Under Assumptions \ref{assu:smoothness}--\ref{assu:mini_Eigen}, if the
  regularization parameters $\lambda_{\ell}$'s are chosen as in Theorem
  \ref{thm:estimation_first_stage} and
  \begin{align*}
    \mu=2\lambda_{\max}r\bigg(\frac{2m}{\rho}\bigg)^{1/2}
    \big(7\sigma_0+8\sqrt{5}B\max_{\ell}\sigma_{\ell}+30B\big),
  \end{align*}
  then with probability at least $1-86(pqm)^{-1}$, the regularized estimator
  $\widehat{\bata}$ of \eqref{eq:solution_second} satisfies
  \begin{align*}
    \frac{1}{2n}\big\|\widehat{\X}(\widehat{\bata}-\bata)\big\|^{2}_{2}
    +\frac{\mu}{2}\big\|\widehat{\bata}-\bata\big\|_{1}\leq
    2\mu\big\|\widehat{\bata}_{\L}-\bata_{\L}\big\|_{1}.
  \end{align*}
\end{lemma}
\begin{proof}
  By the optimality of $\widehat{\bata}$, we have
  \begin{align*}
    \frac{1}{2n}\big\|\Y-\widehat{\X}\widehat{\bata}\big\|^{2}_{2}
    +\mu\big\|\widehat{\bata}\big\|_{1}\leq\frac{1}{2n}
    \big\|\Y-\widehat{\X}\bata\big\|^{2}_{2}+\mu\|\bata\|_{1}.
  \end{align*}
  Substituting $\Y=\X\bata+\yita$, we have
  \begin{align*}
    &\|\Y-\widehat{\X}\widehat{\bata}\|^{2}_{2}
      =\|\yita-(\widehat{\X}\widehat{\bata}-\X\bata)\|^{2}_{2}\\
    &=\|\yita\|^{2}_{2}+\|\widehat{\X}\widehat{\bata}-\X\bata\|^{2}_{2}
      -2\yita^{\T}(\widehat{\X}\widehat{\bata}-\X\bata)\\
    &=\|\yita\|^{2}_{2}+\|\widehat{\X}\widehat{\bata}-\X\bata
      +\widehat{\X}\bata-\widehat{\X}\bata\|^{2}_{2}-2\yita^{\T}
      (\widehat{\X}\widehat{\bata}-\X\bata)\\
    &=\|\yita\|^{2}_{2}+\|\widehat{\X}(\widehat{\bata}-\bata)\|^{2}_{2}
      +\|(\widehat{\X}-\X)\bata\|^{2}_{2}-2\yita^{\T}
      (\widehat{\X}\widehat{\bata}-\X\bata)\\
    &\qquad+2\bata^{\T}(\widehat{\X}-\X)^{\T}
      \widehat{\X}(\widehat{\bata}-\bata)
  \end{align*}
  and
  \begin{align*}
    \|\Y-\widehat{\X}\bata\|^{2}_{2}
    =\|\yita-(\widehat{\X}-\X)\bata\|^{2}_{2}
    =\|\yita\|^{2}_{2}+\|(\widehat{\X}-\X)\bata\|^{2}_{2}
    -2\yita^{\T}(\widehat{\X}-\X)\bata.
  \end{align*}
  It then follows that
  \begin{align}
    \frac{1}{2n}\|\widehat{\X}(\widehat{\bata}-\bata)\|^{2}_{2}
    &\nonumber\leq\mu\|\bata\|_{1}-\mu\|\widehat{\bata}\|_{1}
      +\frac{1}{n}\yita^{\T}\widehat{\X}(\widehat{\bata}-\bata)
      -\frac{1}{n}\bata^{\T}(\widehat{\X}-\X)^{\T}
      \widehat{\X}(\widehat{\bata}-\bata)\\
    &\label{eq:bound_pred_first_stage}
      \leq\mu\|\bata\|_{1}-\mu\|\widehat{\bata}\|_{1}
      +\bigg\|\frac{1}{n}\widehat{\X}^{\T}\yita
      -\frac{1}{n}\widehat{\X}^{\T}
      (\widehat{\X}-\X)\bata\bigg\|_{\infty}
      \|\widehat{\bata}-\bata\|_{1}.
  \end{align}
  Next, we find a probability bound for the following event
  \begin{align*}
    \bigg\|\frac{1}{n}\widehat{\X}^{\T}\yita-\frac{1}{n}
    \widehat{\X}^{\T}(\widehat{\X}-\X)\bata\bigg\|_{\infty}
    \leq\frac{\mu}{2}.
  \end{align*}
  Substituting $\widehat{\X}=\U\widehat{\Gama}$ and $\X=F+\varepsilon$, we write
  \begin{align*}
    &\frac{1}{n}\widehat{\X}^{\T}\yita
    -\frac{1}{n}\widehat{\X}^{\T}(\widehat{\X}-\X)\bata
    =\frac{1}{n}\big(\U\widehat{\Gama}\big)^{\T}\yita
      -\frac{1}{n}\big(\U\widehat{\Gama}\big)^{\T}
      \big(\U\widehat{\Gama}-\X\big)\bata\\
    &=\frac{1}{n}\widehat{\Gama}^{\T}\U^{\T}\yita
      -\frac{1}{n}\widehat{\Gama}^{\T}\U^{\T}
      (\U\widehat{\Gama}-F-\varepsilon)\bata\\
    &=\frac{1}{n}\big(\widehat{\Gama}^{\T}\U^{\T}
      -F^{\T}\big)\yita+\frac{1}{n}F^{\T}\yita
      +\frac{1}{n}\big(\widehat{\Gama}^{\T}\U^{\T}
      -F^{\T}\big)\varepsilon\bata+\frac{1}{n}F^{\T}
      \varepsilon\bata\\
    &\quad-\frac{1}{n}(\widehat{\Gama}^{\T}\U^{\T}-F^{\T})
      (\U\widehat{\Gama}-F)\bata-\frac{1}{n}F^{\T}
      (\U\widehat{\Gama}-F)\bata\\
    &\eqqcolon{}T_{1}+T_{2}+T_{3}+T_{4}+T_{5}+T_{6}.
  \end{align*}
  For $T_1$, it follows from Theorem \ref{thm:estimation_first_stage} that
  \begin{align*}
    \|T_1\|_\infty\leq\frac{1}{n} \max_{\ell}
    \bigg\|\sum_{j=1}^{q}F_{j\ell}-\U\widehat{\gama}_\ell\bigg\|_2
    \|\yita\|_2\leq5\lambda_{\max}r^{1/2}
    \bigg(\frac{2m}{\rho}\bigg)^{1/2}\frac{\|\eta\|_2}{n^{1/2}}
  \end{align*}
  holds with probability at least $1-20(pqm)^{-1}$. For a Gaussian variable
  $\eta_i$, $\eta_i^{2}$ is sub-exponential and follows the
  $\sigma_{0}^{2}\chi^{2}$ distribution. Apply Lemma \ref{lem:concen_chi} with
  $t=\{16\log(pqm)/n\}^{1/2}$ to obtain
  \begin{align*}
    \bigg|\frac{1}{\sigma^{2}_0n}\|\eta\|^{2}_2-1\bigg|
    \leq4\bigg\{\frac{\log(pqm)}{n}\bigg\}^{1/2},
  \end{align*}
  which holds with probability at least $1-2(pqm)^{-2}$. Therefore, we have
  \begin{align*}
    &\|T_1\|_\infty
      \leq5\lambda_{\max}r^{1/2}\bigg(\frac{2m}{\rho}\bigg)^{1/2}
      \bigg(\frac{\|\eta\|_2^{2}}{n}\bigg)^{1/2}\\
    &\leq5\lambda_{\max}r^{1/2}\bigg(\frac{2m}{\rho}\bigg)^{1/2}
      \bigg[4\sigma_0^{2}\{\log(pqm)/n\}^{1/2}+\sigma_0^{2}\bigg]\\
    &\leq5\lambda_{\max}r^{1/2}\bigg(\frac{10m}{\rho}\bigg)^{1/2}\sigma_0,
  \end{align*}
  which holds with probability at least $1-22(pqm)^{-1}$ when
  $\log(pqm)/n\leq1$. For $T_2$, we have
  \begin{align*}
    \|T_2\|_{\infty}=\frac{1}{n}\|\f^{\T}\yita\|_{\infty}
    \leq\frac{1}{n}\max_{\ell}\bigg|\bigg(\sum_{j=1}^{q}
    F_{j\ell}\bigg)^{\T}\yita\bigg|.
  \end{align*}
  By the tail bound of a standard normal random variable $X$
  \begin{align*}
    \pr\left(|X|\geq{}t\right)\leq2\bigg(
    \frac{2}{\pi}\bigg)^{1/2}\frac{\exp(-t^{2}/2)}{t},
  \end{align*}
  we have
  \begin{align*}
    &\pr\left(\|T_2\|_{\infty}\geq t\right)
      \leq\pr\bigg\{\frac{1}{n}\max_{\ell}\bigg|
      \bigg(\sum_{j=1}^{q}F_{j\ell}\bigg)^{\T}
      \yita\bigg|\geq t\bigg\}\\
    &\leq{}p\pr\bigg(\frac{1}{n}\bigg|
      \sum_{j=1}^{q}F_{j\ell}^{\T}\yita\bigg|\geq{}t\bigg)
      =p\pr\bigg(\frac{1}{n}\bigg|\sum_{j\in J_\ell}
      F_{j\ell}^{\T}\yita\bigg|\geq t\bigg)\\
    &=p\pr\Bigg[\frac{\big|\sum_{j\in J_\ell}F_{j\ell}^{\T}
      \yita\big|}{\big\{\sum_{i=1}^{n}\big(\sum_{j\in J_{\ell}}
      f_{j\ell}(Z_{ij})\big)^{2}\big\}^{1/2}\sigma_0}\geq
      \frac{nt}{\big\{\sum_{i=1}^{n}\big(\sum_{j\in J_{\ell}}
      f_{j\ell}(Z_{ij})\big)^{2}\big\}^{1/2}\sigma_0}\Bigg]\\
    &\leq{}2p\bigg(\frac{2}{\pi}\bigg)^{1/2}
      \exp\Bigg[\frac{-n^{2}t^{2}}{2\sigma_0^{2}
      \sum_{i=1}^{n}\big\{\sum_{j\in J_{\ell} }
      f_{j\ell}(Z_{ij})\big\}^{2}}\Bigg]
      \frac{\sigma_0\bigg[\sum_{i=1}^{n}\big\{
      \sum_{j\in J_{\ell}}f_{j\ell}(Z_{ij})\big\}^{2}\bigg]}{nt}.
  \end{align*}
  Note that
  $\sum_{i=1}^{n}\bigl\{\sum_{j\in J_{\ell} } f_{j\ell}(Z_{ij})\bigr\}^{2}\leq
  r^{2}C_0^{2}n$.
  By setting $t=2C_0\{r^{2}\log(pqm)/n\}^{1/2}\sigma_0$, we obtain
  \begin{align*}
    \pr\bigg[\|T_2\|_{\infty}\geq 2C_0\{r^{2}\log(pqm)/n\}^{1/2}
    \sigma_0\bigg]\leq p(pqm)^{-2}.
  \end{align*}
  Therefore, we can bound $T_2$ as
  \begin{align*}
    \|T_2\|_\infty\leq2C_0\sigma_0\bigg\{
    \frac{r^{2}\log(pqm)}{n}\bigg\}^{1/2},
  \end{align*}
  which holds with probability at least $1-p(pqm)^{-2}$. Now for $T_3$, we have
  \begin{align*}
    \|T_3\|_{\infty}\leq \|\bata\|_1\frac{1}{n}
    \big\|\big(\U\widehat{\Gama}-F\big)^{\T}
    \varepsilon\big\|_{\infty}\leq{}B\frac{1}{n}
    \max_{\ell,\ell'}\bigg\|\U\widehat{\gama}_\ell
    -\sum_{j=1}^{q}F_{j\ell}\bigg\|_2\|\varepsilon_{\ell'}\|_2.
  \end{align*}
  It follows from Theorem \ref{thm:estimation_first_stage} that
  \begin{align*}
    \max_{\ell}\bigg\|\U\widehat{\gama}_\ell-\sum_{j=1}^{q}
    F_{j\ell}\bigg\|_2\leq 5\lambda_{\max}r^{1/2}
    \bigg(\frac{2m}{\rho}\bigg)^{1/2}\sqrt{n},
  \end{align*}
  which implies that
  \begin{align*}
    \|T_3\|_{\infty}\leq5B\lambda_{\max}r^{1/2}
    \bigg(\frac{2m}{\rho}\bigg)^{1/2}
    \max_{\ell'}\frac{\|\varepsilon_{\ell'}\|_2}{\sqrt{n}}
  \end{align*}
  holds with probability at least $1-20(pqm)^{-1}$. Now applying Lemma
  \ref{lem:concen_chi} again with $t=\{16\log(pqm)/n\}^{1/2}$, we get
  \begin{align*}
    \left|\frac{\|\varepsilon_{\ell'}\|_2^{2}}{n\sigma_{\ell'}^{2}}
    -1\right|\leq 4\bigg\{\frac{\log(pqm)}{n}\bigg\}^{1/2}
  \end{align*}
  with probability at least $1-2(pqm)^{-2}$, which implies
  \begin{align*}
    \frac{\|\varepsilon_{\ell'}\|_2}{n^{1/2}}\leq
    \bigg[4\sigma_{\ell'}^{2}\{\log(pqm)/n\}^{1/2}
    +\sigma_{\ell'}^{2}\bigg]\leq \sqrt{5}\sigma_{\ell'}
  \end{align*}
  with probability at least $1-2(pqm)^{-2}$ when $\log(pqm)/n\leq 1$. Therefore,
  we have
  \begin{align*}
    \|T_3\|_{\infty}\leq 5\sqrt{5}\max_{\ell'}\sigma_{\ell'}B
    \lambda_{\max}r^{1/2}\bigg(\frac{2m}{\rho}\bigg)^{1/2}
    =5B\max_{\ell'}\sigma_{\ell'}\lambda_{\max}
    \bigg(\frac{10rm}{\rho}\bigg)^{1/2},
  \end{align*}
  which holds with probability at least $1-22(pqm)^{-1}$. Similar to $T_2$, we
  can bound $T_4$ as
  \begin{align*}
    \|T_4\|_{\infty}= \frac{1}{n}\|F^{\T}\eps\bata\|_{\infty}
    \leq\frac{\|\bata\|_1}{n}\left\|F^{\T}\eps\right\|_{\infty}
    \leq\frac{B}{n}\big\|F^{\T}\eps\big\|_{\infty}\leq
    \frac{B}{n}\max_{\ell,\ell'}\bigg|\bigg(\sum_{j=1}^{q}
    F_{j\ell}\bigg)^{\T}\varepsilon_{\ell'}\bigg|.
  \end{align*}
  Therefore, we have
  \begin{align*}
    &\pr(\|T_4\|_{\infty}\geq t)
      \leq p^{2}\pr\bigg\{\frac{B}{n}\bigg|\bigg(\sum_{j=1}^{q}
      F_{j\ell}\bigg)^{\T}\varepsilon_{\ell'}\bigg|\geq t\bigg\}
      =p^{2}\pr\bigg\{\frac{B}{n}\bigg|\bigg(\sum_{j\in J_\ell}
      F_{j\ell}\bigg)^{\T}\varepsilon_{\ell'}\bigg|\geq t\bigg\}\\
    & =p^{2}\pr\Bigg[\frac{\big|\big(\sum_{j\in J_\ell}F_{j\ell}\big)^{\T}
      \varepsilon_{\ell'}\big|}{\big\{\sum_{i=1}^{n}
      \big(\sum_{j\in J_{\ell} }f_{j\ell}(Z_{ij})\big)^{2}\big\}^{1/2}
      \sigma_{\ell'}}\geq \frac{nt}{B\big\{\sum_{i=1}^{n}
      \big(\sum_{j\in J_{\ell} }f_{j\ell}(Z_{ij})\big)^{2}
      \big\}^{1/2}\sigma_{\ell'}}\Bigg]\\
    &\leq{}2p^{2}\Bigl(\frac{2}{\pi}\Bigr)^{1/2}
      \exp\Bigg[\frac{-n^{2}t^{2}}{2\sigma_{\ell'}^{2}B^{2}
      \sum_{i}\big\{\sum_{j\in J_{\ell} }f_{j\ell}(Z_{ij})
      \big\}^{2}}\Bigg]\frac{B\big[\sum_{i}
      \big\{\sum_{j\in J_{\ell} }f_{j\ell}(Z_{ij})\big\}^{2}
      \big]^{1/2}\sigma_{\ell'}}{nt}.
  \end{align*}
  Since $\sum_{i=1}^{n}\bigl\{\sum_{j\in{}J_{\ell}}f_{j\ell}(Z_{ij})\bigr\}^{2}
  \leq{}r^{2}C_0^{2}n$, by setting
  $t=\max_{\ell'}\sigma_{\ell'}BC_0\{8r^{2}\log(pqm)/n\}^{1/2}$, we obtain
  \begin{align*}
    \|T_4\|_{\infty}\leq \max_{\ell'}\sigma_{\ell'}BC_0
    \{8r^{2}\log(pqm)/n\}^{1/2},
  \end{align*}
  which holds with probability at least $1-p^{2}(pqm)^{-4}$. Now we bound $T_5$
  as
  \begin{align*}
    &\|T_5\|_{\infty}
      \leq\frac{\|\bata\|_1}{n}\big\|\big(\widehat{\Gama}^{\T}
      \U^{\T}-F^{\T}\big)\big(\U\widehat{\Gama}-F\big)\big\|_{\infty}\\
    &\leq\frac{B}{n}\max_{\ell,\ell'}\bigg\|\sum_{j=1}^{q}
      F_{j\ell}-\U\widehat{\gama}_{\ell}\bigg\|_{2}\bigg\|\sum_{j=1}^{q}
      F_{j\ell'}-\U\widehat{\gama}_{\ell'}\bigg\|_2\\
    &=\frac{B}{n}\max_{\ell}\bigg\|\sum_{j=1}^{q}
      F_{j\ell}-\U\widehat{\gama}_{\ell}\bigg\|_2^{2}
      \leq25B\lambda_{\max}^{2}r\frac{2m}{\rho},
  \end{align*}
  which holds with probability at least $1-20(pqm)^{-1}$. Finally, we bound $T_6$
  as
  \begin{align*}
    &\|T_6\|_{\infty}
      \leq\frac{\|\bata\|_1}{n}\big\|F^{\T}
      \big(\U\widehat{\Gama}-F\big)\big\|_{\infty}
      \leq\frac{B}{n}\max_{\ell,\ell'}\bigg\|\sum_{j=1}^{q}
      F_{j\ell}\bigg\|_2\bigg\|\U\widehat{\gama}_{\ell'}
      -\sum_{j=1}^{q}F_{j\ell'}\bigg\|_2\\
    &\leq\frac{B}{n}(rn)^{1/2}C_0\max_{\ell'}
      \bigg\|\U\widehat{\gama}_{\ell'}
      -\sum_{j=1}^{q}F_{j\ell'}\bigg\|_2
      \leq5Br\lambda_{\max}\bigg(\frac{2m}{\rho}\bigg)^{1/2},
  \end{align*}
  which holds with probability at least $1-20(pqm)^{-1}$. Therefore, we have
  \begin{align*}
    &\bigg\|\frac{1}{n}\widehat{\X}^{\T}\yita
      -\frac{1}{n}\widehat{\X}^{\T}(\widehat{\X}-\X)
      \bata\bigg\|_{\infty}\\
    &\leq\|T_1\|_{\infty}+\|T_2\|_{\infty}+\|T_3\|_{\infty}
      +\|T_4\|_{\infty}+\|T_5\|_{\infty}+\|T_6\|_{\infty}\\
    &\leq 5\sigma_0\lambda_{\max}\bigg(\frac{2rm}{\rho}\bigg)^{1/2}
      +2C_0\sigma_0r\bigg\{\frac{\log(pqm)}{n}\bigg\}^{1/2}
      +5B\max_{\ell}\sigma_{\ell}(r)^{1/2}\lambda_{\max}
      \bigg(\frac{10rm}{n}\bigg)^{1/2}\\
    &\qquad+\max_{\ell}\sigma_{\ell}BC_0r\bigg\{\frac{8\log(pqm)}{n}
      \bigg\}^{1/2}+25B\lambda_{\max}^{2}r\frac{2m}{\rho}
      +5Br\lambda_{\max}\bigg(\frac{2m}{\rho}\bigg)^{1/2}\\
    &\leq7\sigma_0\lambda_{\max}\bigg(\frac{2rm}{\rho}\bigg)^{1/2}
      +8B\max_{\ell}\sigma_{\ell}r\lambda_{\max}\bigg(
      \frac{10m}{n}\bigg)^{1/2}+30rB\lambda_{\max}
      \bigg(\frac{2m}{\rho}\bigg)^{1/2}\\
    &\leq\lambda_{\max}r\bigg(\frac{2m}{\rho}\bigg)^{1/2}
      \big(7\sigma_0+8\sqrt{5}B\max_{\ell}\sigma_{\ell}+30B\big),
  \end{align*}
  which holds with probability at least $1-86(pqm)^{-1}$ when
  $\lambda_{\max}(2m/\rho)^{1/2}\leq1$, where we use the definition of
  $\lambda_{\max}$ in the third inequality. This, together with
  \eqref{eq:bound_pred_first_stage}, implies
  \begin{align*}
    \frac{1}{2n}\|\widehat{\X}(\widehat{\bata}-\bata)\|^{2}_{2}
    \leq\mu\|\bata\|_{1}-\mu\|\widehat{\bata}\|_{1}
    +\frac{\mu}{2}\|\widehat{\bata}-\bata\|_{1}.
  \end{align*}
  Adding $\mu\|\widehat{\bata}-\bata\|^{2}_{2}/2$ to both sides yields
  \begin{align*}
    &\frac{1}{2n}\|\widehat{\X}(\widehat{\bata}-\bata)\|^{2}_{2}
      +\frac{\mu}{2}\|\widehat{\bata}-\bata\|_{1}
      \leq\mu(\|\bata\|_{1}-\|\widehat{\bata}\|_{1}
      +\|\widehat{\bata}-\bata\|_{1})\\
    &=\mu(\|\bata_{\L}\|_{1}+\|\bata_{\L^{c}}\|_{1}
      -\|\widehat{\bata}_{\L}\|_{1}-\|\widehat{\bata}_{\L^{c}}\|_{1}
      +\|\widehat{\bata}_{\L}-\bata_{\L}\|_{1}
      +\|\widehat{\bata}_{\L^{c}}-\bata_{\L^{c}}\|_{1})\\
    &=\mu(\|\bata_{\L}\|_{1}-\|\widehat{\bata}_{\L}\|_{1}
      -\|\widehat{\bata}_{\L^{c}}\|_{1}+\|\widehat{\bata}_{\L}
      -\bata_{\L}\|_{1}+\|\widehat{\bata}_{\L^{c}}\|_{1})\\
    &=\mu(\|\bata_{\L}\|_{1}-\|\widehat{\bata}_{\L}\|_{1}
      +\|\widehat{\bata}_{\L}-\bata_{\L}\|_{1})
      \leq2\mu\|\widehat{\bata}_{\L}-\bata_{\L}\|_{1},
  \end{align*}
  where the last inequality follows from the triangle inequality. This completes
  the proof.
\end{proof}
Now we are ready to prove our main results. We note from Lemma
\ref{lem:second_inter_res_first} that with probability at least
$1-86(pqm)^{-1}$, we have
\begin{align}\label{eq:combination_estimation_pred}
  \frac{1}{2n}\|\widehat{\X}(\widehat{\bata}-\bata)\|^{2}_{2}
  \leq2\mu\|\widehat{\bata}_{\L}-\bata_{\L}\|_{1}\leq2\mu s^{1/2}
  \|\widehat{\bata}_{\L}-\bata_{\L}\|_{2}
\end{align}
and
\begin{align*}
  \frac{\mu}{2}\|\widehat{\bata}-\bata\|_{1}
  \leq 2\mu\|\widehat{\bata}_{\L}-\bata_{\L}\|_{1}.
\end{align*}
By Lemma \ref{lem:first_inter_res_first}, with probability at least
$1-148(pqm)^{-1}$, we have
\begin{align}\label{eq:re_application}
  \|\widehat{\bata}_{\L}-\bata_{\L}\|_{2}\leq
  \frac{2\|\widehat{\X}(\widehat{\bata}-\bata)\|_{2}}
  {n^{1/2}\kappa}.
\end{align}
Combining \eqref{eq:combination_estimation_pred} and \eqref{eq:re_application},
we obtain
\begin{align*}
  \|\widehat{\X}(\widehat{\bata}-\bata)\|^{2}_{2}
  \leq\frac{64}{\kappa^{2}}ns\mu^{2}
\end{align*}
and
\begin{align*}
  \|\widehat{\bata}-\bata\|_{1}\leq4\|\widehat{\bata}_{\L}-\bata_{\L}\|_{1}
  \leq4s^{1/2}\|\widehat{\bata}_{\mathcal{L}}-\bata_{\mathcal{L}}\|_{2}
  \leq\frac{64}{\kappa^{2}}s\mu,
\end{align*}
which hold with probability at least $1-234(pqm)^{-1}$.

\section*{Section D: Proof of Results in Section \ref{sec:inference}}
\subsection*{Proof of Lemma \ref{lem:decom}}

\begin{proof}
  Note that
  \begin{align*}
    &\widetilde{\bata}
      =\widehat{\bata}+\widehat{\Omega}\widehat{\D}^\T
      (\Y-\X\widehat{\bata})/n\\
    &=\widehat{\bata}+\widehat{\Omega}\widehat{\D}^\T
      (\X(\bata-\widehat{\bata})+\yita)/n\\
    &=\widehat{\bata}+\widehat{\Omega}\widehat{\D}^\T
      (\widehat{\D}(\bata-\widehat{\bata})
      +(\X-\widehat{\D})(\bata-\widehat{\bata})+\yita)/n\\
    &=\bata+\widehat{\Omega}\widehat{\D}^\T\yita/n
      +\widehat{\Omega}\widehat{\D}^\T(\X-\widehat{\D})
      (\bata-\widehat{\bata})/n+(\widehat{\Omega}
      \widehat{\Sigma}_{\d}-I)(\bata-\widehat{\bata}).
  \end{align*}
  Now decompose the second term on the rightmost-hand side of the above display
  as
  \begin{align*}
    &\widehat{\Omega}\widehat{\D}^\T\yita/n
      =\widehat{\Omega}\D^\T\yita/n
      +\widehat{\Omega}(\widehat{\D}-\D)^\T\yita/n\\
    &=\Omega\D^\T\yita/n+(\widehat{\Omega}-\Omega)\D^\T
      \yita/n+\widehat{\Omega}(\widehat{\D}-\D)^\T\yita/n,
  \end{align*}
  which completes the proof.
\end{proof}

\subsection*{Proof of Theorem \ref{thm:asym_normality}}

\begin{proof}
  In this proof, we will temporarily assume $\|R_k\|_{\infty}=o_p(1),k=1,2,3,4$,
  which will be elaborated in the next section. First, we show $\omega_\ell$ is
  bounded away from zero for large $n$, that is, $\omega_{\ell}\geq c$ for some
  constant $c>0$. This follows immediately since $\Omega_{\ell\ell}$ is lower
  bounded by a constant. Define
  \begin{align*}
    T_{\ell,n}=\frac{1}{\sqrt{n}}\sum_{i=1}^{n}
    \big(\theta_{\ell}^{\T}\Gama^{\T}
    \tU_i^{\T}\big)\frac{\yita_i}{\omega_{\ell}},
  \end{align*}
  where $\tU_{i}$ is the $i$th row of $\tU$. Then, we have
  \begin{align*}
    &\bigg|T_{\ell,n}-n^{1/2}\frac{\tilde{\bata}_{\ell}
    -\bata_{\ell}}{\omega_{\ell}}\bigg|
    =\bigg|\frac{1}{\sqrt{n}}\sum_{i=1}^{n}\theta_{\ell}^{\T}
      \Gama^{\T}\tU_i^{\T}\frac{\yita_i}{\omega_{\ell}}
      -\frac{\theta_{\ell}^{\T}\D^{\T}\yita}{n^{1/2}\omega_\ell}
      -\sum_{l=1}^{4}\frac{R_{\ell l}}{\omega_{\ell}}\bigg|\\
    &\leq\bigg|\frac{1}{\sqrt{n}}\sum_{i=1}^{n}
      \theta_{\ell}^{\T}\Gama^{\T}\tU_i^{\T}
      \frac{\yita_i}{\omega_{\ell}}-\frac{\theta_{\ell}^{\T}
      \D^{\T}\yita}{\sqrt{n}\omega_{\ell}}\bigg|+\sum_{l=1}^{4}
      \frac{\left\|R_{l}\right\|_{\infty}}{\min_{\ell}\omega_{\ell}}.
  \end{align*}
  As noted above, we will control the remainder terms $\|R_k\|_{\infty}$ in the
  next section and for now we assume they are all $o_{p}(1)$. Now we have
  \begin{align*}
    &\bigg|\frac{1}{n^{1/2}}\sum_{i=1}^{n}\theta_{\ell}^{\T}
      \Gama^{\T}\tU_i^{\T}\frac{\yita_i}{\omega_{\ell}}
      -\frac{\theta_{\ell}^{\T}\D^{\T}\yita}{n^{1/2}
      \omega_{\ell}}\bigg|
      =\frac{1}{n^{1/2}\omega_\ell}\bigg|\frac{1}{n}
      \theta_{\ell}^{\T}\Gama^{\T}\tU^{\T}11^{\T}\yita\bigg|
      \leq\frac{1}{n^{1/2}\omega_{\ell}}\big|\theta_{\ell}^{\T}
      \Gama^{\T}\tU^{\T}1\big|\frac{1}{n}|1^{\T}\yita|\\
    &\leq\frac{1}{n^{1/2}\omega_{\ell}}\bigg\{\left|
      \theta_{\ell}^{\T}\left(\Gama^{\T}\tU^{\T}-F^{\T}\right)
      1\right|+\left|\theta_{\ell}^{\T}F^{\T}1\right|
      \bigg\}\frac{1}{n}\left|1^{\T}\yita\right|,
  \end{align*}
  where the first equality follows from the fact
  \begin{align*}
    &\bigg|\frac{1}{n^{1/2}}\sum_{i=1}^{n}\theta_{\ell}^{\T}
      \Gama^{\T}\tU_i^{\T}\frac{\yita_i}{\omega_{\ell}}
      -\frac{\theta_{\ell}^{\T}\Gama^{\T}\U^{\T}\yita}
      {n^{1/2}\omega_{\ell}}\bigg|\\
    &=\frac{1}{n^{1/2}\omega_{\ell}}\bigg|\sum_{i=1}^{n}
      \theta_{\ell}^{\T}\Gama^{\T}\tU_i^{\T}\yita_i
      -\sum_{i=1}^{n}\theta_{\ell}^{\T}\Gama^{\T}\tU_i^{\T}
      \yita_i+\frac{1}{n}\theta_{\ell}^{\T}\Gama^{\T}
      \tU^{\T}11^{\T}\yita\bigg|\\
    &=\frac{1}{n^{1/2}\omega_\ell}\bigg|\frac{1}{n}
      \theta_{\ell}^{\T}\Gama^{\T}\tU^{\T}11^{\T}\yita\bigg|.
  \end{align*}
  Apply the H\"older's inequality to get
  \begin{align*}
    &\bigg|\frac{1}{n^{1/2}}\sum_{i=1}^{n}\theta_{\ell}^{\T}
      \Gama^{\T}\tU_i^{\T}\frac{\yita_i}{\omega_{\ell}}
      -\frac{\theta_{\ell}^{\T}\D^{\T}\yita}{n^{1/2}
      \omega_{\ell}}\bigg|\\
    &\leq\frac{1}{n^{1/2}\omega_{\ell}}\Bigg(\left\|
      \theta_\ell\right\|_{1}\max_{\ell}\bigg\|\sum_{j=1}^{q}
      F_{j\ell}-\tU\bar{\gama}_{\ell}\bigg\|_2\|1\|_2
      +\left\|\theta_{\ell}\right\|_1\left\|F^{\T}1\right\|_{\infty}
      \Bigg)\frac{1}{n}\bigg|\sum_{i=1}^{n}\yita_i\bigg|\\
    &\leq\Bigg\{m_{\Omega}\max_{\ell}\bigg\|\sum_{j=1}^{q}
      F_{j\ell}-\tU\bar{\gama}_{\ell}\bigg\|_2/\omega_{\ell}
      +m_{\Omega}\max_{\ell}\bigg|\sum_{j=1}^{q}\sum_{i=1}^{n}
      f_{j\ell}(Z_{ij})\bigg|/(n^{1/2}\omega_{\ell})\Bigg\}
      \frac{1}{n}\bigg|\sum_{i=1}^{n}\yita_i\bigg|\\
    &\leq\Bigg\{m_{\Omega}\max_{\ell}\sum_{j\in J_{\ell}}
      \sqrt{n}\sup_{z}\bigg|f_{j\ell}(z)-\sum_{k=1}^{m_n}
      \phi_{k}(z)\bar{\gamma}_{kj\ell}\bigg|/\omega_{\ell}\\
    &\qquad+m_{\Omega}\max_{\ell}\sum_{j\in J_{\ell}}\bigg|\sum_{i=1}^{n}
      f_{j\ell}(Z_{ij})\bigg|/(n^{1/2}\omega_{\ell})\Bigg\}
      \frac{1}{n}\bigg|\sum_{i=1}^{n}\yita_i\bigg|\\
    &\leq\bigg\{m_{\Omega}rn^{1/2}C_{L}m^{-d}/\omega_{\ell}
      +m_{\Omega}\max_{\ell}r\bigg|\sum_{i=1}^{n}f_{j\ell}(Z_{ij})
      \bigg|/(n^{1/2}\omega_{\ell})\bigg\}\frac{1}{n}
      \bigg|\sum_{i=1}^{n}\yita_i\bigg|.
  \end{align*}
  It then follows from the Gaussian tail probability that
  \begin{align*}
    \pr\bigg(\bigg|\frac{1}{n^{1/2}\sigma_0}\sum_{i=1}^{n}
    \yita_i\bigg|\geq t \bigg)\leq 2\bigg(
    \frac{2}{\pi}\bigg)^{1/2}\frac{\exp(-t^{2}/2)}{t}.
  \end{align*}
  Setting $t=2\{\log(n)\}^{1/2}$, we obtain with probability at least $1-n^{-1}$
  \begin{align*}
    \bigg|\frac{1}{n}\sum_{i=1}^{n}\yita_i\bigg|
    \leq\sigma_0\bigg\{\frac{\log(n)}{n}\bigg\}^{1/2}.
  \end{align*}
  Similarly, by Lemma \ref{lem:hoeffding}, we have
  \begin{align*}
    \bigg|\frac{1}{n}\sum_{i=1}^{n}f_{j\ell}(Z_{ij})
    \bigg|\leq \big\{4C_0^{2}\log(pqm)/n\big\}^{1/2}
  \end{align*}
  with probability at least $1-2(pqm)^{-2}$. By the union bound argument, we
  have
  \begin{align*}
    \left|\frac{1}{n^{1/2}}\sum_{i=1}^{n}\theta_{\ell}^{\T}
    \Gama^{\T}\tU_i^{\T}\frac{\yita_i}{\omega_{\ell}}
    -\frac{\theta_{\ell}^{\T}\D^{\T}\yita}{n^{1/2}
    \omega_{\ell}}\right|\leq \sigma_0\{\log(n)\}^{1/2}
    r\frac{m_{\Omega}}{\omega_{\ell}}\bigg[C_{L}m^{-d}
    +2C_0\big\{\log(pqm)/n\big\}^{1/2}\bigg]
  \end{align*}
  with probability at least $1-2(pqm)^{-1}-n^{-1}$. Since $\omega_{\ell}$ is
  lower bounded by a constant when $n$ is large, we conclude
  \begin{align*}
    \bigg|\frac{1}{n^{1/2}}\sum_{i=1}^{n}\theta_{\ell}^{\T}
    \Gama^{\T}\tU_i^{\T}\frac{\yita_i}{\omega_{\ell}}
    -\frac{\theta_{\ell}^{\T}\D^{\T}\yita}{n^{1/2}
    \omega_{\ell}}\bigg|=o_{p}(1).
  \end{align*}
  It follows that
  \begin{align*}
    \bigg|T_{\ell,n}-n^{1/2}\frac{\tilde{\bata}_{\ell}
    -\bata_{\ell}}{\omega_{\ell}}\bigg|=o_{p}(1).
  \end{align*}
  Therefore, $T_{\ell,n}$ and
  $\sqrt{n}(\tilde{\bata}_{\ell}-\bata_{\ell})/\omega_{\ell}$ share the same
  weak limit. Now we show $T_{\ell,n}$ converges in distribution to the standard
  normal. Note that
  \begin{align*}
    &T_{\ell,n}
      =\frac{1}{n^{1/2}}\sum_{i=1}^{n}\theta_{\ell}^{\T}
      \Gama^{\T}\tU_i^{\T}\frac{\yita_i}{\omega_{\ell}}\\
    &=\frac{1}{n^{1/2}}\sum_{i=1}^{n}\frac{\theta_{\ell}^{\T}
      \Gama^{\T}\tU_{i}^{\T}\yita_i}{\big\{\frac{1}{n}
      \sum_{i=1}^{n}\theta_{\ell}^{\T}\Gama^{\T}\mathbb{E}
      \big(\tU_i^{\T}\tU_i\big)\Gama\theta_{\ell}\big\}^{1/2}
      \sigma_0}\frac{\big\{\frac{1}{n}\sum_{i=1}^{n}
      \theta_{\ell}^{\T}\Gama^{\T}\mathbb{E}\big(\tU_i^{\T}
      \tU_i\big)\Gama\theta_{\ell}\big\}^{1/2}\sigma_0}
      {\sigma_0\big\{\theta_{\ell}^{\T}\mathbb{E}
      \left(\Gama^{\T}\U^{\T}\U\Gama/n\right)
      \theta_{\ell}\big\}^{1/2}}\\
    &\qquad\times \frac{\sigma_0\big\{\theta_{\ell}^{\T}
      \mathbb{E}\left(\Gama^{\T}\U^{\T}\U\Gama/n\right)
      \theta_{\ell}\big\}^{1/2}}{\omega_\ell}.
  \end{align*}
  It follows that
  \begin{align*}
    &\pr\bigg[\frac{1}{\sqrt{n}}\sum_{i=1}^{n}
      \frac{\theta_{\ell}^{\T}\Gama^{\T}\tU_{i}^{\T}
      \yita_i}{\big\{\frac{1}{n}\sum_{i=1}^{n}\theta_{\ell}^{\T}
      \Gama^{\T}\mathbb{E}\big(\tU_i^{\T}\tU_i\big)
      \Gama\theta_{\ell}\big\}^{1/2}\sigma_0}\leq t\bigg]\\
    &=\mathbb{E}\Bigg[\pr\Bigg\{\frac{1}{n^{1/2}}
      \sum_{i=1}^{n}\frac{\theta_{\ell}^{\T}\Gama^{\T}
      \tU_{i}^{\T}\yita_i}{\big(\frac{1}{n}\sum_{i=1}^{n}
      \theta_{\ell}^{\T}\Gama^{\T}\mathbb{E}\big(\tU_i^{\T}
      \tU_i\big)\Gama\theta_{\ell}\big)^{1/2}\sigma_0}
      \leq t\bigg|\Z \Bigg\}\Bigg]=\mathbb{E}\{\Phi(t)\}=\Phi(t)
  \end{align*}
  for all $t\in\mathbb{R}$, where $\Phi(t)$ denotes the CDF of the standard
  normal distribution. Finally, we note that
  \begin{align*}
    \frac{\big\{\frac{1}{n}\sum_{i=1}^{n}\theta_{\ell}^{\T}
    \Gama^{\T}\mathbb{E}\big(\tU_i^{\T}\tU_i\big)\Gama
    \theta_{\ell}\big\}^{1/2}\sigma_0}{\sigma_0\big\{\theta_{\ell}^{\T}
    \mathbb{E}\left(\Gama^{\T}\U^{\T}\U\Gama/n\right)
    \theta_{\ell}\big\}^{1/2}}=\bigg(\frac{n}{n-1}\bigg)^{1/2},
  \end{align*}
  where we use the fact that
  $\mathbb{E}(\U^{\T}\U/n)=\frac{n-1}{n}\mathbb{E}(\tU^{\T}\tU/n)$. Then, by the
  triangular inequality, we have
  \begin{align*}
    &\bigg\|\mathbb{E}\bigg(\frac{F^{\T}F}{n}
      -\frac{\Gama^{\T}\U^{\T}\U\Gama}{n}\bigg)\bigg\|_{\infty}\\
    &=\bigg\|\mathbb{E}\bigg(\frac{F^{\T}F}{n}-\frac{n-1}{n}
      \frac{\Gama^{\T}\tU^{\T}\tU\Gama}{n}\bigg)\bigg\|_{\infty}\\
    &\leq\bigg\|\mathbb{E}\bigg(\frac{n-1}{n}\frac{F^{\T}F}{n}
      -\frac{n-1}{n}\frac{\Gama^{\T}\tU^{\T}\tU\Gama}{n}\bigg)
      \bigg\|_{\infty}+\bigg\|\frac{1}{n}\mathbb{E}
      \bigg(\frac{F^{\T}F}{n}\bigg)\bigg\|_{\infty}\\
    &=\frac{n-1}{n}\bigg\|\mathbb{E}\bigg(
      \frac{F^{\T}F}{n}-\frac{F^{\T}\tU\Gama}{n}
      +\frac{F^{\T}\tU\Gama}{n}-\frac{\Gama^{\T}
      \tU^{\T}\tU\Gama}{n}\bigg)\bigg\|_{\infty}
      +\bigg\|\frac{1}{n}\mathbb{E}\bigg(
      \frac{F^{\T}F}{n}\bigg)\bigg\|_{\infty}\\
    &\leq\frac{n-1}{n}\bigg\|\mathbb{E}\bigg(
      \frac{F^{\T}F}{n}-\frac{F^{\T}\tU\Gama}{n}
      \bigg)\bigg\|_{\infty}+\frac{n-1}{n}
      \bigg\|\mathbb{E}\bigg(\frac{F^{\T}\tU\Gama}{n}
      -\frac{\Gama^{\T}\tU^{\T}\tU\Gama}{n}\bigg)\bigg\|_{\infty}\\
    &\qquad+\bigg\|\frac{1}{n}\mathbb{E}\bigg(
      \frac{F^{\T}F}{n}\bigg)\bigg\|_{\infty}.
  \end{align*}
  It then follows from the triangular and Jensen's inequalities that
  \begin{align*}
    &\bigg\|\mathbb{E}\bigg(\frac{F^{\T}F}{n}
      -\frac{\Gama^{\T}\U^{\T}\U\Gama}{n}\bigg)\bigg\|_{\infty}\\
    &\leq\frac{n-1}{n}\max_{\ell,\ell'}\bigg|\mathbb{E}
      \bigg\{\frac{\big(\sum_{j=1}^{q}F_{j\ell}^{\T}\big)
      \big(\sum_{j=1}^{q}F_{j\ell'}-\tU\gama_{\ell'}
      \big)}{n}\bigg\}\bigg|\\
    &\qquad+\frac{n-1}{n}\max_{\ell,\ell'}\bigg|\mathbb{E}
      \bigg\{\frac{\big(\tU\gama_{\ell}\big)^{\T}\big(
      \sum_{j=1}^{q}F_{j\ell'}-\tU\gama_{\ell'}\big)}{n}\bigg\}\bigg|\\
    &\qquad+\frac{1}{n}\mathbb{E}\max_{\ell,\ell'}
      \bigg\|\frac{\sum_{j=1}^{q}F_{j\ell}}{n^{1/2}}\bigg\|_{2}
      \bigg\|\frac{\sum_{j=1}^{q}F_{j\ell'}}{n^{1/2}}\bigg\|_{2}\\
    &\leq\frac{n-1}{n}\mathbb{E}\max_{\ell,\ell'}
      \bigg|\bigg\{\frac{\big(\sum_{j=1}^{q}F_{j\ell}^{\T}
      \big)\big(\sum_{j=1}^{q}F_{j\ell'}-\tU\gama_{\ell'}
      \big)}{n}\bigg\}\bigg|\\
    &\qquad+\frac{n-1}{n}\mathbb{E}\max_{\ell,\ell'}
      \bigg|\frac{\big(\tU\gama_{\ell}\big)^{\T}
      \big(\sum_{j=1}^{q}F_{j\ell'}-\tU\gama_{\ell'}\big)}{n}\bigg|\\
    &\qquad+\frac{1}{n}\mathbb{E}\max_{\ell,\ell'}
      \bigg\|\frac{\sum_{j=1}^{q}F_{j\ell}}{n^{1/2}}
      \bigg\|_{2}\bigg\|\frac{\sum_{j=1}^{q}
      F_{j\ell'}}{n^{1/2}}\bigg\|_{2}.
  \end{align*}

  Now we apply the Cauchy--Schwarz inequality to get
  \begin{align*}
    &\bigg\|\mathbb{E}\bigg(\frac{F^{\T}F}{n}
      -\frac{\Gama^{\T}\U^{\T}\U\Gama}{n}\bigg)\bigg\|_{\infty}\\
    &\leq\frac{n-1}{n}\mathbb{E}\max_{\ell,\ell'}
      \frac{\big\|\sum_{j=1}^{q}F_{j\ell}\big\|_2}{n^{1/2}}
      \bigg\|\frac{\sum_{j=1}^{q}F_{j\ell'}
      -\tU\gama_{\ell'}}{n^{1/2}}\bigg\|_2\\
    &\qquad+\frac{n-1}{n}\mathbb{E}\max_{\ell,\ell'}
      \bigg\|\frac{\tU\gama_{\ell}-\sum_{j=1}^{q}F_{j\ell}
      +\sum_{j=1}^{q}F_{j\ell}}{n^{1/2}}\bigg\|_2\bigg\|
      \frac{\sum_{j=1}^{q}F_{j\ell'}-\tU\gama_{\ell'}}
      {n^{1/2}}\bigg\|_2\\
    &\qquad+\frac{1}{n}\mathbb{E}\max_{\ell,\ell'}\bigg\|
      \frac{\sum_{j=1}^{q}F_{j\ell}}{n^{1/2}}\bigg\|_{2}
      \bigg\|\frac{\sum_{j=1}^{q}F_{j\ell'}}{n^{1/2}}\bigg\|_{2}\\
    &\leq\frac{n-1}{n}rC_0\mathbb{E}\max_{\ell}\bigg\|
      \frac{\sum_{j=1}^{q}F_{j\ell'}-\tU\gama_{\ell'}}
      {n^{1/2}}\bigg\|\\
    &\qquad+\frac{n-1}{n}\mathbb{E}\max_{\ell,\ell'}
      \Bigg\{\bigg(\bigg\|\frac{\tU\gama_{\ell}
      -\sum_{j=1}^{q}F_{j\ell}}{n^{1/2}}\bigg\|
      +\bigg\|\frac{\sum_{j=1}^{q}F_{j\ell}}{n^{1/2}}
      \bigg\|\bigg)\bigg\|\frac{\sum_{j=1}^{q}
      F_{j\ell'}-\tU\gama_{\ell'}}{n^{1/2}}\bigg\|\Bigg\}\\
    &\qquad+\frac{1}{n}\mathbb{E}\max_{\ell,\ell'}\bigg\|
      \frac{\sum_{j=1}^{q}F_{j\ell}}{n^{1/2}}\bigg\|_{2}
      \bigg\|\frac{\sum_{j=1}^{q}F_{j\ell'}}{n^{1/2}}\bigg\|_{2}\\
    &\leq\frac{n-1}{n}rC_0\mathbb{E}\max_{\ell}
      \sum_{j\in J_{\ell}}\bigg\|\frac{F_{j\ell}
      -\tU\bar{\gama}_{j\ell}}{n^{1/2}}\bigg\|_2\\
    &\qquad+\frac{n-1}{n}\mathbb{E}\max_{\ell}\bigg(
      \sum_{j\in J_{\ell}}\bigg\|\frac{\tU\gama_{\ell}
      -\sum_{j=1}^{q}F_{j\ell}}{n^{1/2}}\bigg\|_2
      +rC_0\bigg)\max_{\ell'}\bigg\|\frac{\sum_{j=1}^{q}
      F_{j\ell'}-\tU\gama_{\ell'}}{n^{1/2}}\bigg\|\\
    &\qquad+\frac{r^{2}C_0^{2}}{n}.
  \end{align*}
  Finally, we obtain
  \begin{align*}
    &\bigg\|\mathbb{E}\bigg(\frac{F^{\T}F}{n}
      -\frac{\Gama^{\T}\U^{\T}\U\Gama}{n}\bigg)\bigg\|_{\infty}\\
    &\leq\frac{n-1}{n}r^{2}C_0\sup_{z\in[a,b]}\bigg|
      f_{j\ell}(z)-\sum_{k=1}^{m_n}\bar{\gamma}_{kj\ell}
      \phi_{k}(z)\bigg|+\frac{r^{2}C_0^{2}}{n}\\
    &\quad+\frac{n-1}{n}\bigg\{r\sup_{z\in[a,b]}\bigg|
      f_{j\ell}(z)-\sum_{k=1}^{m_n}\bar{\gamma}_{kj\ell}
      \phi_{k}(z)\bigg|+rC_0\bigg\}\bigg\{r\sup_{z\in[a,b]}
      \bigg|f_{j\ell}(z)-\sum_{k=1}^{m_n}
      \bar{\gamma}_{jk\ell}\phi_{k}(z)\bigg|\bigg\},
  \end{align*}
  where we apply the Cauchy--Schwarz and triangular inequalities. Now by Lemma
  \ref{lem:B_spline_approx}, we have
  \begin{align*}
    &\bigg\|\mathbb{E}\bigg(\frac{F^{\T}F}{n}
      -\frac{\Gama^{\T}\U^{\T}\U\Gama}{n}\bigg)\bigg\|_{\infty}\\
    &\leq\frac{n-1}{n}r^{2}C_0C_Lm_n^{-d}
      +\frac{n-1}{n}\big(rC_Lm_n^{-d}+rC_0\big)
      rC_Lm_n^{-d}+\frac{r^{2}C_0^{2}}{n}\\
    &\leq3\frac{n-1}{n}r^{2}C_0C_Lm_n^{-d}
      +\frac{r^{2}C_0^{2}}{n}.
  \end{align*}
  The above inequality holds when $C_Lm_n^{-d}\leq C_0$. It then follows that
  \begin{align*}
    &\bigg|\theta_\ell^{\T}\mathbb{E}\bigg(\frac{F^{\T}F}{n}
      \bigg)\theta_\ell-\theta_\ell^{\T}\mathbb{E}\bigg(
      \frac{\Gama^{\T}\U^{\T}\U\Gama}{n}\bigg)\theta_\ell\bigg|
      \leq\|\theta_\ell\|_1^{2}\bigg\|\mathbb{E}\bigg(
      \frac{F^{\T}F}{n}-\frac{\Gama^{\T}\U^{\T}\U\Gama}{n}
      \bigg)\bigg\|_{\infty}\\
    &\leq3m_{\Omega}^{2}\frac{n-1}{n}r^{2}C_0C_Lm_n^{-d}
      +m_{\Omega}^{2}\frac{r^{2}C_0^{2}}{n},
  \end{align*}
  from which we know
  \begin{align*}
    \sigma_0\bigg\{\theta_\ell^{\T}\mathbb{E}\bigg(
    \frac{\Gama^{\T}\U^{\T}\U\Gama}{n}\bigg)
    \theta_\ell\bigg\}^{1/2}\rightarrow\omega_\ell.
  \end{align*}
  Next, we show, for a consistent estimator $\widehat{\omega}_{\ell}$,
  \begin{align*}
    n^{1/2}(\tilde{\bata}_{\ell}-\bata_{\ell})/
    \widehat{\omega}_\ell\rightsquigarrow\mathcal{N}(0,1).
  \end{align*}
  Since $\omega_{\ell}$ is lower bounded by a constant when $n$ is large, for a
  consistent estimator $\widehat{\omega}_{\ell}$, we know
  $\widehat{\omega}_\ell=\Omega_{p}(1)$. Therefore, we have
  \begin{align*}
    \left|\frac{\omega_\ell}{\widehat{\omega}_\ell}-1\right|=o_{p}(1).
  \end{align*}
  By Slutsky's theorem, we have
  \begin{align*}
    n^{1/2}(\tilde{\bata}_{\ell}-\bata_{\ell})/\widehat{\omega}_\ell
    =\frac{n^{1/2}(\tilde{\bata}_{\ell}-\bata_{\ell})}{\omega_\ell}
    \frac{\omega_\ell}{\widehat{\omega}_\ell}\rightsquigarrow
    \mathcal{N}(0,1).
  \end{align*}
\end{proof}

\subsection*{Proof of Lemma \ref{lem:consis_omega} }

\begin{proof}
  Note that
  \begin{align*}
    \widehat{\omega}_\ell^{2}=\widehat{\sigma}_0^{2}
    \widehat{\theta}_\ell^{\T}\frac{\widehat{\Gama}^{\T}
    \U^{\T}\U\widehat{\Gama}}{n}\widehat{\theta}_\ell.
  \end{align*}
  We first prove $\widehat{\sigma}_0$ is a consistent estimator of $\sigma_0$ by
  showing $\widehat{\sigma}_0^{2}\rightarrow\sigma_0^{2}$ in probability. Note
  that
  \begin{align*}
    &\big|\widehat{\sigma}^{2}_0-\sigma_0^{2}\big|
      =\bigg|\frac{1}{n}\big\|\Y-\X\widehat{\bata}
      \big\|_2^{2}-\sigma_0^{2}\bigg|
      =\bigg|\frac{1}{n}\big\|\X\bata+\yita
      -\X\widehat{\bata}\big\|_2^{2}-\sigma_0^{2}\bigg|\\
    &=\bigg|\frac{1}{n}\|\yita\|_2^{2}+\frac{1}{n}
      \big\|\X\bata-\X\widehat{\bata}\big\|_2^{2}
      +\frac{2}{n}\yita^{\T}\big(\X\bata
      -\X\widehat{\bata}\big)-\sigma_0^{2}\bigg|\\
    &\leq\bigg|\frac{1}{n}\|\yita\|_2^{2}
      -\sigma_0^{2}\bigg|+\frac{\big\|\X\bata
      -\X\widehat{\bata}\big\|_2^{2}}{n}
      +\frac{2\|\yita\|_2}{n^{1/2}}\frac{\big\|
      \X\bata-\X\widehat{\bata}\big\|_2}{n^{1/2}}\\
    &\leq\bigg|\frac{1}{n}\|\yita\|_2^{2}
      -\sigma_0^{2}\bigg|+\frac{64}{\kappa^{2}}s\mu^{2}
      +\frac{2\|\yita\|_2}{n^{1/2}}\frac{8}{\kappa}s^{1/2}\mu,
  \end{align*}
  which holds with probability at least $1-234(pqm)^{-1}$ by applying Theorem
  \ref{thm:second_stage_consistency}. From Theorem
  \ref{thm:second_stage_consistency}, we have
  \begin{align*}
    \bigg|\frac{1}{n}\left\|\yita\right\|_2^{2}
    -\sigma_0^{2}\bigg|\leq 4\sigma_0^{2}
    \bigg\{\frac{\log(pqm)}{n}\bigg\}^{1/2}
  \end{align*}
  with probability at least $1-2(pqm)^{-2}$, which implies
  \begin{align*}
    \left|\widehat{\sigma}^{2}_0-\sigma_0^{2}\right|
    \leq4\sigma_0^{2}\bigg\{\frac{\log(pqm)}{n}\bigg\}^{1/2}
    +\frac{64}{\kappa^{2}}s\mu^{2}+2\sqrt{5}
    \sigma_0\frac{8}{\kappa}s^{1/2}\mu
  \end{align*}
  holds with probability at least $1-236(pqm)^{-1}$ when $\log(pqm_n)/n\leq 1$.
  Then we consider
  \begin{align*}
    &\bigg|\widehat{\theta}_\ell^{\T}\frac{\widehat{\Gama}^{\T}
      \U^{\T}\U\widehat{\Gama}}{n}\widehat{\theta}_\ell
      -\theta_\ell^{\T}\frac{\Gama^{\T}\mathbb{E}
      (\U^{\T}\U)\Gama}{n}\theta_\ell\bigg|\\
    &=\bigg|\widehat{\theta}_\ell^{\T}\frac{\widehat{\Gama}^{\T}
      \U^{\T}\U\widehat{\Gama}}{n}\widehat{\theta}_\ell
      -\widehat{\theta}_{\ell}^{\T}\frac{F^{\T}F}{n}
      \widehat{\theta}_\ell+\widehat{\theta}_{\ell}^{\T}
      \frac{F^{\T}F}{n}\widehat{\theta}_\ell-\theta_\ell^{\T}
      \frac{\Gama^{\T}\mathbb{E}(\U^{\T}\U)\Gama}{n}\theta_\ell\bigg|\\
    &\leq\bigg|\widehat{\theta}_\ell^{\T}\frac{\widehat{\Gama}^{\T}
      \U^{\T}\U\widehat{\Gama}}{n}\widehat{\theta}_\ell
      -\widehat{\theta}_{\ell}^{\T}\frac{F^{\T}F}{n}
      \widehat{\theta}_\ell\bigg|+\bigg|\theta_{\ell}^{\T}\frac{F^{\T}F}{n}
      \theta_\ell-\theta_\ell^{\T}\frac{F^{\T}F}{n}\theta_\ell\\
    &\qquad+\widehat{\theta}_\ell^{\T}\frac{F^{\T}F}{n}
      \widehat{\theta}_\ell-\theta_\ell^{\T}\frac{\Gama^{\T}
      \mathbb{E}(\U^{\T}\U)\Gama}{n}\theta_\ell\bigg|\\
    &\leq\big\|\widehat{\theta}_{\ell}\big\|_1^{2}\bigg\|
      \frac{\widehat{\Gama}^{\T}\U^{\T}\U\widehat{\Gama}}{n}
      -\frac{F^{\T}F}{n}\bigg\|_{\infty}+\bigg|\theta_{\ell}^{\T}
      \frac{F^{\T}F}{n}\theta_{\ell}-\widehat{\theta}_\ell^{\T}
      \frac{F^{\T}F}{n}\widehat{\theta}_\ell\bigg|\\
    &\qquad+\bigg|\theta_{\ell}^{\T}\frac{F^{\T}F}{n}\theta_{\ell}
      -\theta_\ell^{\T}\frac{\Gama^{\T}\mathbb{E}(\U^{\T}\U)
      \Gama}{n}\theta_\ell\bigg|\\
    &\eqqcolon{}T_1+T_2+T_3
  \end{align*}
  For $T_1$, we know that
  \begin{align*}
    T_1\leq\|\theta\|_1^{2}\bigg\|\frac{\widehat{\Gama}^{\T}
    \U^{\T}\U\widehat{\Gama}}{n}-\frac{F^{\T}F}{n}
    \bigg\|_{\infty}\leq 35m_{\Omega}^{2}\lambda_{\max}rC_0
    \bigg\{\frac{2rm}{\rho}\bigg\}^{1/2}
  \end{align*}
  with probability at least $1-60(pqm)^{-1}$ when
  $\lambda_{\max}\{2m/\rho\}^{1/2}\leq\sqrt{r}C_0$, where we use the definition
  of $\widehat{\theta}$ in the first inequality and Lemma
  \ref{lem:first_inter_res_first} in the second inequality. For $T_2$, we have
  \begin{align*}
    T_2\leq \bigg|\theta_{\ell}^{\T}\frac{F^{\T}F}{n}\theta_{\ell}
    -\widehat{\theta}_\ell^{\T}\frac{F^{\T}F}{n}\widehat{\theta}_\ell\bigg|
    \leq\left\|\theta_\ell\right\|_1\big\|\theta_\ell-\widehat{\theta}_{\ell}
    \big\|_1\bigg\|\frac{F^{\T}F}{n}\bigg\|_{\infty}\leq
    m_{\Omega}r^{2}C_0^{2}\big\|\theta_\ell-\widehat{\theta}_{\ell}\big\|_1,
  \end{align*}
  where we use the fact that
  \begin{align*}
    \bigg|\bigg(\sum_{j=1}^{q}f_{j\ell}(Z_j)\bigg)
    \bigg(\sum_{j=1}^{q}f_{j\ell}(Z_j)\bigg)\bigg|\leq r^{2}C_0^{2}.
  \end{align*}
  Applying Lemma \ref{lem:l1}, we have
  \begin{align*}
    T_2\leq 2c_bm_{\Omega}s_{\Omega}r^{2}C_0^{2}(2m_{\Omega}\upsilon)^{1-b}
  \end{align*}
  with probability at least $1-62(pqm)^{-1}$, where
  $\upsilon=36m_{\Omega}\lambda_{\max}rC_0\{2rm/\rho\}^{1/2}, 0\leq b<1$. For
  $T_3$, we have
  \begin{align*}
    &T_{3}=\bigg|\theta_{\ell}^{\T}\bigg(\frac{F^{\T}F}{n}
      -\frac{\Gama^{\T}\mathbb{E}\left(\U^{\T}\U\right)
      \Gama}{n}\bigg)\theta_{\ell}\bigg|\\
    &\leq\|\theta_{\ell}\|_1^{2}\bigg\|\frac{F^{\T}F}{n}
      -\frac{\Gama^{\T}\mathbb{E}\left(\U^{\T}\U\right)
      \Gama}{n}\bigg\|_{\infty}\\
    &\leq{}m_{\Omega}^{2}\bigg(3\frac{n-1}{n}
      r^{2}C_LC_0m^{-d}+\frac{r^{2}C_0^{2}}{n}\bigg),
  \end{align*}
  which follows from Theorem \ref{thm:asym_normality}. Therefore, we have
  \begin{align*}
    &\bigg|\widehat{\theta}_\ell^{\T}\frac{\widehat{\Gama}^{\T}
      \U^{\T}\U\widehat{\Gama}}{n}\widehat{\theta}_\ell
      -\theta_\ell^{\T}\frac{\Gama^{\T}\mathbb{E}(\U^{\T}\U)
      \Gama}{n}\theta_\ell\bigg|\\
    &\leq35m_{\Omega}\lambda_{\max}rC_0\bigg\{\frac{2rm}{\rho}\bigg\}^{1/2}
      +2c_bm_{\Omega}s_{\Omega}r^{2}C_0^{2}(2m_{\Omega}\upsilon)^{1-b}\\
    &\qquad+m_{\Omega}^{2}\bigg(3\frac{n-1}{n}r^{2}C_LC_0m_n^{-d}
      +\frac{r^{2}C_0^{2}}{n}\bigg).
  \end{align*}
  Now we have
  \begin{align*}
    &\big|\widehat{\omega}_{\ell}^{2}-\omega_{\ell}^{2}\big|
      =\bigg|\widehat{\sigma}_0^{2}\widehat{\theta}_\ell^{\T}
      \frac{\widehat{\Gama}^{\T}\U^{\T}\U\widehat{\Gama}}{n}
      \widehat{\theta}_\ell-\sigma_0^{2}\theta_{\ell}^{\T}
      \mathbb{E}\left(\Gama^{\T}\U^{\T}\U\Gama/n
      \right)\theta_{\ell}\bigg|\\
    &\leq\bigg|\widehat{\sigma}_0^{2}\widehat{\theta}_\ell^{\T}
      \frac{\widehat{\Gama}^{\T}\U^{\T}\U\widehat{\Gama}}{n}
      \widehat{\theta}_\ell-\sigma_0^{2}\widehat{\theta}_\ell^{\T}
      \frac{\widehat{\Gama}^{\T}\U^{\T}\U\widehat{\Gama}}{n}
      \widehat{\theta}_\ell+\sigma_0^{2}\widehat{\theta}_\ell^{\T}
      \frac{\widehat{\Gama}^{\T}\U^{\T}\U\widehat{\Gama}}{n}
      \widehat{\theta}_\ell-\sigma_0^{2}\theta_{\ell}^{\T}
      \mathbb{E}\left(\Gama^{\T}\U^{\T}\U\Gama/n
      \right)\theta_{\ell}\bigg|\\
    &\leq\big|\widehat{\sigma}_0^{2}-\sigma_0^{2}\big|\bigg|
      \widehat{\theta}_\ell^{\T}\frac{\widehat{\Gama}^{\T}\U^{\T}
      \U\widehat{\Gama}}{n}\widehat{\theta}_\ell\bigg|
      +\bigg|\widehat{\theta}_\ell^{\T}\frac{\widehat{\Gama}^{\T}
      \U^{\T}\U\widehat{\Gama}}{n}\widehat{\theta}_\ell
      -\theta_{\ell}^{\T}\mathbb{E}\left(\Gama^{\T}\U^{\T}
      \U\Gama/n\right)\theta_{\ell}\bigg|\sigma_0^{2}\\
    &\leq\bigg[4\sigma_0^{2}\bigg\{\frac{\log(pqm)}{n}
      \bigg\}^{1/2}+\frac{64}{\kappa^{2}}s\mu^{2}
      +2\sqrt{5}\sigma_0\frac{8}{\kappa}s^{1/2}\mu\bigg]
      \bigg|\widehat{\theta}_\ell^{\T}\frac{\widehat{\Gama}^{\T}
      \U^{\T}\U\widehat{\Gama}}{n}\widehat{\theta}_{\ell}\bigg|\\
    &\qquad+\sigma_0^{2}\bigg[35m_{\Omega}\lambda_{\max}rC_0
      \bigg\{\frac{2rm}{\rho}\bigg\}^{1/2}
      +2c_km_{\Omega}s_{\Omega}r^{2}C_0^{2}(2m_{\Omega}\upsilon)^{1-b}\\
    &\qquad+m_{\Omega}^{2}\bigg\{3\frac{n-1}{n}r^{2}
      C_LC_0m_n^{-d}+\frac{r^{2}C_0^{2}}{n}\bigg\}\bigg].
  \end{align*}
  We also have
  \begin{align*}
    &\left|\widehat{\theta}_\ell^{\T}\frac{\widehat{\Gama}^{\T}
      \U^{\T}\U\widehat{\Gama}}{n}\widehat{\theta}_{\ell}\right|
      \leq\left\|\widehat{\theta}_\ell\right\|_{1}^{2}
      \left\|\frac{\widehat{\Gama}^{\T}\U^{\T}\U\widehat{\Gama}}{n}
      -\frac{F^{\T}F}{n}+\frac{F^{\T}F}{n}\right\|_{\infty}\\
    &\leq{}m_{\Omega}^{2}\left(\left\|\frac{\widehat{\Gama}^{\T}
      \U^{\T}\U\widehat{\Gama}}{n}-\frac{F^{\T}F}{n}
      \right\|_{\infty}+\left\|\frac{F^{\T}F}{n}
      \right\|_{\infty}\right)\\
    &\leq{}35m_{\Omega}^{2}\lambda_{\max}rC_0\bigg\{
      \frac{2rm}{\rho}\bigg\}^{1/2}+m_{\Omega}r^{2}C_0^{2}.
  \end{align*}
  Therefore, we get
  \begin{align*}
    \big|\widehat{\omega}^{2}_{\ell}-\omega^{2}_{\ell}\big|
    &\leq\bigg[4\sigma_0^{2}\bigg\{\frac{\log(pqm)}{n}\bigg\}^{1/2}
      +\frac{64}{\kappa^{2}}s\mu^{2}+2\sqrt{5}\sigma_0
      \frac{8}{\kappa}s^{1/2}\mu\bigg]\\
    &\qquad\cdot\bigg[35m_{\Omega}^{2}\lambda_{\max}rC_0
      \bigg\{\frac{2rm}{\rho}\bigg\}^{1/2}
      +m_{\Omega}r^{2}C_0^{2}\bigg]\\
    &\qquad+\sigma_0^{2}\bigg[35m_{\Omega}\lambda_{\max}rC_0
      \bigg\{\frac{2rm}{\rho}\bigg\}^{1/2}
      +2c_bm_{\Omega}s_{\Omega}r^{2}C_0^{2}
      (2m_{\Omega}\upsilon)^{1-b}\\
    &\qquad\qquad+m_{\Omega}^{2}\bigg(3\frac{n-1}{n}r^{2}
      C_LC_0m_n^{-d}+\frac{r^{2}C_0^{2}}{n}\bigg)\bigg].
  \end{align*}
\end{proof}

\section*{Section E: Control of Remainder Terms}

We derive specific conditions to make sure the remainder terms satisfy
$\|R_k\|_{\infty}=o_p(1)$, $k=1,2,3,4$. We first find the probability of the
event $\|\Omega\widehat{\Sigma}_d-I\|_{\infty}=o_p(1)$.
\begin{lemma}\label{lem:precision_est}
Suppose Assumption \ref{assu:precision_mat} holds, then
\begin{align*}
  \big\|\Omega\widehat{\Sigma}_d-I\big\|_{\infty}\leq
  m_{\Omega}36\lambda_{\max}rC_0\bigg(\frac{2rm}{\rho}\bigg)^{1/2}
\end{align*}
holds with probability at least $1-62(pqm)^{-1}$, when
$\lambda_{\max}\{2m/\rho\}^{1/2}\leq r^{1/2}C_0$.
\end{lemma}
\begin{proof}
  Note that
  \begin{align*}
    &\big\|\Omega\widehat{\Sigma}_d-I\big\|_{\infty}
    =\bigg\|\Omega\bigg\{\frac{\widehat{\Gama}^{\T}\U^{\T}
      \U\widehat{\Gama}}{n}-\mathbb{E}\bigg(\frac{F^{\T}F}
      {n}\bigg)\bigg\}\bigg\|_{\infty}\\
    &\leq\|\Omega\|_1\bigg\|\frac{\widehat{\Gama}^{\T}\U^{\T}
      \U\widehat{\Gama}}{n}-\mathbb{E}\bigg(
      \frac{F^{\T}F}{n}\bigg)\bigg\|_{\infty}\\
    &\leq m_{\Omega}\bigg\|\Sigma_{f}-\frac{F^{\T}F}{n}
      +\frac{F^{\T}F}{n}-\frac{\widehat{\X}^{\T}
      \widehat{\X}}{n}\bigg\|_{\infty}\\
    &\leq{}m_{\Omega}\bigg(\bigg\|\Sigma_{f}-\frac{F^{\T}F}{n}
      \bigg\|_{\infty}+\bigg\|\frac{F^{\T}F}{n}
      -\frac{\widehat{\X}^{\T}\widehat{\X}}{n}
      \bigg\|_{\infty}\bigg).
  \end{align*}
  From Lemma \ref{lem:re_sec_stage}, we have
  \begin{align*}
    \bigg\|\Sigma_{f}-\frac{F^{\T}F}{n}\bigg\|_{\infty}\leq
    4r^{2}C_0^{2}\bigg\{\frac{\log(pqm)}{n}\bigg\}^{1/2}
  \end{align*}
  with probability at least $1-2p^{2}(pqm)^{-4}$. Also, from Lemma
  \ref{lem:first_inter_res_first}, we have
  \begin{align*}
    \bigg\|\frac{F^{\T}F}{n}-\frac{\widehat{\X}^{\T}\widehat{\X}}{n}
    \bigg\|_{\infty}\leq 35\lambda_{\max}rC_0\bigg(\frac{2rm}{\rho}\bigg)
  \end{align*}
  with probability at least $1-60(pqm)^{-1}$, when
  $\lambda_{\max}(2m/\rho)^{1/2}\leq r^{1/2}C_0$. Therefore, we have
  \begin{align*}
    \big\|\Omega\widehat{\Sigma}_d-I\big\|_{\infty}\leq
    m_{\Omega}36\lambda_{\max}rC_0\bigg(\frac{2rm}{\rho}\bigg)^{1/2}
  \end{align*}
  with probability at least $1-62(pqm)^{-1}$ when
  $\lambda_{\max}(2m/\rho)^{1/2}\leq r^{1/2}C_0$.
\end{proof}

By Lemma \ref{lem:precision_est}, $\upsilon$ can be chosen as
$36m_{\Omega}\lambda_{\max}rC_0\sqrt{2rm/\rho}$. The next result provides a
probabilistic bound for $\|\U^{\T}\yita/n\|_{\infty}$.

\begin{lemma}\label{lem:spline_residual}
  \begin{align*}
    \pr\bigg[\frac{\|\U^{\T}\yita\|_{\infty}}{n}\leq
    4\big\{\log(pqm/n)\big\}^{1/2}\bigg]\geq1-(pqm)^{-1}.
  \end{align*}
\end{lemma}
\begin{proof}
  Write $\|\U^{\T}\yita/n\|_{\infty}$ as
  $\left\|\U^{\T}\yita\right\|_{\infty}=\max_{j,k}
  \left|\U_{jk}^{\T}\yita\right|$, where $\U_{jk}$ is the $k$th column of spline
  matrix $\U_{j}$. Recall the Gaussian tail bound
  \begin{align*}
    \pr\left(|X|\geq t \right)\leq2\bigg(
    \frac{2}{\pi}\bigg)^{1/2}\frac{\exp(-t^{2}/2)}{t}.
  \end{align*}
  Now we have
  \begin{align*}
    &\pr\bigg(\frac{\|\U^{\T}\yita\|_{\infty}}{n}\geq t\bigg)
    \leq{}qm_n\pr\bigg(\frac{\big|\U_{jk}^{\T}\yita\big|}{n}\geq{}t\bigg)\\
    &=qm_n\pr\Bigg[\frac{\big|\U_{jk}^{\T}\yita\big|}{\sigma_0
      \big\{\sum_{i=1}^{n}\psi_{k}^{2}(Z_{ij})\big\}^{1/2}}\geq
      \frac{nt}{\sigma_0\big\{\sum_{i=1}^{n}\psi_{k}^{2}(Z_{ij})\big\}^{1/2}}\Bigg]\\
    &\leq2qm_n\bigg(\frac{2}{\pi}\bigg)^{1/2}\exp\bigg\{-\frac{n^{2}t^{2}}
      {2\sigma_0^{2}\sum_{i=1}^{n}\psi^{2}_{k}(Z_{ij})}\bigg\}\frac{\sigma_0
      \big\{\sum_{i=1}^{n}\psi_{k}^{2}(Z_{ij})\big\}^{1/2}}{nt}.
  \end{align*}
  Since $|\psi_{k}(Z_{ij})|\leq 2$, setting $t=4\sigma_0\{\log(pqm/n)\}^{1/2}$,
  we obtain
  \begin{align*}
    \pr\bigg[\frac{\|\U^{\T}\yita\|_{\infty}}{n}\geq
    4\sigma_0\{\log(pqm/n)\}^{1/2}\bigg]\leq{}qm(pqm)^{-2}.
  \end{align*}
\end{proof}

\begin{lemma}[Lemma A.1 in \citeauthor{gold2020inference},
  \citeyear{gold2020inference}]\label{lem:l_infty}
  Suppose (i) $\|\Omega\|_{1}$ is bounded above by a constant
  $m_{\Omega}<\infty$, and (ii) $\widehat{\Omega}$ is an estimate of $\Omega$
  with rows $\widehat{\theta}_\ell$ obtained as solutions to
  \eqref{eq:opt_precision}. Then
  \begin{align*}
    \|\widehat{\Omega}-\Omega\|_{\infty}\leq 2m_{\Omega}\upsilon
  \end{align*}
  holds with probability at least $1-62(pqm)^{-1}$.
\end{lemma}
\begin{proof}
  Under the conditions of this Lemma, we know that
  \begin{align*}
    \big\|\Omega\widehat{\Sigma}_d-I\big\|_{\infty}\leq
    \upsilon,\ \big\|\widehat{\Omega}\widehat{\Sigma}_{d}
    -I\big\|_{\infty}\leq\upsilon.
  \end{align*}
  Then, we have
  \begin{align*}
    &\Omega-\widehat{\Omega}=\bigg\{I-\widehat{\Omega}\mathbb{E}
      \bigg(\frac{F^{\T}F}{n}\bigg)\bigg\}\Omega
      =\bigg[I+\widehat{\Omega}\bigg\{\widehat{\Sigma}_d
      -\mathbb{E}\bigg(\frac{F^{\T}F}{n}\bigg)\bigg\}
      -\widehat{\Omega}\widehat{\Sigma}_d\bigg]\Omega\\
    &=\big(I-\widehat{\Omega}\widehat{\Sigma}_d\big)\Omega
      -\widehat{\Omega}\bigg\{\mathbb{E}\bigg(\frac{F^{\T}F}
      {n}\bigg)-\widehat{\Sigma}_d\bigg\}\Omega
      =\big(I-\widehat{\Omega}\widehat{\Sigma}_d\big)\Omega
      -\widehat{\Omega}\big(I-\widehat{\Sigma}_d\Omega\big).
  \end{align*}
  By H\"older's inequality, we have
  \begin{align*}
    \widehat{\Omega}\big(I-\widehat{\Sigma}_d\Omega\big)
    \leq\big\|I-\widehat{\Sigma}_d\Omega\big\|_{\infty}
    \big\|\widehat{\Omega}^{\T}\big\|_{1}\leq
    \big\|I-\widehat{\Sigma}_d\Omega\big\|_{\infty}
    \big\|\Omega^{\T}\big\|_{1}\leq m_{\Omega}\upsilon
  \end{align*}
  and
  \begin{align*}
    \big\|\big(I-\widehat{\Omega}\widehat{\Sigma}_d\big)
    \Omega\big\|_{\infty}\leq\|\Omega\|_1
    \big\|I-\widehat{\Omega}\widehat{\Sigma}_d\big\|_{\infty}
    \leq{}m_\Omega\upsilon.
  \end{align*}
  Therefore, we have
  \begin{align*}
    \big\|\Omega-\widehat{\Omega}\big\|_{\infty}\leq 2m_{\Omega}\upsilon.
  \end{align*}
  We have shown in Lemma~\ref{lem:precision_est} that
  \begin{align*}
    \big\|\Omega\widehat{\Sigma}_d-I\big\|_{\infty}\leq
    36m_{\Omega}\lambda_{\max}rC_0\bigg(\frac{2rm}{\rho}\bigg)^{1/2}
  \end{align*}
  with probability at least $1-62(pqm)^{-1}$. Therefore, the proof is complete
  by choosing $\upsilon=36m_{\Omega}\lambda_{\max}rC_0(2rm/\rho)^{1/2}$.
\end{proof}

\begin{lemma}[Lemma A.2 in \citeauthor{gold2020inference},
  \citeyear{gold2020inference}]\label{lem:l1}
  Suppose, in addition to the conditions of Lemma \ref{lem:l_infty},
  $\Omega\in\mathcal{U}(m_\Omega,b,s_{\Omega})$. Then,
  \begin{align*}
    \|\widehat{\theta}_\ell-\theta_\ell\|_1\leq
    2c_b(2m_\Omega\upsilon)^{1-b}s_\Omega
  \end{align*}
  for each $\ell\in\{1,\ldots,p\}$, where $c_b=1+2^{1-b}+3^{1-b}$.
\end{lemma}
\begin{proof}
  See the proof of Theorem 6 in \cite{cai2011constrained}.
\end{proof}

Now we examine each remainder term $R_k$ in the subsequent subsections.

\subsection*{Control of $R_1$}

\begin{lemma}\label{lem:bound_r1}
  Suppose the conditions of Theorem \ref{thm:second_stage_consistency} and Lemma
  \ref{lem:l1} hold. If
  \begin{align*}
    n^{1/2}2c_b(2m_\Omega\upsilon)^{1-b}s_{\Omega}\bigg[4\sqrt{5}
    \sigma_0rC_Lm_n^{-d}+14C_0\sigma_0\bigg\{\frac{r^{2}\log(pqm)}
    {n}\bigg\}^{1/2}\bigg]=o(1),
  \end{align*}
  then $\|R_1\|_{\infty}=o(1)$ with probability at least $1-5(pqm)^{-1}$.
\end{lemma}
\begin{proof}
  Note that
  \begin{align*}
    &R_1=(\widehat{\Omega}-\Omega)\D^{\T}\yita/n^{1/2}
      =(\widehat{\Omega}-\Omega)(\Gama^{\T}\U^{\T}-F^{\T}
      +F^{\T})\yita/n^{1/2}\\
    &\leq\|\widehat{\Omega}^{\T}-\Omega^{\T}\|_1\big(
      \|(\Gama^{\T}\U^{\T}-F^{\T})\yita/n^{1/2}\|_{\infty}
      +\|F^{\T}\yita/n^{1/2}\|_{\infty}\big)\\
    &=n^{1/2}\|\widehat{\Omega}^{\T}-\Omega^{\T}\|_1
      \big(\|(\Gama^{\T}\U^{\T}-F^{\T})\yita/n\|_{\infty}
      +\|F^{\T}\yita/n\|_{\infty}\big).
  \end{align*}
  As is shown in the proof of Theorem \ref{thm:second_stage_consistency},
  $\|F^{\T}\yita/n\|_{\infty}\leq2C_0\sigma_0\{r^{2}\log(pqm)/n\}^{1/2}$ holds
  with probability at least $1-p(pqm)^{-2}$. Applying Lemma
  \ref{lem:approximation_error_bound} and the $\chi^{2}$ concentration bound, we
  have
  \begin{align*}
    &\|(\Gama^{\T}\U^{\T}-F^{\T})\yita/n\|_{\infty}
      \leq\frac{1}{n^{1/2}}\max_{\ell}\bigg\|\U
      \bar{\gama}_{\ell}-\sum_{j=1}^{q}F_{j\ell}\bigg\|_{2}
      \max_{\ell'}\frac{\|\yita\|_2}{n^{1/2}}\\
    &\leq4\sqrt{5}\sigma_0rC_Lm_n^{-d}
      +4r\sigma_0\{5C_0^{2}\log(pqm)/n\}^{1/2}
  \end{align*}
  with probability at least $1-4r(pqm)^{-2}$. It follows that
  \begin{align*}
    &R_1\leq{}n^{1/2}\|\widehat{\Omega}^{\T}
      -\Omega^{\T}\|_1\|F^{\T}\yita/n\|_{\infty}\\
    &\leq{}n^{1/2}2c_b(2m_\Omega\upsilon)^{1-b}s_{\Omega}
      \bigg[4\sqrt{5}\sigma_0rC_Lm_n^{-d}+14C_0\sigma_0
      \bigg\{\frac{r^{2}\log(pqm)}{n}\bigg\}^{1/2}\bigg]
  \end{align*}
  holds with probability at least $1-5(pqm)^{-1}$.
\end{proof}

\subsection*{Control of $R_2$}

\begin{lemma}
  Under the same conditions of Lemma \ref{lem:bound_r1}, if
  \begin{align*}
    m_{\Omega}128r\lambda_{\max}\{\log(pqm)\}^{1/2}
    \frac{m_n^{3/2}}{\rho}=o(1),
  \end{align*}
  then $\|R_2\|_{\infty}=o(1)$ with probability at least $1-19(pqm)^{-1}$.
\end{lemma}
\begin{proof}
  Note that
  \begin{align*}
    &\|R_2\|_{\infty}=\widehat{\Omega}\big(\widehat{\Gama}^{\T}
      \U^{\T}-\Gama^{\T}\U^{\T}\big)\yita/n^{1/2}\\
    &\leq{}n^{1/2}\|\widehat{\Omega}^{\T}\|_1
      \|\widehat{\Gama}-\Gama\|_1\|\U^{\T}\yita/n\|_{\infty}\\
    &\leq{}n^{1/2}m_{\Omega}\max_{\ell=1,\ldots,p}
      \|\widehat{\gama}_\ell-\bar{\gama}_\ell\|_{1}
      \|\U^{\T}\yita/n\|_{\infty}.
  \end{align*}
  Recall from the proof of Theorem \ref{thm:estimation_first_stage}
  \begin{align*}
    \max_{\ell=1,\ldots,p}\|\widehat{\gama}_\ell-\bar{\gama}_\ell\|_{1}
    \leq\max_{\ell}m^{1/2}\sum_{j=1}^{q}\|\widehat{\gama}_{j\ell}
    -\bar{\gama}_{j\ell}\|_{2}\leq{}m^{1/2} 32r\lambda_{\max}\frac{m}{\rho}
  \end{align*}
  holds with probability at least $1-18(pqm)^{-1}$. Then, by Lemma
  \ref{lem:spline_residual}, we obtain
  \begin{align*}
    \|\U^{\T}\yita/n\|_{\infty}\leq4\bigg\{\frac{\log(pqm)}{n}\bigg\}^{1/2}
  \end{align*}
  with probability at least $1-(pqm)^{-1}$. It follows that
  \begin{align*}
    \|R_2\|_{\infty}\leq{}m_{\Omega}128r\lambda_{\max}
    \{\log(pqm)\}^{1/2}\frac{m_n^{3/2}}{\rho}
  \end{align*}
  holds with with probability at least $1-19(pqm)^{-1}$.
\end{proof}

\subsection*{Control of $R_3$}

\begin{lemma}
  Under the same conditions of Lemma \ref{lem:bound_r1}, if
  \begin{align*}
    n^{1/2}m_{\Omega}\lambda_{\max}\bigg(\frac{2rm}{\rho}\bigg)^{1/2}
    \big(30rC_0+16\max_{\ell}\sigma_\ell\big)\frac{64}{\kappa^{2}}s\mu=o(1),
  \end{align*}
  then $\|R_3\|_{\infty}=o(1)$ with probability at least $1-277(pqm)^{-1}$.
\end{lemma}

\begin{proof}
  We first apply H\"{o}lder's inequality to get
  \begin{align*}
    \|R_3\|_{\infty}\leq n^{1/2}\|\widehat{\Omega}^{\T}\|_1
    \|\widehat{\D}^{\T}(\X-\widehat{\D})/n\|_{\infty}
    \|\widehat{\bata}-\bata\|_1.
  \end{align*}
  Note that
  \begin{align*}
    &\widehat{\D}^{\T}(\X-\widehat{\D})/n
      =\widehat{\D}^{\T}(F+\E-\widehat{\D})/n\\
    &=\widehat{\D}^{\T}(F-\widehat{\D})/n+\widehat{\D}^{\T}\E/n.
  \end{align*}
  For the first term on the right-hand side above, we have
  \begin{align*}
    &\|\widehat{\Gama}^{\T}\U^{\T}(F-\widehat{\D})/n\|_{\infty}
      \leq\|F^{\T}(F-\widehat{\D})/n\|_{\infty}+\|(\widehat{\Gama}^{\T}
      \U^{\T}-F^{\T})(F-\widehat{\D})/n\|_{\infty}\\
    &\leq\max_{\ell,\ell'}\frac{\big\|\sum_{j=1}^{q}F_{j\ell'}
      \big\|_{2}}{n^{1/2}}\bigg\|\frac{\sum_{j=1}^{q}F_{j\ell}
      -\U\widehat{\gama}_{\ell}}{n^{1/2}}\bigg\|_2
      +\max_{\ell}\bigg\|\frac{\sum_{j=1}^{q}F_{j\ell}
      -\U\widehat{\gama}_{\ell}}{n^{1/2}}\bigg\|_2^{2}\\
    &\leq{}rC_05\lambda_{\max}r^{1/2}\bigg(\frac{2m}{\rho}\bigg)^{1/2}
      +25\lambda_{\max}^{2}r\frac{2m}{\rho}
      \leq30\lambda_{\max}r^{3/2}C_0\bigg(\frac{2m}{\rho}\bigg)^{1/2},
  \end{align*}
  which holds with probability at least $1-20(pqm)^{-1}$ when
  $\lambda_{\max}r^{1/2}(2m/\rho)^{1/2}\leq C_0$. For the second term, we have
  \begin{align*}
    &\|\widehat{\D}^{\T}\E/n\|_{\infty}
      =\|\widehat{\Gama}^{\T}\U^{\T}\E/n\|_{\infty}\\
    &\leq\|(\widehat{\Gama}^{\T}\U^{\T}-F^{\T})\E/n\|_{\infty}
      +\|F^{\T}\E/n\|_{\infty}\\
    &\leq\max_{\ell,\ell'}\bigg\|\frac{\sum_{j=1}^{q}F_{j\ell}
      -\U\widehat{\gama}_{\ell}}{n^{1/2}}\bigg\|_2\bigg\|
      \frac{\varepsilon_{\ell'}}{n^{1/2}}\bigg\|_2+\|F^{\T}\E/n\|_{\infty}\\
    & \leq 5\sqrt{5}\lambda_{\max}\max_{\ell'}\sigma_{\ell'}r^{1/2}
      \bigg(\frac{2m}{\rho}\bigg)^{1/2}+\|F^{\T}\E/n\|_{\infty}\\
    &\leq5\sqrt{5}\max_{\ell'}\sigma_{\ell'}\lambda_{\max}r^{1/2}
      \bigg(\frac{2m}{\rho}\bigg)^{1/2}+\max_{\ell'}
      \sigma_{\ell'}C_0\{8r^{2}\log(pqm)/n\}^{1/2}\\
    &\leq16\max_{\ell}\sigma_\ell\lambda_{\max}r^{1/2}
      \bigg(\frac{2m}{\rho}\bigg)^{1/2}
  \end{align*}
  with probability at least $1-23(pqm)^{-1}$, where we use the arguments in the
  proof of Theorem \ref{thm:second_stage_consistency} when bounding
  $\|T_3\|_{\infty}$ and $\|T_4\|_{\infty}$. It follows that
  \begin{align*}
    \big\|\widehat{\D}^{\T}(\X-\widehat{\D})/n\big\|_{\infty}\leq
    \lambda_{\max}\bigg(\frac{2rm}{\rho}\bigg)^{1/2}
    \big(30rC_0+16\max_{\ell}\sigma_\ell\big)
  \end{align*}
  holds with probability at leasst $1-43(pqm)^{-1}$. Finally, we have
  \begin{align*}
    &\|R_3\|_{\infty}\leq n^{1/2}\|\Omega^{\T}\|_1\|\widehat{\D}^{\T}
      (\D-\widehat{\D})/n\|_{\infty}\|\widehat{\bata}-\bata\|_1\\
    &=n^{1/2}m_{\Omega}\lambda_{\max}\bigg(\frac{2rm}{\rho}\bigg)^{1/2}
      \big(30rC_0+16\max_{\ell}\sigma_\ell\big)\frac{64}{\kappa^{2}}s\mu
  \end{align*}
  with probability at least $1-277(pqm)^{-1}$.
\end{proof}

\subsection*{Control of $R_4$}

\begin{lemma}
  Under the same conditions of Lemma \ref{lem:bound_r1}, if
  \begin{align*}
    n^{1/2}36m_{\Omega}\lambda_{\max}rC_{0}\bigg(\frac{2rm}
    {\rho}\bigg)^{1/2}\frac{64}{\kappa^{2}}s\mu=o(1),
  \end{align*}
  then $\|R_4\|_{\infty}=o(1)$ with probability at least $1-296(pqm)^{-1}$.
\end{lemma}
\begin{proof}
  By H\"{o}lder's inequality, we have
  \begin{align*}
    \|R_4\|_{\infty}\leq n^{1/2}\|\widehat{\Omega}
    \widehat{\Sigma}_d-I\|_\infty\|\widehat{\bata}-\bata\|_1.
  \end{align*}
  Applying Lemma \ref{lem:precision_est}, we obtain
  \begin{align*}
    \|R_4\|_{\infty}\leq n^{1/2}36m_{\Omega}\lambda_{\max}rC_{0}
    \bigg(\frac{2rm}{\rho}\bigg)^{1/2}\|\widehat{\bata}-\bata\|_1,
  \end{align*}
  which holds with probability at least $1-62(pqm)^{-1}$ when
  $\lambda_{\max}(2m/\rho)^{1/2}\leq r^{1/2}C_0$. Apply Theorem
  \ref{thm:second_stage_consistency} to get
  \begin{align*}
    \|R_4\|_{\infty}\leq n^{1/2}36m_{\Omega}\lambda_{\max}rC_{0}
    \bigg(\frac{2rm}{\rho}\bigg)^{1/2}\frac{64}{\kappa^{2}}s\mu,
  \end{align*}
  which holds with probability at least $1-296(pqm)^{-1}$.
\end{proof}

\section*{Section F: Additional Experiments}

In this section, we present more experiments to demonstrate the finite sample
performance of our proposed estimator when the dimensions of the treatment and
instrumental variables are large.

\begin{table}[h]
\centering
\caption{$L_1$ error of each method averaged over one hundred replications with standard
  deviation shown in parentheses for $p=100$.}
\begin{tabular}{ccccccc}
\toprule
Sample & \multicolumn{3}{c}{Linear} & \multicolumn{3}{c}{Nonlinear}\cr\cmidrule(lr){2-4}\cmidrule(lr){5-7}
size & Our method & 2SLS-L & PLS & Our method & 2SLS-L & PLS\\
\midrule
 100 & 1.15 (0.42)& 1.48 (0.62) & 0.75 (0.27) & 1.94 (0.98)& 1.55 (0.74) & 0.94 (0.31)\\
 \addlinespace[0.7mm]
200 & 0.62 (0.33)& 0.62 (0.32) &0.69 (0.28) & 1.06 (0.53)& 1.87 (0.81) & 0.84 (0.32)\\
\addlinespace[0.7mm]
400 &0.47 (0.26) & 0.41 (0.21)& 0.82 (0.34)& 0.65 (0.29) & 2.55 (0.95)& 1.06 (0.28)\\
\addlinespace[0.7mm]
600& 0.38 (0.18)& 0.35 (0.16)& 1.20 (0.51) & 0.45 (0.26)& 2.95 (1.11)& 1.27 (0.26)\\
\addlinespace[0.7mm]
800& 0.29 (0.14)& 0.28 (0.13)& 1.36 (0.42) & 0.38 (0.17)& 3.22 (1.52)& 1.38 (0.29)\\
\addlinespace[0.7mm]
1000& 0.23 (0.10) & 0.23 (0.10)& 1.44 (0.38)& 0.30 (0.13)& 3.49 (1.52)& 1.43 (0.26)\\
\addlinespace[0.7mm]
1200& 0.23 (0.10)& 0.24 (0.10)& 1.68 (0.37)& 0.27 (0.15)& 3.44 (1.28)& 1.46 (0.21)\\
\addlinespace[0.7mm]
1400& 0.19 (0.08)& 0.20 (0.09)& 1.75 (0.38)& 0.24 (0.13)& 3.83 (1.75)& 1.55 (0.21)\\
\bottomrule
\end{tabular}
\end{table}

\begin{table}[h]
\centering
\caption{$L_1$ error of each method averaged over one hundred replications with
  standard deviation shown in parentheses for $p=200$.}
\begin{tabular}{ccccccc}
\toprule
Sample & \multicolumn{3}{c}{Linear} & \multicolumn{3}{c}{Nonlinear}\cr\cmidrule(lr){2-4}\cmidrule(lr){5-7}
size & Our method & 2SLS-L & PLS & Our method & 2SLS-lL & PLS\\
\midrule
 100 & 1.14 (0.52)& 0.52 (0.76) & 0.80 (0.23) & 2.53 (1.47)& 1.21 (0.50) & 1.02 (0.31)\\
 \addlinespace[0.7mm]
200 & 0.65 (0.28)& 0.74 (0.38) &0.74 (0.38) & 0.99 (0.41)& 1.35 (0.80) & 0.84 (0.29)\\
\addlinespace[0.7mm]
400 &0.42 (0.20) & 0.42 (0.20)& 0.64 (0.23) & 0.62 (0.24) & 2.03 (1.02)& 0.93 (0.29)\\
\addlinespace[0.7mm]
600& 0.30 (0.16)& 0.28 (0.14)& 0.73 (0.31) & 0.43 (0.19)& 2.49 (1.16)& 1.09 (0.27)\\
\addlinespace[0.7mm]
800& 0.26 (0.11)& 0.25 (0.10)& 0.92 (0.35) & 0.35 (0.15)& 3.42 (1.75)& 1.25 (0.24)\\
\addlinespace[0.7mm]
1000& 0.21 (0.09) & 0.20 (0.08)& 1.14 (0.43)& 0.30 (0.15)& 3.42 (1.65)& 1.35 (0.22)\\
\addlinespace[0.7mm]
1200& 0.21 (0.10)& 0.21 (0.10)& 1.36 (0.47)& 0.29 (0.15)& 3.92 (2.32)& 1.39 (0.24)\\
\addlinespace[0.7mm]
1400& 0.18 (0.07)& 0.19 (0.08)& 1.58 (0.47)& 0.24 (0.12)& 4.12 (2.43)& 1.47 (0.21)\\
\bottomrule
\end{tabular}
\end{table}

\begin{table}[h]
\centering
\caption{$L_1$ error of each method averaged over one hundred replications with
  standard deviation shown in the parentheses when $p=400$.}
\begin{tabular}{ccccccc}
\toprule
Sample & \multicolumn{3}{c}{Linear} & \multicolumn{3}{c}{Nonlinear}\cr\cmidrule(lr){2-4}\cmidrule(lr){5-7}
size & Our method & 2SLS-L & PLS & Our method & 2SLS-L & PLS\\
\midrule
 100 & 1.23 (0.68)& 2.05 (0.86) & 0.91 (0.34) & 2.69 (1.24)& 1.08 (0.45) & 1.14 (0.36)\\
 \addlinespace[0.7mm]
200 & 0.74 (0.60)& 1.00 (0.50) &0.57 (0.21) & 1.11 (0.64)& 1.11 (0.63) & 0.83 (0.29)\\
\addlinespace[0.7mm]
400 &0.39 (0.16) & 0.36 (0.15)& 0.51 (0.18) & 0.57 (0.26) & 1.62 (0.98)& 0.86 (0.28)\\
\addlinespace[0.7mm]
600& 0.35 (0.15)& 0.29 (0.11)& 0.55 (0.21) & 0.43 (0.17)& 1.75 (1.04)& 1.00 (0.22)\\
\addlinespace[0.7mm]
800& 0.27 (0.13)& 0.23 (0.11)& 0.56 (0.19) & 0.37 (0.18)& 2.81 (1.50)& 1.14 (0.24)\\
\addlinespace[0.7mm]
1000& 0.24 (0.11) & 0.23 (0.09)& 0.56 (0.19)& 0.31 (0.12)& 3.20 (2.11)& 1.19 (0.22)\\
\addlinespace[0.7mm]
1200& 0.20 (0.08)& 0.20 (0.08)& 0.84 (0.37)& 0.23 (0.09)& 3.47 (1.90)& 1.29 (0.22)\\
\addlinespace[0.7mm]
1400& 0.19 (0.07)& 0.18 (0.08)& 0.91 (0.41)& 0.22 (0.10)& 4.12 (2.98)& 1.36 (0.20)\\
\bottomrule
\end{tabular}
\end{table}

\newpage

\vskip 0.2in
\bibliography{sample}

\end{document}